\documentclass[12pt]{article}
\usepackage{geometry}
\usepackage{amsgen,amsthm,amsmath,amstext,amsbsy,amsopn,amssymb,latexsym,appendix}
\usepackage{algorithm}
\usepackage{algpseudocode}
\usepackage[usenames]{color}
\usepackage{array}
\usepackage{multirow}
\usepackage{multicol}
\usepackage{booktabs} 
\usepackage{diagbox}
\usepackage{graphicx}
\usepackage{hyperref}
\usepackage{listings}
\usepackage{subcaption}
\usepackage{float}
\usepackage{indentfirst}
\usepackage[authoryear,round]{natbib}
\usepackage{setspace}

\definecolor{purple}{rgb}{0.84, 0.17, 0.89}
\definecolor{red2}{rgb}{0.7, 0, 0.1}
\definecolor{blue2}{rgb}{0.1,0.1,0.65}
\usepackage{color}

\def\uk{\supk}
\def\Lkh{\widehat{L}_k}
\def\summ{\sum_{k=1}^m}

\def\estatt{e_{\operatorname{stat},\Theta}}
\def\estatw{e_{\operatorname{stat},w}}

\def\etat{\eta_{\Theta}}
\def\etaw{\eta_w}
\def\QQ{\mathbb{Q}}
\def\rank{\operatorname{rank}}
\def\Tsl{\Theta^*_{\lambda}}
\def\wsl{w^*_{\lambda}}

\def\supk{^{\scriptscriptstyle (k)}}
\def\spacingset#1{\renewcommand{\baselinestretch}%
{#1}\small\normalsize} \spacingset{1}

\title{Prior-Informed Generalizable Low-Rank Representation Learning via Regularized Adversarial Reweighting}
\author{Keyao Zhan$^{}$, Junwei Lu$^{}$\footnotemark[1]\ , Tianxi Cai$^{}$\footnotemark[1]\\[1ex]
$^{}$Department of Biostatistics, Harvard University\\
}
\date{\today}

\usepackage{my_tex}

\begin{document}

\maketitle

\footnotetext[1]{: These authors jointly supervised this work.}
\begin{abstract}

Modern learning problems increasingly leverage data from heterogeneous sources to improve robustness and generalizability under distribution shift, often with high-dimensional parameters admitting low-rank structure. Existing group distributionally robust optimization (group-DRO) approaches optimize against adverse source mixtures, but extending this principle to low-rank matrix estimation with a broad class of losses leads to a challenging nonconvex-concave problem for which methods and theory remain limited. We propose \textbf{P}rior-\textbf{I}nformed generalizable \textbf{L}ow-rank representation \textbf{L}earning via \textbf{A}dversarial \textbf{R}eweighting (\textbf{PILLAR}), which smooths the group-DRO objective by regularizing adversarial source weights toward a user-specified reference distribution while learning a shared low-rank representation. By tempering extreme adversarial reweighting, PILLAR improves estimation stability and reduces the potential conservativeness of group-DRO. Using Burer-Monteiro factorization, we optimize the nonconvex-concave objective by bifactored gradient descent with mirror ascent. We establish linear convergence to a statistical neighborhood of the population saddle point and derive an error bound that separates low-rank estimation from adversarial reweighting. We further show that PILLAR improves on vanilla group DRO when the latter's excess conservativeness dominates source heterogeneity and prior misspecification. Simulations and an electronic health record study of suicide-attempt risk demonstrate improved worst-source and transported predictive accuracy.

\vspace{.5em}

\end{abstract}

\section{Introduction}
\label{sec:intro}


Modern statistical learning problems are increasingly characterized by two structural features: data are collected from multiple heterogeneous sources, and the parameter of interest is high-dimensional yet often intrinsically low-dimensional. In multi-center clinical studies, multi-institutional electronic health records, and distributed learning systems, the training data arise from several source distributions, while the distribution encountered at deployment is often unknown and differs from every observed source. At the same time, the signal of interest frequently admits a low-rank matrix representation, as in multi-task regression, matrix denoising and completion, latent factor models, and representation learning.

Standard empirical risk minimization (ERM) implicitly pools information across all observed sources equally, and this strategy can be fragile: a predictor that performs well on the pooled sample may rely heavily on dominant sources and generalize poorly to underrepresented or shifted ones. Distributionally robust optimization (DRO) addresses this issue by optimizing performance over a collection of plausible test distributions rather than under a single empirical distribution \citep{rahimian2022frameworks,kuhn2025distributionally}.
When the data come explicitly from multiple sources, group DRO builds the uncertainty set as an unknown mixture of the observed source distributions, and the learner minimizes the worst-case risk over the admissible simplex weights assigned to those sources \citep{hashimoto2018fairness,sagawa2019distributionally,mohri2019agnostic,guo2024statistical}. Such a formulation is especially natural in biomedical and federated settings, where the future test population may have a different source composition from the training sample. However, existing group DRO methods often achieve robustness at the cost of being overly conservative for future target populations of interest, because they guard against the entire simplex of possible source mixtures , which can be unnecessarily large when the sources are heterogeneous yet some prior information about the target population of interest is available. 

In parallel, low-rank modeling has become a fundamental approach for high-dimensional statistical estimation. Low-rank structure captures the idea that many high-dimensional objects are governed by a small number of latent factors, shared features, or common directions of variation. Convex approaches based on nuclear-norm and spectral regularization provide a statistically principled relaxation of rank constraints, and nonconvex factorized formulations further improve scalability by parameterizing the matrix directly through low-dimensional factors \citep{argyriou2008convex,koltchinskii2011nuclear,negahban2012restricted,tu2016low,ma2018implicit,chi2019nonconvex}. Despite these advances, the bulk of the low-rank literature is built around pooled or source-averaged objectives. It therefore does not directly answer the question of how to learn a low-rank representation whose performance is stable and generalizable across populations. Existing group DRO theory is largely developed outside the low-rank setting. Combining adversarial source reweighting with low-rank matrix estimation for a broad class of losses leads to a challenging nonconvex--concave problem, for which both optimization methods and statistical theory remain limited. 

In this paper, we propose \emph{Prior-Informed generalizable Low-rank representation Learning via Adversarial Reweighting} (PILLAR). Given heterogeneous multi-source data, PILLAR models the future test distribution as an uncertain source mixture and smooths the group DRO objective by regularizing the adversarial mixture weights toward a user-specified reference distribution while learning a shared low rank representation. The regularization can be viewed as imposing a prior on the source weights, with the tuning parameter controlling how strongly information from this prior is borrowed. By tempering extreme adversarial reweighting, PILLAR improves estimation stability while reducing the potential over-conservativeness of traditional maximin-type group DRO.  This adversarial reweighting connects naturally to the group DRO literature \citep{sagawa2019distributionally}, while the prior-informed regularization also builds a bridge between group DRO and $f$-divergence DRO \citep{duchi2021statistics}.

PILLAR solves the induced nonconvex--concave problem by alternating gradient descent on the low-rank factors and mirror ascent on the adversarial weights. The primal update exploits the factorized matrix structure with a balancing regularization, and the dual update follows the geometry of the simplex induced by the prior-guided penalty. 
We establish that the proposed algorithm converges linearly to a statistical neighborhood of the population saddle point and achieves a sharp Frobenius-norm estimation rate that explicitly separates the low-rank estimation error from the additional cost of adversarial source reweighting. We further show that PILLAR improves upon vanilla group DRO when the excess conservativeness of guarding against the full source simplex dominates the effects of source heterogeneity and prior mis-specification. To our knowledge, these results provide the first rigorous optimization and statistical theory for prior-informed generalizable low-rank representation learning.

\subsection{Our contributions}

Our contributions are fourfold. First, we introduce a general framework for \emph{group-robust low-rank representation learning}. PILLAR learns the low-rank representation itself under uncertainty about the target source composition, seeking latent factors that are stable across plausible deployment mixtures and hence generalizable under future distribution shift. This perspective is especially relevant in multi-source biomedical, federated, and multi-task settings, where the scientific target is often a shared low-dimensional structure whose generalizability across populations matters as much as its estimation accuracy. The proposed framework accommodates a wide range of low-rank problems and extends group-robust learning beyond the low-dimensional parametric models and task-specific losses most commonly analyzed in the existing DRO literature. PILLAR solves a nonconvex--concave estimation problem that couples low-rank structure with adversarial source reweighting, a setting not covered by existing group DRO or low-rank theory.

Second, we introduce a \emph{prior-informed uncertainty set} for source reweighting, establish its theoretical advantages and clarify its connection to $f$-divergence DRO. Rather than allowing the adversary to range freely over the simplex, as in group DRO, PILLAR regularizes source weights toward a reference mixture that encodes prior knowledge about plausible deployment populations, with the regularization level controlling how strongly that prior knowledge is trusted. This yields a principled interpolation between prior-weighted estimation and fully adversarial group reweighting. By tempering extreme adversarial weights, the regularization improves stability and can reduce the excess conservativeness of vanilla group DRO. We formally characterize when this improvement occurs, balancing the conservativeness of vanilla group DRO against source heterogeneity and possible prior mis-specification. We further establish a precise connection between prior-informed group DRO and $f$-divergence DRO, showing that source-level adversarial reweighting can be viewed as a finite-mixture analogue of divergence-based distributional robustness and clarifying how the choice of divergence and regularization jointly determines the effective uncertainty set.

Third, we develop a rigorous optimization analysis for the proposed nonconvex--concave formulation. The interaction between the low-rank factor space and the adversarial weight simplex creates an entangled and delicate geometry. To handle this structure, we propose a bifactored gradient descent--mirror ascent algorithm (BGDMA), which alternates between gradient updates of the low-rank factors and mirror-ascent updates of the source weights. Under regularity conditions and a suitable initialization, we prove that BGDMA converges linearly to a statistical neighborhood of the population saddle point. Existing nonconvex low-rank optimization results apply only to the minimization of pooled objectives, and minimax problems that combine a nonconvex low-rank parameterization with a simplex-constrained adversary have not previously been studied.

Fourth, we establish a sharp and interpretable statistical error rate for nonconvex-concave estimation. Our rate makes explicit how robustness depends on source imbalance, the strength of the prior-guided regularization, and the number of sources. The bound separates the intrinsic low-rank estimation error from the additional error introduced by adversarial source reweighting. In particular, the number of sources enters only through a logarithmic term, which is substantially sharper than the polynomial or linear dependence that often appears in uniform worst-group analyses. Moreover, we make the connection to \(f\)-divergence DRO precise: when the sources degenerate to i.i.d.\ observations, the PILLAR criterion coincides with the Lagrangian form of empirical $f$-divergence DRO, and we prove that the penalty scaling $\lambda_n\asymp\sqrt n$ recovers exactly the statistically calibrated regime of that literature with an $f$-divergence ambiguity radius of order $1/n$. To our knowledge, this precise rate correspondence between prior-informed group DRO and $f$-divergence DRO has not been established previously, even for low-dimensional parametric models.

Comprehensive simulations and an application to suicide-attempt risk modeling with electronic health records (EHR) data demonstrate the practical benefit of PILLAR under source heterogeneity and distribution shift compared to existing methods, including improved worst-source and transported predictive accuracy.

\subsection{Related work}

\noindent\textbf{Low-rank estimation.}
Low-rank regularization is a central tool for controlling model complexity in matrix-valued estimation problems. Convex approaches based on the nuclear norm and related spectral penalties provide statistically principled relaxations of rank constraints, with sharp guarantees in many settings\citep{argyriou2007spectral,argyriou2008convex,koltchinskii2011nuclear,negahban2012restricted}. These methods yield globally solvable convex programs but can be computationally expensive in large-scale settings. To improve scalability, a large literature has developed nonconvex factorized formulations for matrix sensing, matrix completion, reduced-rank regression, blind deconvolution, and related problems \citep{burer2003nonlinear,chi2019nonconvex,ma2018implicit}. 
Several works also show that balancing penalties on the factors are essential for obtaining a stable geometry and a clean translation from factor error to matrix error \citep{tu2016low,yu2020simultaneous}. These methods, however, are largely developed for pooled or source-averaged objectives. Our setting instead couples the low-rank factors with adversarial source-weights, leading to a nonconvex--concave problem that is not covered by existing optimization theory.

\noindent\textbf{Distributionally robust optimization.}
DRO provides a principled framework for learning under uncertainty in the data-generating distribution \citep{rahimian2022frameworks,kuhn2025distributionally}. Two especially influential ambiguity-set constructions are \(f\)-divergence balls and Wasserstein balls \citep{duchi2019variance,duchi2021statistics,mohajerin2018data,kuhn2019wasserstein}. In the finite-sample statistical literature, \(f\)-divergence DRO has been linked to variance regularization and to uniform control over latent subpopulations \citep{duchi2019variance,duchi2021learning}. In parallel, group DRO has emerged as a practical way to improve worst-group performance and to mitigate the failure of empirical risk minimization on minority or shifted subpopulations \citep{hashimoto2018fairness,sagawa2019distributionally}. Recent work has  developed optimization guarantees for group DRO \citep{zhang2023stochastic}, but primarily for convex parameter spaces and optimization complexity, not for  structured high-dimensional estimation with a low-rank matrix parameter. Related work also incorporated prior information or heterogeneous source structure into  DRO formulations \citep{rychener2024wasserstein,cao2026robust}, but does not study  adversarial reweighting together with low-rank representation learning. PILLAR addresses this gap and futher gives a direct connection between prior-informed group DRO and $f$-divergence DRO.

\noindent\textbf{Multi-source data integration and federated learning.}
The idea that the test distribution may be an unknown mixture of source or client distributions also appears in data integration and federated learning. Agnostic federated learning, for example, optimizes performance for any target distribution formed by a mixture of client distributions \citep{mohri2019agnostic}. Federated multi-task learning further exploits task relatedness across clients, and several recommender-system papers adopt matrix-factorization structures in distributed settings \citep{smith2017federated,du2021federated,duan2023adaptive}. These approaches recognize both heterogeneity and structure sharing, but they do not combine a prior-informed group-DRO objective with low-rank matrix estimation or  provide corresponding local linear-convergence and sharp matrix-estimation guarantees.

\subsection{Organization of the paper}

The rest of the paper is organized as follows. Section~\ref{sec:problem_formulation} introduces the general problem setup, defines the PILLAR model, and then develops PILLAR estimation and presents the BGDMA algorithm. Rigorous theoretical guarantees for the PILLAR estimator are given in Section~\ref{sec:theory}. Sections~\ref{sec:sim} and~\ref{sec:real} investigate the numerical performance of PILLAR through comprehensive simulation studies and a real-world application. The proofs of the main theorems, together with the key technical tools on which they rely, are given in the supplement.

\section{PILLAR Methodology}
\label{sec:problem_formulation}
\subsection{Multi-source low-rank learning}
We study a multi-source learning problem with data come from $m$ related but heterogeneous source distributions. For source $k\in[m]:=\{1,\dots,m\}$, we observe i.i.d. samples $\{X_i\supk\}_{i=1}^{n_k}$ from $\PP\supk$, where $X_i\supk$ denote the data for subject $i$ in source $k$ and may include covariates and responses in supervised learning or structured measurements in matrix problems. Sources may correspond to different institutions, populations, locations, or time periods.

Our parameter of interest is a matrix $\Theta\in\RR^{d_1\times d_2}$, representing, for example, a coefficient matrix in multi-task regression, a signal matrix in matrix recovery, or a shared representation across related learning problems. We allow $d_1$ and $d_2$ to be large, and assume that $\Theta$ has low intrisic rank belonging to 
\begin{equation}
    \Xi=\Xi(r):=\left\{\Theta\in\RR^{d_1\times d_2}\,\middle|\,\operatorname{rank}(\Theta)\le r\right\},
\end{equation}
For each source $k$, let $\ell_k(\Theta,x):\RR^{d_1 \times d_2} \times\mathcal X\to\RR$ be the loss of $\Theta$ on an observation $x$ from source $k$. The loss function may be common across groups or source-specific, depending on the application. We define the corresponding population and empirical risks by
\begin{equation}
    L_k(\Theta):=\EE_{X\sim \PP\supk}\bigl[\ell_k(\Theta,X)\bigr], \qquad
    \widehat L_k(\Theta):=\frac{1}{n_k}\sum_{i=1}^{n_k}\ell_k(\Theta,X_i\supk).
\end{equation}
We also write
\( L(\Theta):=\bigl(L_1(\Theta),\dots,L_m(\Theta)\bigr)\in\RR^m,
    \widehat L(\Theta):=\bigl(\widehat L_1(\Theta),\dots,\widehat L_m(\Theta)\bigr)\in\RR^m.
\)
This formulation
covers a broad class of supervised, unsupervised, and matrix recovery problems with heterogeneous source-specific risks and a shared low-rank parameter.

\begin{remark} 
PILLAR is not restricted to a particular loss. For example, exponential-family losses yield reduced-rank generalized linear models, while trace-regression losses cover matrix regression and binary, count, or Gaussian matrix completion. 
Thus, the framework extends beyond the squared-error or cross-entropy loss \citep{mohri2019agnostic,sagawa2019distributionally,guo2024statistical} to a broad class of multi-source low-rank learning problems.
\end{remark}

\subsection{Prior-informed adversarial reweighting}\label{sec:method:prior}
Because the deployment population may not have the same source composition from the training data, we model its distribution as an unknown mixture of the observed source distributions. Let $w\in \Delta_m$ denote the mixture weights, where
$$\textstyle 
\Delta_m:=\left\{w\in\RR^m: \min_{k \in [m]} w_k\ge 0,\sum_{k=1}^m w_k=1\right\} .
$$
For a fixed $\Theta$, $\sum_{k=1}^m w_k L_k(\Theta)$ is the risk under the mixture distribution $\sum_{k=1}^m w_k \PP\supk$. Standard group DRO protects against the worst mixture of the entire simplex, but this can be overly conservative when extreme mixtures are unlikely or prior information about the deployment population is available. We therefore regularize the adversarial weights toward a reference mixture. Let $D:\Delta_m\to\RR$ be a convex regularizer which encodes prior information, $\lambda>0$ a tuning parameter, and define the regularized worst-mixture risk
\begin{equation}
\label{eq:Phi_lambda}
    \Phi_\lambda(\Theta)
    :=\max_{w\in\Delta_m}
    \Bigl\{\summ w_k L_k(\Theta)-\lambda D(w)\Bigr\}.
\end{equation}

\begin{definition}[PILLAR target and PILLAR saddle point]
\label{def:pillar_target}
The \emph{PILLAR target} at regularization level $\lambda$ is the minimizer of $\Phi_\lambda$ over the low-rank class, and the \emph{PILLAR weight} is the mixture that attains the inner maximum at that minimizer:
\begin{equation}
\label{eq:PILLAR_def}
    \Tsl := \argmin{\Theta\in\Xi(r)}\;\Phi_\lambda(\Theta),
    \qquad
    \wsl := \argmax{w\in\Delta_m}\Bigl\{\summ w_k L_k(\Tsl)-\lambda D(w)\Bigr\}.
\end{equation}
We call the pair $(\Tsl,\wsl)$ the \emph{PILLAR saddle point}.
\end{definition}

The inner maximizer in \eqref{eq:PILLAR_def} is unique whenever $D$ is strictly convex on $\Delta_m$, which the KL divergence \eqref{eq:KL_example} is and which Assumption~\ref{ass:w_regularizer_mirror} imposes in general. Both components of the saddle point depend on $\lambda$, and much of what follows concerns how they move with it, so we carry the subscript whenever that dependence is at issue. When $\lambda$ is held fixed, as in the theory of Section~\ref{sec:theory} and its proofs, we drop it and write $(\Theta^*,w^*)$ for the same pair. The low-rank constraint gives this robust objective a concrete structural interpretation: $\Tsl\in \Xi(r)$ forces the model to operate through at most $r$ latent directions shared across sources. In supervised settings, these directions represent common predictive factors, while in matrix recovery and representation-learning they capture shared low-dimensional structure. Combined with adversarial weighting, PILLAR favor directions that remain useful across heterogeneous source mixtures while downweighting source-specific effects that do not transport well.



The regularizer $D(w)$ determines how departures from the reference mixture are penalized. A natural choice is the Kullback–Leibler (KL) divergence from a reference weight $w^{\mathrm{ref}}\in\operatorname{int}(\Delta_m)$ that encodes the prior,
\begin{equation}
        D(w)=D_{\mathrm{KL}}(w\|w^{\mathrm{ref}})
    =\sum_{k=1}^m w_k\log\!\left(\frac{w_k}{w_k^{\mathrm{ref}}}\right)
    \label{eq:KL_example}
\end{equation}
measures the discrepancy between a candidate test mixture $w$ and a reference $w^{\mathrm{ref}}$. The tuning parameter $\lambda$ governs a robustness--efficiency tradeoff. As $\lambda\to 0$,  \eqref{eq:PILLAR_def} approaches group DRO. As $\lambda\to\infty$, the adversarial weights are forced toward $w^\mathrm{ref}$, recovering a prior-weighted pooling. Intermediate values interpolate between these two regimes, allowing robustness to plausible source shifts without guarding equally against all mixtures. 

\begin{figure}[htbp]
    \centering
    \includegraphics[width=0.95\linewidth]{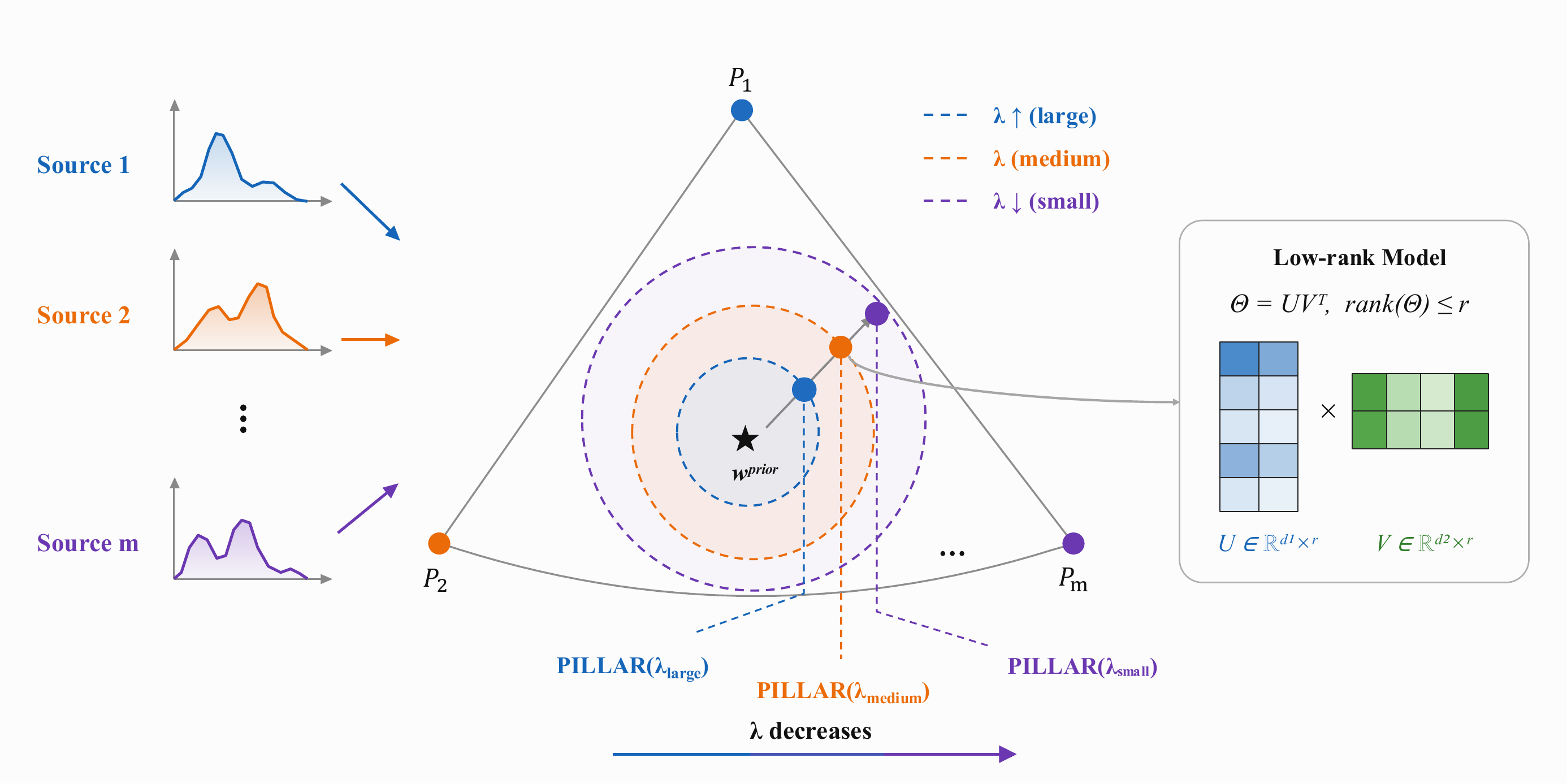}
    \caption{Schematic illustration of the PILLAR method. We collect data from $m$ heterogeneous sources, and estimate a shared low-rank parameterization $\Tsl$, with $\mathrm{rank}\leq r$. PILLAR estimates source weights  over the simplex of source distributions. The tuning parameter $\lambda$ controls regularization toward the prespecified prior weight $w^{\mathrm{prior}}$: a larger $\lambda$ yields stronger shrinkage toward the prior, whereas a smaller $\lambda$ permits greater data-driven departure toward the simplex boundary.}
    \label{fig:pillar_illustration}
\end{figure}


PILLAR also connects group DRO with $f$-divergence. When $D$ is an $f$-divergence between $w$ and $w^{\mathrm{ref}}$, the penalty defines a source-level analogue of an $f$-divergence uncertainty set. In the special case where each ``source" consists of one i.i.d observation, PILLAR reduces to the the Lagrangian form of empirical $f$-divergence DRO. We show in Section~\ref{sec:theory} this connection at the level of both the criterion and the statistical rate.

\begin{remark}
The reference mixture $w^{\mathrm{ref}}$ represents prior information about the likely deployment population. Natural defaults include uniform $w^{\mathrm{ref}}=(1/m)\mathbf 1$ and the empirical source proportions. When scientific knowledge or  an external target population are available, $w^{\mathrm{ref}}$ can instead reflect a particular population or be estimated using calibration or transportability methods \citep{westreich2017transportability,zhan2024domain,xiong2026guided}.  PILLAR then protects against residual uncertainty around this reference mixture.
\end{remark}



\subsection{Factorized Estimation and BGDMA Algorithm}
\label{sec:proposed_method}



We next estimate the PILLAR target from the data. The empirical problem combines a nonconvex low rank constraint with simplex-constrained adversarial weights. We handle these two structures using the Burer--Monteiro factorization \citep{burer2003nonlinear,burer2005local} for the low-rank parameter and miror ascent for the source weights. We write $\Theta = UV^\top$ with $U\in\RR^{d_1\times r}$ and $V\in\RR^{d_2\times r}$. This parametrization automatically enforces $\mbox{rank}(\Theta) \le r$ and reduces the optimization dimension from $d_1d_2$ to $(d_1+d_2)r$. The factoraization is not unique up to a common orthogonal matrix $O\in\mathbb O(r)$ but does not affect the estimated matrix $UV^\top = (UO)(VO)^\top$. 

Following the nonconvex low-rank literature \citep{tu2016low,park2018finding,zhu2021global}, we use the balancing penalty $$p(U,V):=\frac{1}{4}\|U^\top U - V^\top V\|_F^2.$$
This term vanishes when $U^\top U = V^\top V$, in which case $U$ and $V$ have matched singular values. It removes scale imbalance between the two factors and helps keep the concatenated factor $[U;V]\in\RR^{(d_1+d_2)\times r}$ well conditioned, which in turn allows the optimization error in $(U,V)$ to be translated into estimation error for the matrix product $\Theta=UV^\top$.

Replacing the population risks by their empirical counterparts, and writing $\widehat L_k(U,V):=\widehat L_k(UV^\top)$, we estimate the PILLAR target through the empirical saddle-point problem
\begin{equation}
\label{eq:empirical_minmax}
\min_{U,V}\;\max_{w\in\Delta_m}\widehat f(U,V,w) := \sum_{k=1}^m w_k \widehat L_k(U,V) + p(U,V) - \lambda D(w).
\end{equation}
The minimization over $(U,V)$ enforces the low-rank structure through the factorization $\Theta=UV^\top$, while the maximization over $w$ searches for adverse source mixtures subject to regularization toward the prior. The two sets of variables have different geometries: $(U,V)$ are unconstrained Euclidean variables, whereas $w$ live on the probability simplex. We therefore combine low-rank nonconvex optimization with a simplex-respecting adversarial update, leading to the BGDMA algorithm.

Given the current iterate $(U^t,V^t,w^t)$, we update the factors by
\begin{align}
\label{eq:UV_update}
    U^{t+1}
    &= U^t - \eta_\Theta \nabla_U \widehat f(U^t,V^t,w^t), \quad V^{t+1}
= V^t - \eta_\Theta \nabla_V \widehat f(U^t,V^t,w^t),
\end{align}
where $\eta_\Theta>0$ is the step size. These updates jointly adapt the shared low-rank representation to the current adversarial weighting of the source-specific losses. If desired, additional structure such as norm or sparsity constraints can be incorporated through projection or a proximal modification \citep{yu2020simultaneous}. For $w$, we use mirror ascent to preserve the simplex geometry. Given a mirror map $\phi:\Delta_m\to\RR$, we define the associated Bregman divergence
\[
    B_\phi(w,\widetilde w)
    := \phi(w)-\phi(\widetilde w)-\langle \nabla\phi(\widetilde w),\, w-\widetilde w\rangle.
\]
The weight update is
\begin{equation}
\label{eq:mirror_ascent_general}
    w^{t+1}
    \in \argmax{w\in\Delta_m}
    \left\{
        \left\langle \nabla_w \widehat f(U^t,V^t,w^t),\, w-w^t \right\rangle
        - \frac{1}{\eta_w} B_\phi(w,w^t)
    \right\},
\end{equation}
with \(
    \nabla_w \widehat f(U^t,V^t,w^t)
    = \widehat L(U^t,V^t) - \lambda \nabla D(w^t),
\), where $\eta_w>0$ is the mirror-ascent step size. The weight update combines two forces: it increases the weights of groups with larger current empirical loss, while simultaneously shrinking the adversary back toward the prior mixture through the regularization term $\lambda D(w)$.

In this general form, the factor updates \eqref{eq:UV_update} apply to any differentiable loss, while the mirror update \eqref{eq:mirror_ascent_general} accommodates any convex $D$ with a compatible mirror map $\phi$. Thus, BGDMA is not tied to a specific loss or prior. A particularly natural and computationally convenient choice is KL regularization. Take the negative entropy mirror map,
$$
    \phi(w)=\sum_{k=1}^m w_k \log w_k, \quad 
    \mbox{for which } B_\phi(w,\widetilde w)=D_{\mathrm{KL}}(w\|\widetilde w).
$$
If we further take 
\(
    D(w)=D_{\mathrm{KL}}(w\|w^{\mathrm{ref}})
,\)
the mirror-ascent step \eqref{eq:mirror_ascent_general} reduces to the classical \textit{exponentiated gradient ascent} update:
\begin{equation}
\label{eq:w_update_expgrad_shrinkage}
    w_k^{t+1}
    \propto
    (w_k^t)^{1-\eta_w\lambda}
    (w_k^{\mathrm{ref}})^{\eta_w\lambda}
    \exp\!\bigl\{\eta_w \widehat L_k(U^t,V^t)\bigr\},
    \qquad k\in[m],
\end{equation}
followed by normalization so that $\sum_{k=1}^m w_k^{t+1}=1$. This form makes the role of the penalty especially transparent: groups with larger current loss receive more mass through the exponential factor, while the prior mixture $w^{\mathrm{ref}}$ pulls the weights back toward a scientifically plausible baseline. When $\lambda=0$, \eqref{eq:w_update_expgrad_shrinkage} reduces to the classical exponentiated-gradient ascent for group DRO; as $\lambda$ increases, the prior has greater influence on the adversarial weights. Although the KL case admits especially simple updates, the algorithm and the theory are developed for the broader class of strongly convex penalties and mirror maps.

The bifactored gradient descent--mirror ascent algorithm is inherently local, so a good initialization $(\Theta^0,w^0)$ is important. A general approach is to first obtain a convex low-rank estimate
\begin{equation}
\label{eq:convex_init}
    \Theta^0 \in \argmin{\Theta\in\RR^{d_1\times d_2}}
    \left\{
        \sum_{k=1}^m \bar w_k \widehat L_k(\Theta) + \operatorname{pen}(\Theta)
    \right\},
\end{equation}
where $w^0=\bar w$ is a fixed reference weight such as the uniform mixture, and $\operatorname{pen}(\Theta)$ is a convex low-rank regularizer such as the nuclear norm. This provides a stable warm start before switching to the more scalable nonconvex factoraized optimization. In practice, \eqref{eq:convex_init} need only be solved approximately, for example using a small number of proximal gradient iterations. For the KL-regularized PILLAR objective, the  inner maximization over $w$ admits the closed form
\begin{equation}
\label{eq:kl_profile_identity}
    \max_{w\in\Delta_m}
    \left\{
        \sum_{k=1}^m w_k \Lkh(\Theta)
        -
        \lambda
        \sum_{k=1}^m w_k\log\frac{w_k}{w_k^{\mathrm{ref}}}
    \right\}
    =
    \lambda
    \log\left(
        \sum_{k=1}^m w_k^{\mathrm{ref}} e^{\Lkh(\Theta)/\lambda}
    \right).
\end{equation}
Thus, after profiling out $w$, the robust multi-source objective becomes a smooth log-sum-exp aggregation of the source-specific losses. This leads to a convex nuclear-norm initializer
\begin{equation} \label{eq:softmax_init} \Theta^0 \in \argmin{\Theta\in\RR^{d_1\times d_2}} \left\{ \lambda
    \log\left(
        \sum_{k=1}^m w_k^{\mathrm{ref}} e^{\Lkh(\Theta)/\lambda}
    \right) + \tau_n\|\Theta\|_* \right\}, \end{equation}
and then initializing the factors by the balanced SVD of $\Theta^0$. Unlike a pooled initializer with fixed source weights, this construction already reflects the prior-informed adversarial objective that BGDMA will subsequently optimize. The initial adversarial weight is chosen to be the corresponding softmax weight
\begin{equation}
\label{eq:softmax_weight_init}
    w_k^0
    =
    \frac{
        w_k^{\mathrm{ref}}
        \exp\left\{\widehat L_k(\Theta^0)/\lambda\right\}
    }{
        \sum_{j=1}^m
        w_j^{\mathrm{ref}}
        \exp\left\{\widehat L_j(\Theta^0)/\lambda\right\}
    },
    \qquad k\in[m],
\end{equation} 
so sources with larger initial losses receive greater weight, while $w^{\mathrm{ref}}$ continues to anchor the adversary.  Given $\Theta^0$, we compute the rank-$r$ truncated SVD $\Theta^0 = \tilde U \tilde\Sigma
\tilde V^\top$ and set $U^0 = \tilde U \tilde\Sigma^{1/2}$ and $V^0 = \tilde V \tilde\Sigma^{1/2}$ for the bifactored gradient descent step. We show in Section~\ref{subsec:init} that, under suitable conditions, this initialization $\Theta^0$ lies within a constant distance of $\Tsl$.

\begin{algorithm}
\caption{Bifactored Gradient Descent Mirror Ascent (BGDMA)}
\label{alg:BGDMA}
\begin{algorithmic}[1]
\Require Empirical losses $\{\widehat L_k(\Theta)\}_{k=1}^m$, rank $r$, regularization
parameter $\lambda>0$, step sizes $\eta_\Theta,\eta_w>0$, number of iterations
$T$, initialization $\Theta^0,w^0$.
\State Compute a rank-$r$ SVD $\Theta^0 = \tilde U \tilde\Sigma \tilde V^\top$
and set $U^0 = \tilde U \tilde\Sigma^{1/2}$, $V^0 = \tilde V \tilde\Sigma^{1/2}$.
\For{$t=0,1,\dots,T-1$}
  \State Compute gradients $\nabla_U \widehat f(U^t,V^t,w^t)$ and $\nabla_V \widehat f(U^t,V^t,w^t)$.
  \State Gradient descent step on $(U,V)$:
  \Statex
  \begin{equation*}
    U^{t+1} = U^t - \eta_\Theta\, \nabla_U \widehat f(U^t,V^t,w^t),\qquad
    V^{t+1} = V^t - \eta_\Theta\, \nabla_V \widehat f(U^t,V^t,w^t).
  \end{equation*}
  \State Compute $\nabla_w \widehat f(U^t,V^t,w^t) = \widehat L(U^t,V^t) - \lambda \nabla D(w^t)$.
  \State Mirror ascent step on $w$:
  \Statex
  \begin{equation*}
    w^{t+1} \in \argmax{w\in\Delta_m}
    \left\{
        \bigl\langle \nabla_w \widehat f(U^t,V^t,w^t),\, w - w^t \bigr\rangle
        - \frac{1}{\eta_w} B_{\phi}(w,w^t)
    \right\}.
  \end{equation*}
  \State (e.g., exponentiated gradient update~\eqref{eq:w_update_expgrad_shrinkage} when $D$ and $\phi$ are entropy-based).
\EndFor
\State \textbf{Output:} $\widehat\Theta^T := U^T (V^T)^\top,\ \widehat w^T$.
\end{algorithmic}
\end{algorithm}

The complete BGDMA algorithm is summarized in Algorithm~\ref{alg:BGDMA}. Each iteration requires computing the source specific loss gradients, updating the low-rank factors, and performing an $O(m)$ mirror-ascent update for the source weights. When $r$ is moderate and the per-sample loss is smooth, BGDMA operates on only $(d_1+d_2)r$ factor parameters, plus an $O((d_1+d_2)r^2)$ cost from the balancing penalty. 

This factoraization formulation can be substantially more scalable than convex nuclear-norm approaches, which optimize directly over the full $d_1\times d_2$ matrix and typically require sigular-value thresholding or partial SVDs at every step. Such updates cost at least $O(d_1d_2r)$ for a rank-r truncated SVD and require $O(d_1d_2)$ storage. Thus, when $r\ll \min(d_1,d_2)$, the BGDMA reduces both the optimization dimension and the need for repeated large-scale SVD, at the cost of replacing global convexity with a local nonconvex optimization problem.

\section{Theoretical justification}
\label{sec:theory}

This section establishes the main theoretical properties for PILLAR. We first state the regularity conditions, then derive the convergence and statistical error bound for BGDMA. We next formalize the connection to the $f$-divergence DRO, verify the proposed initalization, and quantify when prior informed robustness improves upon vanilla group DRO. 

The regularization level $\lambda$ is held fixed here and in the corresponding proofs, so, following the convention of Definition~\ref{def:pillar_target}, we abbreviate the PILLAR saddle point $(\Tsl,\wsl)$ to $(\Theta^*,w^*)$ and write $\kappa^*=\sigma_{\max}(\Theta^*)/\sigma_{\min}(\Theta^*)$ for its condition number. For a real random variable $X$, define the Orlicz-2 norm as $\|X\|_{\psi_2}:=\inf \{t>0:\ \EE\exp(X^2/t^2)\le 2 \}.$ For a random matrix $G\in\RR^{d_1\times d_2}$, define
\(
\|G\|_{\psi_2,\mathrm{F}} := \|\mathrm{vec}(G)\|_{\psi_2}
:= \sup_{\|A\|_F=1}\ \|\langle G,A\rangle\|_{\psi_2}.
\) We use relint$(\Delta_m)$ to denote the relative interior of the simplex in $\RR^m$.

\subsection{Regularity conditions}
\label{subsec:regularity}

We impose regularity assumptions on curvature and smoothness of the losses, a concentration condition on the sampling noise, a population-level well-specification condition, and a local initialization condition.

\begin{assumption}
\label{ass:rsc_rss}
For each $k \in [m]$, the gradient of the empirical loss $\Lkh (\Theta)$ is locally $\beta_L$-bounded; that is, there is a radius $\varrho>0$ such that $\|\nabla_{\Theta}\widehat{L}_k(\Theta)\|_F \leqslant \beta_L $ for all $\|\Theta-\Theta^*\|_F \leq \varrho$. Moreover, defining the set of approximately low-rank matrices
\begin{equation}
\label{eq:ratio_cone}
\mathcal A:=\bigl\{\Delta\in\RR^{d_1\times d_2}:\ \|\Delta\|_*\le 4\sqrt{2r}\,\|\Delta\|_F\bigr\},
\end{equation}
there exist constants $0<\mu < \beta_G$ such that each empirical loss $\Lkh(\Theta)$ satisfies, for all $\Theta_1,\Theta_2$ with $\Theta_2-\Theta_1\in\mathcal A$,
\begin{equation}
\label{eq:rsc_rss}
    \frac{\mu}{2}\|\Theta_2-\Theta_1\|_F^2
    \;\le\;
    \Lkh(\Theta_2) - \Lkh(\Theta_1)
    - \bigl\langle \nabla \Lkh(\Theta_1), \Theta_2 - \Theta_1 \bigr\rangle
    \;\le\;
    \frac{\beta_G}{2}\|\Theta_2-\Theta_1\|_F^2.
\end{equation}
\end{assumption}

Assumption~\ref{ass:rsc_rss} is a restricted strong convexity and smoothness condition tailored to the low-rank structure: the two-sided bound \eqref{eq:rsc_rss} is required only along differences in $\mathcal A$, the set of approximately low-rank directions familiar from nuclear-norm analyses \citep{negahban2012restricted}. This set contains the low-rank differences needed for the BGDMA analysis, while the larger cone $\mathcal A$ is required only for analyzing the convex initializer in Section~\ref{subsec:init}. 
In standard models (e.g., generalized linear models with sub-Gaussian covariates) these restricted conditions hold with high probability within a bounded parameter region. We state them for the empirical losses for  convenience; they can also be verified with high probability by combining population restricted curvature with an appropriate uniform concentration bound.

\begin{assumption}
\label{ass:w_regularizer_mirror} There exists a mirror map $\phi (\cdot)$, such that the regularizer $D(\cdot)$ is 1-convex and $\beta_D$-smooth relative to $\phi(\cdot)$, i.e. 
\begin{equation}
    B_\phi(u,w) \leq D(u)-D(w) - \langle \nabla D(w),u-w\rangle \leq \beta_DB_{\phi}(u,w), \quad \forall u,w \in \Delta_m, \label{eq:assum:relative}
\end{equation}
and $\phi(\cdot)$ is 1-strongly convex with respect to $\|\cdot \|_1$, i.e. $B_{\phi}(u,w) \geq \frac12\|u-w\|_1^2, \quad  \forall u,w \in \Delta_m.$
\end{assumption}

Assumption~\ref{ass:w_regularizer_mirror} imposes relative convexity and smoothness conditions on the regularizer $D(\cdot)$, as is standard in the mirror descent literature \citep{bauschke2017descent,lu2018relatively}. Taking  the mirror map $\phi$ and KL regularizer satisfies the condition with $\beta_D=1$. For convenience we write $\beta = \max\{\beta_G,\beta_L,\beta_D\}$.

\begin{assumption}
\label{ass:loss_psi2}
For each group $k\in[m]$ and sample $i$, the centered loss variable
\(
e_i\supk := \ell_k(\Theta^*,X_i\supk) - L_k(\Theta^*)
\) satisfies the uniform Orlicz-2 bound
\(
\|e_i\supk\|_{\psi_2}\le K_L
\) for all $i,k$. The centered gradient matrix
\(
G_i\supk := \nabla_\Theta \ell_k(\Theta^*,X_i\supk) - \nabla_\Theta L_k(\Theta^*)\in\RR^{d_1\times d_2}
\) satisfies the uniform Orlicz-2 bound
\(
\|G_i\supk\|_{\psi_2,\mathrm{F}} \le K_G
\) for all $i,k.$
\end{assumption}


\begin{assumption}[Population first-order condition]
\label{ass:pop_foc}
The PILLAR population saddle point $(\Theta^*,w^*)$ satisfies
\(
\summ w_k^*\,\nabla_\Theta L_k(\Theta^*)=0 .
\)
\end{assumption}

Assumptions~\ref{ass:loss_psi2} is a standard sub-Gaussian tail condition ensuring that the empirical saddle-point problem concentrates around its population counterpart; it holds, for example, for bounded or Lipschitz losses with sub-Gaussian observations. Assumption~\ref{ass:pop_foc} is a well-specification condition requiring the rank constraint to be inactive at the PILLAR target;
without it the rates below would acquire an additional approximation-error term.



\begin{assumption}
\label{ass:initialization}
The initialization $(\Theta^0,w^0)$ satisfies
\begin{equation}
\label{eq:init_assum}
\| \Theta^0  - \Theta^*\|_F^2 \leq  I_\Theta \sigma_{\min}^2(\Theta^*) \min\{1,\sigma_{\min}(\Theta^*)\} ,\qquad  \lambda B_{\phi}(w^*,w^0) \leq I_w \sigma_{\min}^2(\Theta^*),
\end{equation}
where $I_\Theta$ and $I_w$ are sufficiently small constants depending only on the regularity parameters of Assumptions~\ref{ass:rsc_rss}--\ref{ass:w_regularizer_mirror}; explicit choices are given in \eqref{eq:appen:ITheta_Iw} in the appendix.
\end{assumption}

Assumption~\ref{ass:initialization} is a local initialization requirement for the nonconvex factorized algorithm rather than an additional modeling assumption. Section~\ref{subsec:init} shows that it is met by the soft-max initializer of Section~\ref{sec:proposed_method}.

\subsection{Convergence and estimation error}

We now establish the main convergence and statistical guarantee for BGDMA. The result shows that the coupled factor and weight updates retain local linear convergence and quantifies the statistical error after optimization error becomes negligible.

\begin{theorem}
\label{thm:PILLAR_convergence}
Suppose Assumptions~\ref{ass:rsc_rss}--\ref{ass:initialization} hold, and let $(\Theta^*,w^*)$ be the PILLAR population saddle point associated with~\eqref{eq:PILLAR_def}. Let $\widehat{\Theta}^T$ and $\widehat w^T$ be the step-$T$ iterates of BGDMA in Algorithm~\ref{alg:BGDMA}. Let $\kappa^*=\sigma_{\max}(\Theta^*)/\sigma_{\min}(\Theta^*)$ and $n_{\min}=\min_{k\in[m]} n_k$, and define the effective sample size $n_{\mathrm{eff}}=\left(\sum_{k=1}^m {(w_k^*)^2}/{n_k}\right)^{-1}$. 

Fix $\delta\in(0,1)$. Suppose that $\lambda\ge c_1\sqrt{\kappa^*}$, that the step sizes satisfy $\eta_\Theta\le c_2/(\beta\sigma_{\max}(\Theta^*))$ and $\eta_w=\,1/\lambda\beta$, where $c_1,c_2>0$ are constants that depend only on $(\mu,\beta)$, and that the sample sizes $n_{\mathrm{eff}}$ and $n_{\min}$ are sufficiently large. Then there exist constants $C_0,C_1>0,\gamma\in(0,1)$ such that, with probability at least $1-\delta$, the recovered low-rank matrix
$\widehat\Theta^T$ satisfies
\begin{equation}
\label{eq:theta_convergence}
    \|\widehat\Theta^T - \Theta^* \|_F \le
C_0(1-\gamma)^T
+
C_1
\sqrt{
\frac{r(d_1+d_2)+\log(1/\delta)}{n_{\mathrm{eff}}}
+
\frac{\log(m/\delta)}{\lambda^2 n_{\min}}
}.
\end{equation}
Consequently, if the algorithm is run for $T\ge C\log(1/\varepsilon_n)$ steps, where
\begin{equation}
    \varepsilon_n^2
:=
\frac{r(d_1+d_2)+\log(1/\delta)}{n_{\mathrm{eff}}}
+
\frac{\log(m/\delta)}{\lambda^2 n_{\min}},
\label{eq:eps_n}
\end{equation}
then $\|\widehat\Theta^T - \Theta^* \|_F \leq C_2 \varepsilon_n.$ 
\end{theorem}

Theorem~\ref{thm:PILLAR_convergence} separates the error of BGDMA into a geometrically decaying optimization component and a statistical term. After $O(\log(1/\varepsilon_n))$ iterations, the latter dominates and has two components as shown in \eqref{eq:eps_n}. The first is the usual low-rank estimation error, governed by the effective sample size $n_{\mathrm{eff}}$. When $n_k\asymp n$ and $w^*$ is diffuse, yielding the pooled low-rank rate;  when an adversary concentrates on one source with $n_{\mathrm{eff}}\asymp n$, yielding the hardest-source rate. The rate thus adapts from full pooling to worst-source learning, with the number of sources entering only through $\log m$, sharper than the polynomial dependence in existing group DRO literature \citep{guo2025statistical,wang2026distributionally}.

The second term in~\eqref{eq:theta_convergence} is the additional cost of learning the adversarial weights. It decreases with $\lambda$, because stronger regularization stablizes the inner maximization, while $\lambda$ simultaneously moves the population target from worst-mixture robustness toward the reference mixture. The bound therefore makes the robustness--stability tradeoff induced by the PILLAR penalty explicit.

The proof, given in Appendix~\ref{appen:proof_thm1}, tracks a joint Lyapunov potential in $(U,V)$ and $w$: the two recursions are coupled, since the current weights drive the low-rank gradient while the current low-rank estimate drives the mirror step, and the condition $\lambda\ge c_1\sqrt{\kappa^*}$ lets both errors be absorbed into a single contraction. All constants, together with the precise sample-size requirement behind ``sufficiently large,'' are made explicit in Theorem~\ref{thm:appen:deterministic} and condition~\eqref{eq:appen:samplesize}.

\begin{remark}[On the condition $\lambda\ge c_1\sqrt{\kappa^*}$]
\label{rem:lambda_condition}
This lower bound is an artifact of the one-potential proof rather than a property of the estimator: it makes the feedback between the primal and dual recursions a fixed fraction of the primal contraction, and we do not expect it to be necessary. It is in any case uninformative about small $\lambda$, including the unregularized limit $\lambda\to0$ in which the target degenerates to the group DRO solution; empirically, BGDMA converges stably over the whole grid down to $\lambda=0$ (Section~\ref{sec:sim}).
\end{remark}

After establishing Theorem \ref{thm:PILLAR_convergence}, the formal connection to $f$-divergence DRO described in Section~\ref{sec:problem_formulation} can now be made exact, at the level of the criterion, the rate, and the ambiguity radius. When every source consists of a single observation drawn from a common distribution, the source simplex becomes the probability simplex over the empirical sample, and the following holds.

\subsection{Connection to $f$-divergence DRO}

We next formulalize the connection to $f$-divergence DRO. When each source consists of a single observation from a common distribution, PILLAR recovers $f$-divergence DRO at both the objective and  statistical scaling levels.

\begin{corollary}[Recovering $f$-divergence DRO]
\label{cor:fdiv}
Let $D(w)=D_f(w\|w^{\mathrm{ref}})$ for a convex $f$ with $f(1)=0$, such that Assumption~\ref{ass:w_regularizer_mirror} holds with constants free of $n$. Suppose the $m=n$ sources each contain a single observation ($n_k=1$) drawn independently from a common distribution $\PP$, with a common loss $\ell$ and $w^{\mathrm{ref}}=n^{-1}\mathbf 1$, and write $\widehat\PP_n$ for the empirical distribution. Then:
\begin{enumerate}
\item[(i)] the empirical PILLAR criterion is the Lagrangian form of empirical $f$-divergence DRO,
\begin{equation}
\label{eq:kl_dro_lagrangian}
\min_{\Theta\in\Xi(r)}\ \max_{w\in\Delta_n}\Bigl\{\sum_{k=1}^n w_k\Lkh(\Theta)-\lambda D(w)\Bigr\}
\;=\;
\min_{\Theta\in\Xi(r)}\ \sup_{\QQ\ll\widehat\PP_n}\Bigl\{\EE_{\QQ}\,\ell(\Theta,X)-\lambda\,D_f\bigl(\QQ\,\|\,\widehat\PP_n\bigr)\Bigr\},
\end{equation}
and, for every $\lambda>0$, the population saddle point is $(\Theta^*,n^{-1}\mathbf 1)$, where $\Theta^*$ minimizes $\EE_\PP\,\ell(\Theta,X)$ over $\Xi(r)$, with $n_{\mathrm{eff}}=n$;
\item[(ii)] if the assumptions of Theorem~\ref{thm:PILLAR_convergence} hold with $\lambda=\lambda_n:=\lambda_0\sqrt n$ for a fixed $\lambda_0>0$, then for all $n\ge(c_1\sqrt{\kappa^*}/\lambda_0)^2$ and $T\ge C\log n$,
\[
\bigl\|\widehat\Theta^T-\Theta^*\bigr\|_F
\;=\;O_{\PP}\!\left(\sqrt{\frac{r(d_1+d_2)+\log n}{n}}\,\right);
\]
\item[(iii)] in the setting of \textup{(ii)}, the realized ambiguity radius
$\widehat\rho_n:=D_f(\widehat w_n\,\|\,w^{\mathrm{ref}})$, where $\widehat w_n$ maximizes $\sum_{k=1}^n w_k\Lkh(\widehat\Theta^T)-\lambda_nD(w)$ over $\Delta_n$, satisfies
\[
\widehat\rho_n\;=\;O_{\PP}\!\left(\frac{\log n}{n}\right).
\]
\end{enumerate}
\end{corollary}

Corollary~\ref{cor:fdiv} strengthens the correspondence in Section~\ref{sec:problem_formulation} from an objective-level analogy to a rate-level equivalence. In constrained $f$-divergence DRO, the statistically calibrated ambiguity radius is of order $1/n$ \citep{duchi2019variance,duchi2021statistics}. The corollary shows that the penalized PILLAR formulation reaches the corresponding regime when $\lambda_n\asymp\sqrt n$. The estimator attains the parametric $\sqrt n$-rate, up to the low-rank complexity $r(d_1+d_2)$ and a logarithmic factor, while the fitted adversary realizes precisely the calibrated radius, $\widehat\rho_n=O_{\PP}(\log n/n)$. Beyond the present degenerate configuration, Theorem~\ref{thm:PILLAR_convergence} extends the principle to heterogeneous sources and low-rank structure, with $n_{\mathrm{eff}}$ playing the role of the effective number of observations retained by the tilted distribution.

\subsection{Initialization guarantee}
\label{subsec:init}

We next verify that the soft-max initializer of Section~\ref{sec:proposed_method} satisfies the local condition of Assumption~\ref{ass:initialization}. Throughout this subsection the regularizer is the KL divergence
\(
    D(w)=D_{\mathrm{KL}}(w\|w^{\mathrm{ref}})
    =
    \sum_{k=1}^m w_k\log({w_k}/{w_k^{\mathrm{ref}}}),
\) for an arbitrary reference weight $w^{\mathrm{ref}}\in\operatorname{int}(\Delta_m)$, and the initialization is given by \eqref{eq:softmax_init} and \eqref{eq:softmax_weight_init}. Because the convex program \eqref{eq:softmax_init} is posed over all of $\RR^{d_1\times d_2}$ rather than over the low-rank set, its analysis relies on Assumption~\ref{ass:rsc_rss} along the full set $\mathcal A$ of \eqref{eq:ratio_cone}, together with convexity of the empirical losses; the associated nuclear-norm geometry is standard \citep{negahban2012restricted,koltchinskii2011nuclear} and is collected in Appendix~\ref{appen:proof_thm2}.

\begin{theorem}\label{thm:softmax_init} Suppose $D(w)=D_{\mathrm{KL}}(w\|w^{\mathrm{ref}})$, each $\Lkh$ is convex on $\RR^{d_1\times d_2}$ and Assumptions~\ref{ass:rsc_rss}, \ref{ass:loss_psi2} and~\ref{ass:pop_foc} hold, with $\varrho\ge\sigma_{\min}(\Theta^*)/4$ in Assumption~\ref{ass:rsc_rss}. Let $\Theta^0$ solve \eqref{eq:softmax_init} and $w^0$ be defined in \eqref{eq:softmax_weight_init}. For any $\delta\in(0,1)$, define
\[
A_\delta
=
K_G\sqrt{d_1+d_2+\log(m/\delta)}
+
\frac{\beta K_L}{\lambda}\sqrt{\log(m/\delta)}.
\]
Choose
\(
\tau_n
=
C_\tau
{A_\delta}/{\sqrt{n_{\min}}},\)
where $C_\tau>0$ is a sufficiently large universal constant. Let $\sigma_{\min}=\sigma_{\min}(\Theta^*)$. If
\begin{equation}
    n_{\min}
\ge
C
\max\left\{
\frac{r A_\delta^2}
{\mu^2 I_\Theta\sigma_{\min}^2\min\{1,\sigma_{\min}\}},
\frac{
\left(
\beta\sqrt r\,A_\delta/\mu
+
K_L\sqrt{\log(m/\delta)}
\right)^2
}
{\lambda I_w\sigma_{\min}^2}
\right\},
\end{equation}
then, with probability at least $1-\delta$, the initialization $(\Theta^0,w^0)$ satisfies the condition \eqref{eq:init_assum} in Assumption~\ref{ass:initialization}.

\end{theorem}

Theorem~\ref{thm:softmax_init} establishes the initialization required by Theorem~\ref{thm:PILLAR_convergence}. The tuning parameter $\tau_n$ is set at the stochastic-gradient scale of the profiled loss, with the two terms of $A_\delta$ capturing the gradient noise and the propagation through the softmax weights. Although stated for KL regularization, the same practical principle extends to other  warm starts, such as a pooled nuclear-norm estimator or a continuation from a larger $\lambda$, after which BGDMA refines the estimator for the final PILLAR objective.

\subsection{Benefit of prior-informed robustness}
\label{subsec:prior_benefit_gdro}

We next quantify when regularizing the adversary toward a reference mixture improves upon vanilla group DRO. Such an improvement holds when the deployment mixtures of interest remain reasonably close to $w^{\mathrm{ref}}$, the setting in which prior information is meaningful.

Because this comparison turns on the choice of $\lambda$, we restore the subscript in this subsection and write $(\Tsl,\wsl)$ for the PILLAR saddle point of Definition~\ref{def:pillar_target}. Alongside it we define two reference models
\begin{equation}
        \Theta^*_{\mathrm{ref}}
    \in
    \argmin{\Theta\in\Xi(r)} \summ w^{\mathrm{ref}}_k L_k(\Theta);
    \qquad
    \Theta^*_{\mathrm{G}}
    \in
    \argmin{\Theta\in\Xi(r)} \max_{k \in [m]}L_k(\Theta),
    \label{eq:refGdef}
\end{equation}
which are the optimal low-rank models under the prior weights and under vanilla group DRO, respectively. We also define the source-risk heterogeneity as
\begin{equation}
      H:= \max_{k\in[m]}L_k(\Theta^*_{\mathrm{ref}})-\min_{k\in[m]}L_k(\Theta^*_{\mathrm{ref}}).
    \label{eq:source_hetero}
\end{equation}
The following result makes explicit when reduction in conservativeness achieved by prior informed regularization outweights the costs of source heterogeneity and possible prior mis-specification.

\begin{theorem}
\label{thm:prior_benefit_gdro}
Suppose Assumption~\ref{ass:w_regularizer_mirror} holds and that $D\ge0$ with $D(w^{\mathrm{ref}})=0$ for some $w^{\mathrm{ref}}\in\operatorname{relint}(\Delta_m)$. For any radius $\rho\geq 0$, define the worst-case risk and the conservativeness gap by
\[
\mathcal R_\rho(\Theta)
    :=
    \sup_{q\in\Delta_m: D(q)\leq \rho} \summ q_k L_k(\Theta),
    \qquad
    \Delta_{\mathrm{ref}}
    := \summ w^{\mathrm{ref}}_k \bigl(L_k(\Theta^*_{\mathrm{G}})-L_k(\Theta^*_{\mathrm{ref}})\bigr).
\]
Then, for every evaluation radius $\rho$, we have
\begin{equation}
    \mathcal R_\rho(\Tsl)-\mathcal R_\rho(\Theta^*_{\mathrm{G}})
    \leq
    -\Delta_{\mathrm{ref}}
    +
    \frac{H^2}{8\lambda}
    +
    \lambda \rho.
    \label{eq:worst_prior_valid_gain}
\end{equation}
In particular, if $\frac{H^2}{8\lambda}+ \lambda \rho \leq \Delta_{\mathrm{ref}}$, then PILLAR has strictly smaller worst-case evaluation risk than vanilla group DRO.
\end{theorem}

Theorem~\ref{thm:prior_benefit_gdro} compares PILLAR with vanilla group DRO over test mixtures within $D$-radius $\rho$ of the reference mixture. The gap $\Delta_{\mathrm{ref}}$ measures how much prior-weighted risk sacrificed by group DRO when it protects against the entire simplex. PILLAR recovers this excess conservativeness, up to two costs: the soft-robustness cost $H^2/(8\lambda)$ and $\lambda\rho$. The first reflects source heterogeneity and decreases as adversary is pulled toward the prior. The second reflects possible prior misspecification and increases with $\lambda$. Thus, stronger regularization is beneficial when the prior is reliable, while weaker regularization is safer when substantial deviation from the prior is possible. For a fixed $\rho$, the bound is minimized at $\lambda^*=H/\sqrt{8\rho}\asymp\rho^{-1/2}$, so the optimal regularization strength increases as the uncertainty radius shrinks. In particular, at the statistically calibrated radius $\rho\asymp 1/n$, this gives $\lambda_n\asymp\sqrt n$, consistent with Corollary~\ref{cor:fdiv}. Appendix~\ref{appen:thm3proof} provides the proof and a sharper characterization of the dependence on $\lambda$, including the detailed corresponding radius-matching property.



\section{Simulation studies}\label{sec:sim}

In this section, we study the empirical performance of PILLAR in a heterogeneous multi-task regression model. Section~\ref{subsec:sim_rate} examines the optimization and statistical error given by Theorem~\ref{thm:PILLAR_convergence}, while  Section~\ref{subsec:sim_pred} compares PILLAR with benchmark estimators evaluates predictive robustness across different deployment distributions.

\subsection{Design and competing estimators}
\label{subsec:sim_design}

We first describe the ingredients shared by every experiment. The quantities that vary between them, namely the dimensions $(d_1,d_2)$, the number of sources $m$, the per-source sample size $n_k$, the rank $r$ and the grid of $\lambda$, are stated as each experiment is introduced.

We generate data $(x_i\uk,y_i\uk)$ for $m$ sources using different distributions. For the regression parameter $\Theta_k$, the $m$ sources share a latent low-rank center $U_0V_0^\top$, with $U_0 \in \RR^{d_1 \times r^*}, V_0 \in \RR^{d_2 \times r^*}$, and each source $k$ ($k \in [m]$) deviates from it individually, so we have $\Theta_k=(U_0+\Delta^U_k)(V_0+\Delta^V_k)^\top \in \RR^{d_1 \times d_2}$, and we generate the covariates $x\uk_i \in \RR^{d_1}$ and multi-task outcome $y\uk_i \in \RR^{d_2}$ using a linear model
\begin{equation}
\label{eq:sim_dgp}
    x_i\uk\sim\mathcal N_{d_1}(\mu_k,\Sigma_k),
    \quad
    y_i\uk = \Theta_k^\top x_i\uk + \varepsilon_i\uk,
    \quad
    \varepsilon_i\uk\sim\mathcal N(0,\sigma_k^2 I_{d_2}),
\end{equation}
where $U_0,V_0$ and  $\Delta^U_k \in \RR^{d_1\times r^*},\Delta^V_k \in \RR^{d_2 \times r^*}$ have independent
centered Gaussian entries, those of the perturbations having half the standard deviation of those of the center. Heterogeneity therefore enters through the regression
operator $\Theta_k$, the covariate law $(\mu_k,\Sigma_k)$ and the noise level $\sigma_k^2$. Thus, heterogeneity extends beyond covariate shift, while the shared low-rank structure still allows information borrowing across sources.  The detailed specification of all the parameters in the model can be found in Appendix Section \ref{appen:sim_details}.

We take
$\Lkh(\Theta)=\{\|Y\uk-X\uk\Theta\|_F^2-\|Y\uk\|_F^2\}/(2n_kd_2)$ and solve
\eqref{eq:empirical_minmax} by BGDMA with the balancing penalty and KL regularization
towards $w^{\mathrm{ref}}=m^{-1}\mathbf 1$. The offset $\|Y\uk\|_F^2$ does not change the gradient in $\Theta$, but prevents the adversary from favoring sources simply because they have larger irreducible noise. Instead, it emphasizes sources whose signal is poorly captured by the shared low-rank model \citep{wang2026distributionally}. At the population level, this removes dependence on $\sigma_k$, yielding
\begin{equation}
\label{eq:sim_pop_risk_main}
    L_k(\Theta)=\frac{1}{2d_2}\Bigl\{\operatorname{tr}\bigl[(\Theta_k-\Theta)^\top A_k(\Theta_k-\Theta)\bigr]-\operatorname{tr}\bigl[\Theta_k^\top A_k\Theta_k\bigr]\Bigr\}, \ \mbox{with} \ A_k:=\Sigma_k+\mu_k\mu_k^\top .
\end{equation}
Since \eqref{eq:sim_pop_risk_main} is available in closed form, we compute the population saddle point $(\Tsl,\wsl)$ of \eqref{eq:PILLAR_def} to numerical
precision, and use $\Tsl$, not any individual $\Theta_k$, as the estimation target: it
is the object the estimator is defined to recover, and it moves with $\lambda$.

We compare PILLAR with 1) \textsc{Per-task}, pooled ridge fitted separately
per response, with neither a rank constraint nor reweighting; 2) \textsc{SVD}, its
truncation to rank $r^*$, low-rank by projection rather than by
optimization; 3) \textsc{BGD}, factorized low-rank gradient descent with the same
balancing penalty on the uniformly averaged source losses; 4) \textsc{L21}, $\ell_{2,1}$ row-sparse
multi-task regression \citep{argyriou2008convex}; 5) \textsc{Nuclear}, its nuclear-norm
counterpart \citep{koltchinskii2011nuclear,negahban2012restricted}, the convex
relaxation without robustness; and 6) \textsc{Maximin} \citep{sagawa2019distributionally}. Of particular interest are two ablations: BGD retains the low-rank structure but fixes the source weights at $w^{\mathrm{ref}}$, while maximin retains adversarial reweighting but removes the low-rank structure. 
All methods see the same data and, where applicable, the same rank and initialization. Following Section~\ref{subsec:init}, BGDMA is warm-started at the pooled ridge estimate followed by a balanced rank-$r$ singular value decomposition, with $w^0=w^{\mathrm{ref}}$, and is run with constant step sizes throughout.

On independent test samples we measure the explained-variation reward
\begin{equation}
\label{eq:sim_reward}
    R\uk(\Theta)
    =
    \bigl\{\|Y_{\mathrm{test}}\uk\|_F^2-\|Y_{\mathrm{test}}\uk-X_{\mathrm{test}}\uk\Theta\|_F^2\bigr\}
    \big/ (n_{\mathrm{test}}d_2),
\end{equation}
for which larger is better. Because $R\uk=-2\Lkh$ on the test data, the worst-source reward aligns directly with the $\lambda\to 0$ PILLAR objective. We report per-source rewards, their average and their worst-source value. For a mixture
$\QQ_v=\summ v_k\PP\uk$, we similarly report the mixture reward $R_v$ and its minimum over a collection of deployment mixtures.

\subsection{Optimization and estimation error}
\label{subsec:sim_rate}

We begin with the two error components of Theorem~\ref{thm:PILLAR_convergence}: the optimization error, which the theorem predicts decays geometrically, and the statistical error $\varepsilon_n$ of \eqref{eq:eps_n}, whose effective sample size $n_{\mathrm{eff}}=(\summ (w^*_{\lambda,k})^2/n_k)^{-1}$ makes a specific prediction about the number of sources. Throughout this subsection we take $m=3$ sources with $d_1=100$ features, $d_2=20$ tasks and true rank $r^*=5$.

\paragraph{Optimization error at the population level.}
We isolate the optimization error by replacing the empirical risks with the population risks \eqref{eq:sim_pop_risk_main}, so that no sampling noise intervenes and BGDMA is directed at the exact saddle point. Fixing $\lambda=1$, for which $\wsl\approx(0.34,0.46,0.20)$ so that no source dominates the target, we run the algorithm from a common initialization at four constant step sizes $\eta=\eta_\Theta=\eta_w$ ranging over $\{0.005,0.01,0.02,0.04\}$, an eightfold range.
Figure~\ref{fig:sim_population}(a) plots the resulting $\log\|\Theta^{(t)}-\Tsl\|_F$ against the iteration count. After a short transient every curve is straight on the log scale, which is the signature of the geometric factor $(1-\gamma)^T$ in \eqref{eq:theta_convergence},
and it appears across the whole range of $\eta$ rather than at a single tuned value. Larger step sizes lead to faster contraction rate over the range considered, although the improvement is limited by the coupling between the factor and weight updates. 

\paragraph{The population adversary as $\lambda$ varies.}
Solving the same population problem over a logarithmic grid of sixteen values $\lambda\in[0.1,100]$ traces the path of $\wsl$ shown in
Figure~\ref{fig:sim_population}(b). At $\lambda=0.1$ the adversary places more than half its mass on the hardest source and under 0.1 on the easiest; by $\lambda=100$ it
has returned to the uniform reference. Thus $\lambda$ moves the PILLAR target continuously from near worst-source protection to the reference mixture, illustrating the robustness--prior borrowing trade off as discussed in Section \ref{sec:method:prior}.

\begin{figure}[htb!]
    \centering
    \begin{subfigure}[b]{0.49\textwidth}
        \centering
        \includegraphics[width=\textwidth]{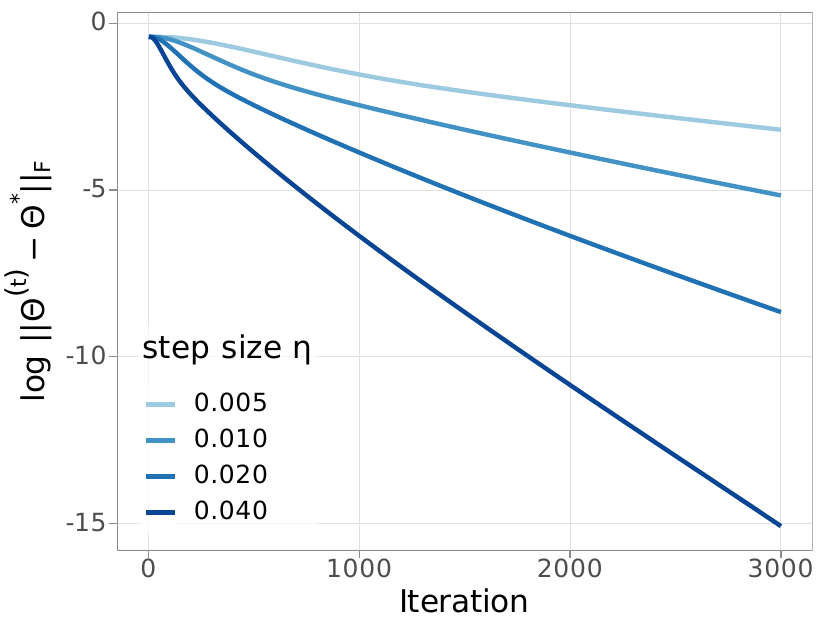}
        \caption{log optimization error}
        \label{fig:population_convergence}
    \end{subfigure}
    \hfill
    \begin{subfigure}[b]{0.49\textwidth}
        \centering
        \includegraphics[width=\textwidth]{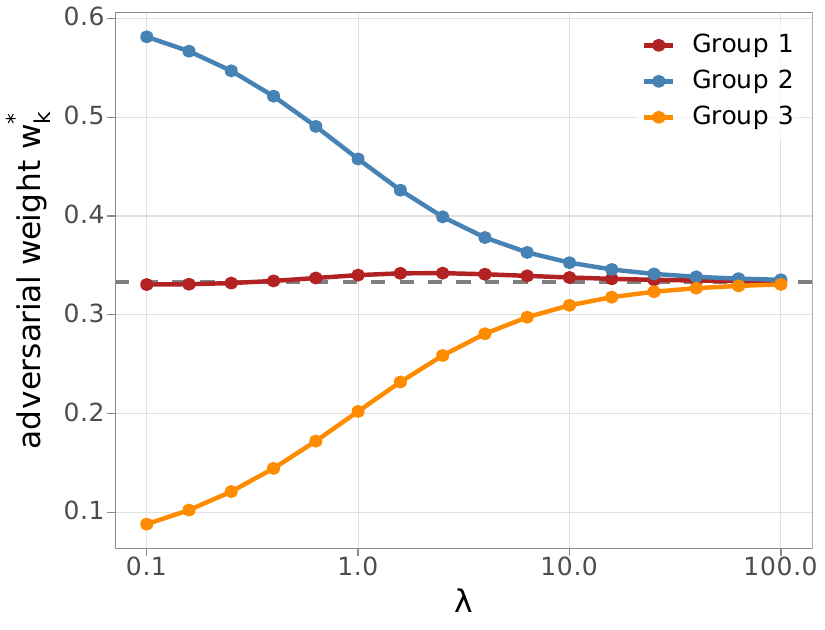}
        \caption{Population PILLAR weights}
        \label{fig:population_weights}
    \end{subfigure}

    \caption{Optimization behavior of BGDMA algorithm.
    \textbf{(a)} $\log\|\Theta^{(t)}-\Tsl\|_F$ against iteration for four constant step
    sizes at $\lambda=1$, all converging to the same saddle.
    \textbf{(b)} Population-optimal weights $w^*_{\lambda,k}$ over a logarithmic grid of
    $\lambda$. Small $\lambda$ concentrates mass on the source that the shared low-rank
    model fits worst; large $\lambda$ returns the weights to
    $w^{\mathrm{ref}}=1/m$ (dashed).}
    \label{fig:sim_population}
\end{figure}

\paragraph{Statistical error in finite samples.} We next assess the finite-sample error $\|\widehat\Theta-\Tsl\|_F$ via two sets of experiments: 1) set $m=3$ and increase the common per-source sample size $n_k$ from $600$ to $3600$; 2) set $n_k=400$ and increase the number of sources $m$ from $4$ to $59$. Each experiment is evaluated at $\lambda\in\{0.1,0.5,1\}$ and working ranks $r\in\{3,5,10\}$, with true rank $r^*=5$, averaged over $B=50$ independent replications. More details about the configurations can be found in Appendix~\ref{appen:sim_details}. Figure~\ref{fig:sim_rate} reports the mean log error with $\pm1$ standard deviation bands, on logarithmic axes so that a slope estimates a rate exponent.

In panel (a), the correctly specified curves are parallel to the plotted $n^{-1/2}$ guide. This is the behavior \eqref{eq:eps_n} predicts once $n_{\mathrm{eff}}\asymp n$ with $m$ held fixed.
When the true rank is unknown, rank under-specification produces persistent approximation bias, whereas over-specification still allows convergence although more slowly over this range. Therefore, when the true rank is unknown, it is safer to use a more conservative rank. 

In panel (b), the behavior depends sharply on $\lambda$, and this can be reflected in the theory through $n_{\mathrm{eff}}$. When $\lambda$ is moderate, the adversary remains diffuse, $w^*_{\lambda,k}\approx1/m$, so $n_{\mathrm{eff}}\asymp mn$ and each new environment contributes; when $\lambda$ is small, the adversary concentrates on a few high-loss sources, $n_{\mathrm{eff}}$ then saturates at their scale, and the remaining environments go almost unused. Additional environments therefore contribute informative only to the extent that the population adversary remains sufficiently diffuse, which is controlled by $\lambda$.

\begin{figure}[htb!]
    \centering
    \begin{subfigure}[b]{0.9\textwidth}
    \centering
\includegraphics[width=\linewidth]{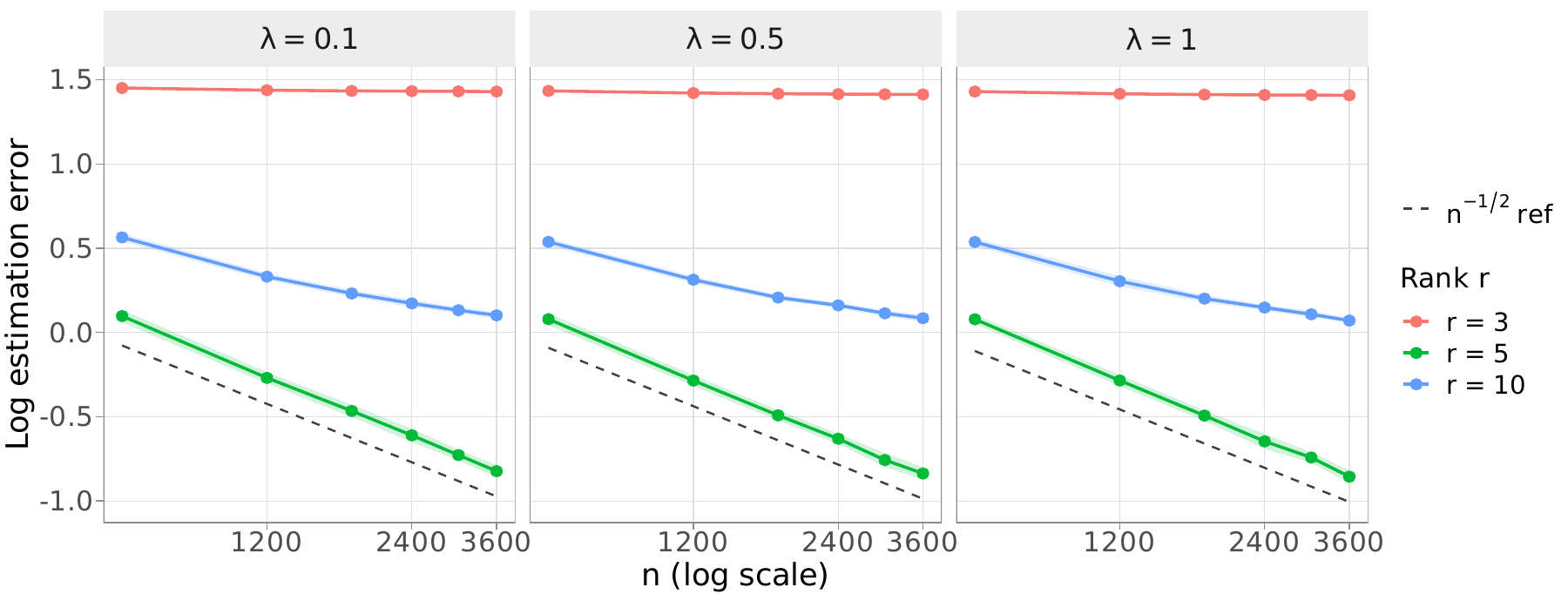}
\caption{$m=3$ and $n_k$ increasing from $600$ to $3600$ ($B=50$ replications).}
    \end{subfigure}
    \begin{subfigure}[b]{0.86\textwidth}
    \centering
\includegraphics[width=\linewidth]{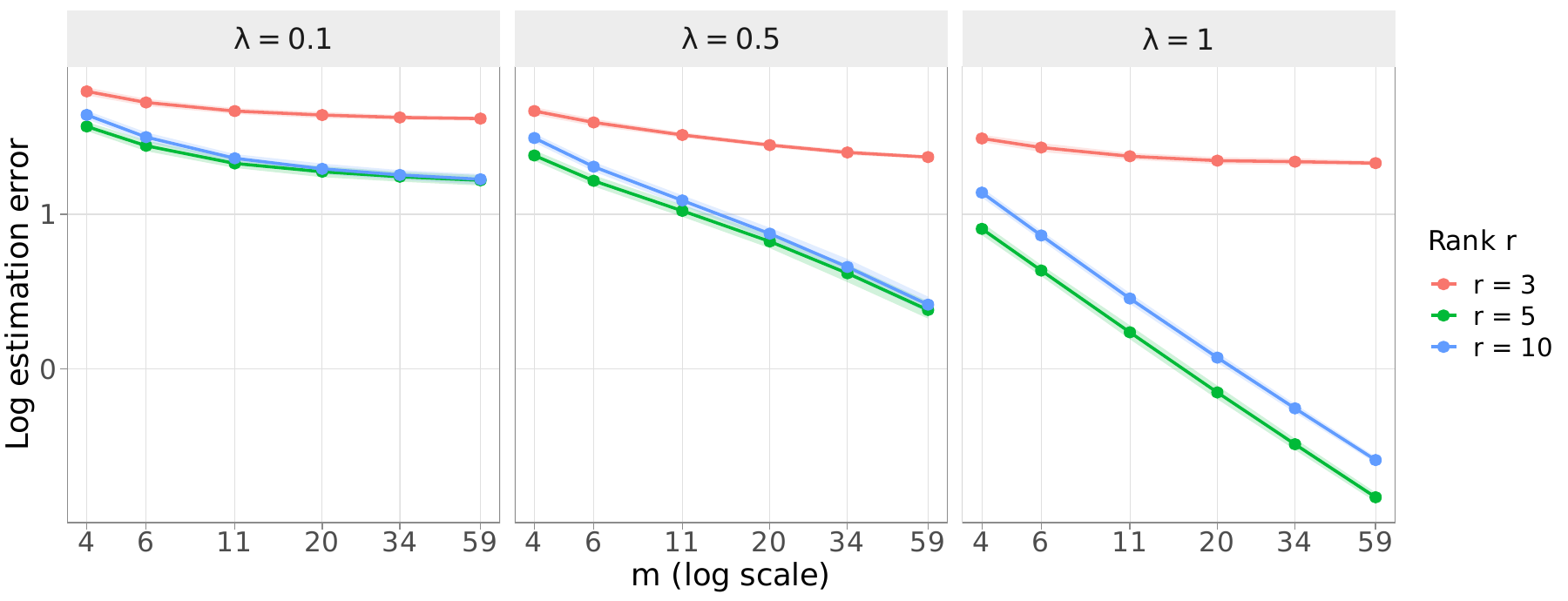}
\caption{$n_k=400$ and $m$ increasing from $4$ to $59$ ($B=50$ replications).}
    \end{subfigure}
    \caption{Estimation error relative to the population target $\Tsl$: mean
    $\log\|\widehat\Theta-\Tsl\|_F$ over independent replications with $\pm1$ standard deviation
    bands, by $\lambda$ (panels) and working rank $r$ against $r^*=5$
    (colors).}
    \label{fig:sim_rate}
\end{figure}

\begin{remark}[On the valid range of $\lambda$.]
    Both estimation error studies use $\lambda\in\{0.1,0.5,1\}$, and Section~\ref{subsec:sim_pred} adds the unregularized case $\lambda=0$. These values are possible to lie below the threshold $\lambda\ge c_1\sqrt{\kappa^*}$ of Theorem~\ref{thm:PILLAR_convergence}, yet BGDMA converged in every configuration we ran.
    As Remark~\ref{rem:lambda_condition} explains, the condition is sufficient for our argument rather than a constraint for the method, and the numerical evidence suggests that the usable range of $\lambda$ is considerably wider than the theory currently certifies. 
\end{remark}


\subsection{Predictive robustness under mixture shift}
\label{subsec:sim_pred}

We evaluate the predictive performance of PILLAR in a larger and harder design: $d_1=200$ features, $d_2=40$ tasks, $m=3$ sources and true rank $r^*=5$, with $n_k=150$ observations per source. Since $n_k<d_1$, no source can be well-fitted on its own, and it is thus important to borrow strength across sources. PILLAR is fitted at $\lambda\in\{0,0.1,0.5,1,5,10\}$ with the working rank set to $r^*$, and all estimators are evaluated on independent test samples of $n_{\mathrm{test}}=2000$ observations. We evaluate predictive robustness in three complementary settings: the observed sources; new environments outside the observed source family; and a continuum of source mixtures interpolating between the two.

\paragraph{The observed environments.}
Figure~\ref{fig:sim_pred}(a) evaluates every estimator on fresh test samples from each of the 3 training source distributions, reporting the 3 per-source rewards together with their average and minimum. Every PILLAR fit exceeds every benchmark in
worst-source reward, the criterion directly targeted as $\lambda \to 0$: the best benchmark reaches $6.92$ (\textsc{BGD}), while PILLAR ranges from $7.29$ at $\lambda=10$ to $8.16$ at
$\lambda=0$. The two ablations help separate the sources of this gain. 
\textsc{BGD} removes the adversary but keeps the
low-rank geometry, reducing the worst-source reward by $1.24$ relative to PILLAR at $\lambda=0$. \textsc{Maximin} retains adversarial reweighting but removes the low rank structure and reaches only $4.39$. With $n_k<d_1$, the unstructured minimax fit cannot borrow strength across sources. The panel also displays the robustness-efficiency trade-off: from $\lambda=0$ to $\lambda=10$, the worst-source reward falls from $8.16$ to $7.29$ while the average rises from $8.48$ to $9.42$. 

\paragraph{Out-of-distribution environments.}
Figure~\ref{fig:sim_pred}(b) leaves the convex hull of the observed sources
altogether. We draw 100 new environments
which are not any mixture of the observed sources; the full OOD configuration can be found in Appendix \ref{appen:subsec:sim_eval}. In this out-of-distribution setting, PILLAR shows the strongest overall reward distribution, with a median reward of $8.29$ against $7.74$ for the best benchmark. The gain therefore transports beyond the mixture family, suggesting that adversarial reweighting recovers latent directions that are stable across environments rather than merely reweighting the training sample. 

\begin{figure}[htbp]
    \centering
    \includegraphics[width=0.9\linewidth]{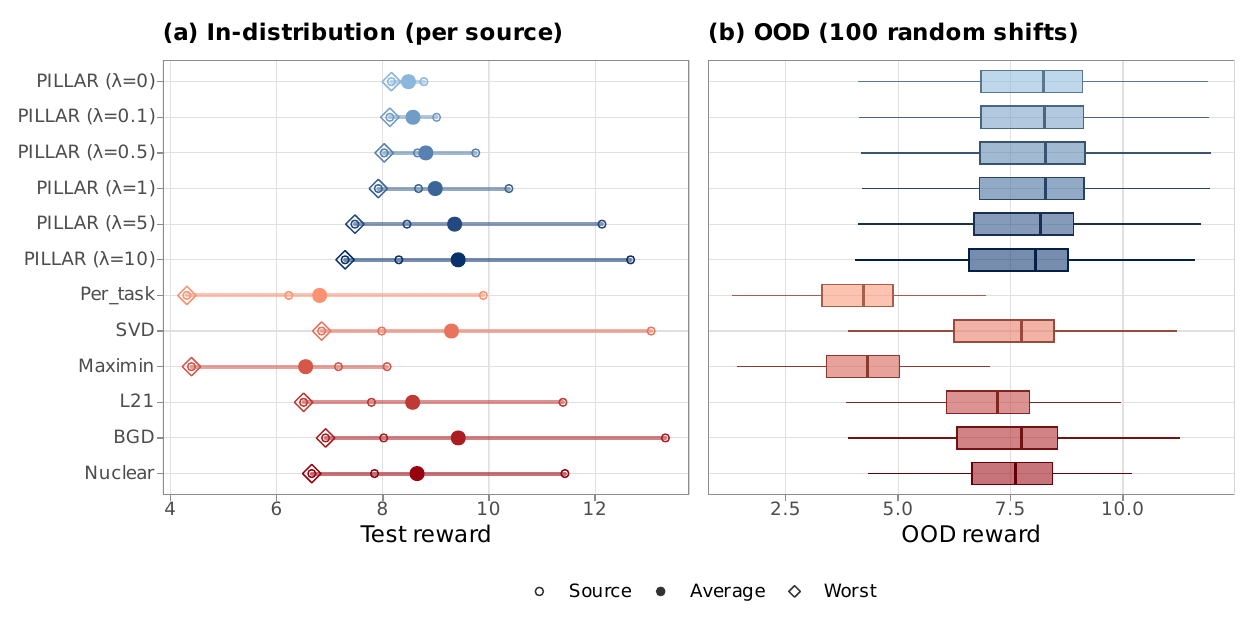}
    \caption{Predictive comparison in the design of Section~\ref{subsec:sim_pred}; larger rewards \eqref{eq:sim_reward} are better, PILLAR
    in blue and benchmarks in red. \textbf{(a)} In-distribution, on $2000$ fresh
    observations per training source: segments give the range of the three per-source
    rewards, open circles the sources, filled circles their average, diamonds the
    worst source. \textbf{(b)} Out-of-distribution, over $100$ environments whose
    signal matrix, covariate mean and scale, and noise level are all drawn outside the
    training configuration.}
    \label{fig:sim_pred}
\end{figure}

\paragraph{A continuum of deployment mixtures.}
The third regime evaluates PILLAR on a continuum of deployment mixtures and directly examines the prediction of Theorem~\ref{thm:prior_benefit_gdro}. The theorem bounds the excess worst-case risk of PILLAR relative to vanilla group DRO over the
neighborhood $\{q\in\Delta_m:D(q)\le\rho\}$. 
To construct an empirical analog of $\rho$, at each proximity level $\alpha$, we draw $500$ mixture weights $v\sim\mathrm{Dirichlet}(\alpha\mathbf 1_m)$, evaluate every estimator on a fresh test sample from $\QQ_v=\summ v_k\PP\uk$, and record the worst reward $\min_v R_v$ over the $500$ draws. Thus small $\alpha$ spreads the mixtures toward the simplex boundary, corresponding to a larger evaluation radius, whereas large $\alpha$ concentrates them near $w^{\mathrm{ref}}$, so $\rho \to 0$. We report this worst reward over the rescaled proximity level for all the methods in Figure~\ref{fig:sim_continuum}. 

We observe two patterns from the results. First, PILLAR lies above every benchmark along the whole path. Its margin is largest when the deployment mixtures are most dispersed and vanishes as they concentrate near the reference mixture: $+1.20$ in reward over the best benchmark at the adversarial end (left) against $+0.04$ at the uniform end (right), where the evaluation mixture is close to the training composition. Second, the best-performing $\lambda$ is interior and rises as the neighborhood contracts, from $\lambda=1$ at the adversarial end to $\lambda=10$ at the uniform end. This is consistent with the term $H^2/8\lambda +\lambda \rho$ in Theorem~\ref{thm:prior_benefit_gdro}: the best $\lambda$ is taken interior, and as $\rho$ gets smaller the best $\lambda$ grows larger.

\begin{figure}[htb]
    \centering
    \includegraphics[width=0.87\linewidth]{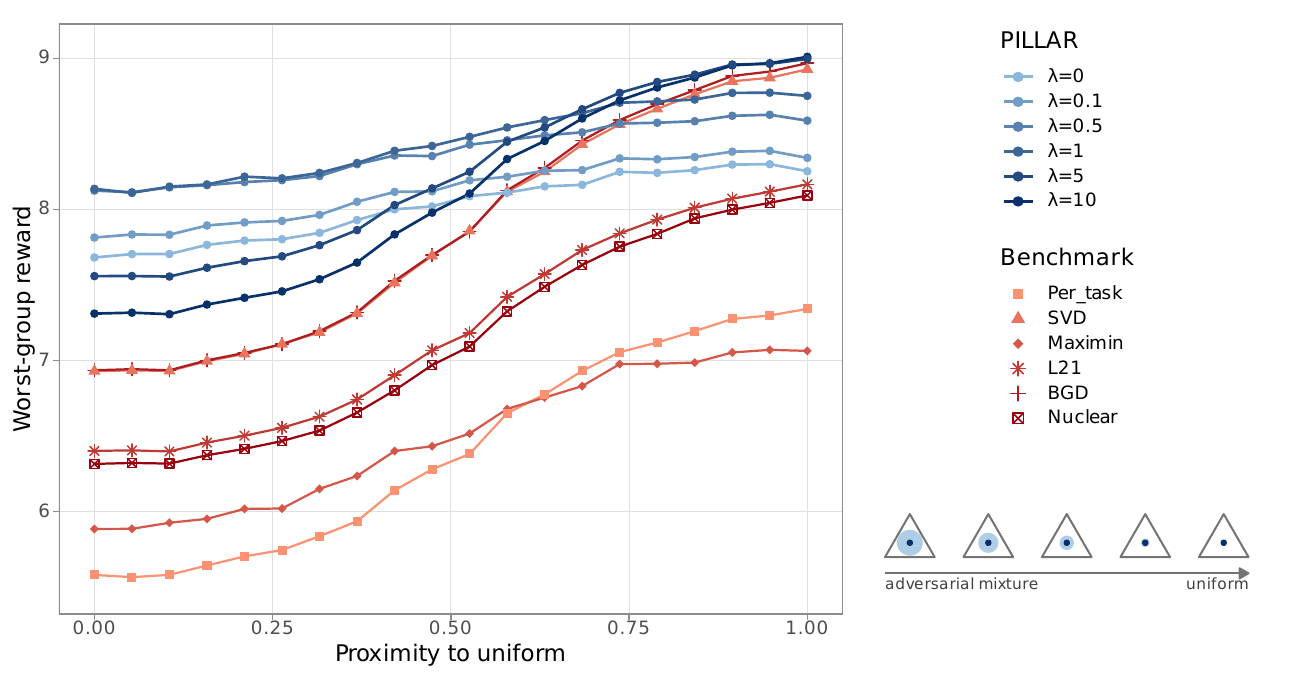}
    \caption{Worst-case reward over a continuum of deployment mixtures. The $x$-axis is the rescaled
    $\log\alpha$ over $\alpha\in[0.1,100]$, so the left end corresponds to a large
    evaluation radius $\rho$ in Theorem~\ref{thm:prior_benefit_gdro} and the right end
    to $\rho\to0$. Curves average $20$ training draws. PILLAR (blue) exceeds every benchmark (red) throughout.}
    \label{fig:sim_continuum}
\end{figure}

\section{Application: Robust EHR-Based Risk Prediction}
\label{sec:real}

Suicide is a leading cause of death worldwide, with an estimated $727{,}000$ deaths in 2021 \citep{who2025suicide}. Health systems increasingly use routinely collected EHR data to flag patients at elevated short-term risk; such models have now been developed and validated across several systems \citep{barak2017predicting,barak2020validation,bentley2025clinician}. Two features of this setting motivate the two components of PILLAR. First, EHR data usually drift as diagnostic coding, care pathways, and recorded incidence all change over calendar time, so a model fitted on historical data may be deployed in a future period with a different composition. We therefore treat calendar periods as sources and use adversarial reweighting to guard against temporal shift. Second, suicide risk is strongly age-dependent: presentation, incidence, and the predictive value of diagnoses can differ substantially across age groups, while some strata (e.g. youngest) are too small to support separate high-dimensional models. We thus treat age strata as tasks and jointly estimate them through a shared low-rank representation. The application therefore follows  the framework of Section~\ref{sec:problem_formulation}, with sources indexing calendar time and tasks indexing age group.

\subsection{Cohort and model specification}
\label{subsec:real_model}

The cohort is drawn from the Mass General Brigham (MGB) Research Patient Data Registry, a centralized repository of longitudinal EHR data from MGB. Patients are included if they have a mental-health diagnosis (PheCodes 295--306) recorded between 2005 and 2022 and at least one code in the preceding two years. PheCodes groups related international classification of diagnosis (ICD) codes into clinically coherent categories \citep{denny2013systematic} such as 297.2 for suicide attempt or self-inflicted injury. We write $t_0$ for that index mental-health encounter, and take as outcome a recorded PheCode 297.2 within the following $48$ weeks. Predictors are built from the 12 months before $t_0$: $233$ code-by-window counts using codes selected by the knowledge-graph-based ONCE procedure \citep{xiong2026bending} and recorded separately over the two 6-month windows preceding $t_0$, together with sex, race and age at $t_0$. This yields $p=238$ variables, standardized using training-sample statistics only. 

The $m=4$ sources are the calendar periods and the $T=6$ tasks are the age strata at $t_0$, yielding the $24$ cells in Table~\ref{tab:real_cohort} and $328{,}823$ patients. The tabel also shows substantial heterogeneity: across calendar periods, cohort size decreases by a factor of $5.5$ while the recorded event rate more than doubles. Therefore, simple pooling is dominated by the oldest data, whereas the period closest to deployment is the smallest and differs most from the pooled training distribution. Event rates also vary markedly across age strata, making separate high-dimensional fitting unstable in the smaller groups. 

\begin{table}[htbp]
\centering
\small
\caption{Cohort: patients and, in parentheses, the one-year suicide-attempt rate
(\%) per calendar period $\times$ age stratum. Sources are periods, tasks are strata.
Marginal rates are computed from the cell counts.}
\label{tab:real_cohort}
\begin{tabular}{@{}lrrrrrrr@{}}
\toprule
& \multicolumn{6}{c}{Age stratum at $t_0$} & \\
\cmidrule(lr){2-7}
Period & $0$--$14$ & $14$--$18$ & $18$--$25$ & $25$--$40$ & $40$--$60$ & $60$--$100$ & Total \\
\midrule
2005--2014 & 5{,}911 & 6{,}133 & 16{,}475 & 34{,}171 & 47{,}387 & 27{,}580 & 137{,}657 \\
           & (0.20) & (0.91) & (0.60) & (0.34) & (0.19) & (0.09) & (0.29) \\[3pt]
2015--2018 & 4{,}204 & 5{,}384 & 13{,}386 & 25{,}957 & 31{,}639 & 30{,}130 & 110{,}700 \\
           & (0.76) & (2.88) & (1.65) & (0.72) & (0.49) & (0.21) & (0.73) \\[3pt]
2019--2020 & 1{,}655 & 2{,}516 & 7{,}191 & 14{,}005 & 14{,}428 & 15{,}434 & 55{,}229 \\
           & (2.84) & (3.78) & (1.17) & (0.76) & (0.42) & (0.30) & (0.79) \\[3pt]
2021--2022 & 732 & 1{,}203 & 3{,}503 & 6{,}489 & 6{,}361 & 6{,}949 & 25{,}237 \\
           & (3.01) & (4.32) & (1.08) & (0.52) & (0.49) & (0.19) & (0.75) \\
\midrule
Total      & 12{,}502 & 15{,}236 & 40{,}555 & 80{,}622 & 99{,}815 & 80{,}093 & 328{,}823 \\
           & (0.90) & (2.35) & (1.09) & (0.55) & (0.34) & (0.18) & (0.56) \\
\bottomrule
\end{tabular}
\end{table}

Let $\beta_j\in\RR^{p}$ be the coefficient vector for stratum $j$, let
$\Theta=[\beta_1,\dots,\beta_6]$, and write $\{(x_{ijk},y_{ijk})\}_{i\le n_{jk}}$ for
the data in period $k$ and stratum $j$. We fit the logistic model
$\PP(y_{ijk}=1\mid x_{ijk})=\sigma(x_{ijk}^\top\beta_j)$, where
$\sigma(u)=(1+e^{-u})^{-1}$, and define the source risk as the stratum-averaged
per-patient negative log-likelihood
\begin{equation}
\label{eq:real_loss}
    \Lkh(\Theta)
    =
    \frac{1}{6}\sum_{j=1}^{6}\frac{1}{n_{jk}}\sum_{i=1}^{n_{jk}}
    \bigl[\log\{1+\exp(x_{ijk}^\top\beta_j)\}-y_{ijk}\,x_{ijk}^\top\beta_j\bigr].
\end{equation}
Averaging equally across strata rather than pooling patients within a period prevents the largest strata from dominating the source loss and hence the adversarial weighting. PILLAR solves \eqref{eq:empirical_minmax} with $\Theta=UV^\top$ of rank $r=3$,  and KL regularization towards a uniform prior $w^{\mathrm{ref}}=\tfrac14\mathbf 1$. 
We compare PILLAR with the logistic analogs of the 6 benchmark estimators inSection~\ref{subsec:sim_design}, with
\textsc{BGD} denoting the pooled low-rank fit in which patients within a stratum are weighted equally, so periods enter in proportion to their size. PILLAR is fitted with the pre-specified value $\lambda=0.1$, without using the held-out period for tuning. Performance is measured by AUC within each period $\times$ stratum cell and then averaged over strata, so every age group counts equally regardless of size.

\subsection{Discrimination and transportability}
\label{subsec:real_auc}

Figure~\ref{fig:real_auc}(a) reports $5$-fold cross-validated AUC by period. Each time-period is divided into 5 folds, and methods are trained on four folds from each time-period and then validated on the held-out fold in every time-period. PILLAR consistently achieves the best AUC compared to all the benchmark methods in all the time periods. This consistency is meaningful because the periods differ substantially in size, event rate and coding practice. 

We next evaluate temporal transportability. In Figure~\ref{fig:real_auc}(b), we use 2005–2008, 2009–2012, and 2013–2016 as training data and evaluate on three held-out test periods: 2017–2018, 2019–2020, and 2021–2022. These targets lie genuinely outside the training data sources rather than being reweightings of them. Moreover, due to the evolution of medication and coding practices, there is substantial heterogeneity between older and newer data. Figure~\ref{fig:real_auc}(b) shows the test AUC for each held-out period, as well as the average and worst AUCs across the three test periods. PILLAR outperforms all benchmarks in each period. For example, it attains the highest average transported AUC of 0.784, compared with 0.763 for the best benchmark. These results show that PILLAR achieves strong generalizability across time and is robust to temporal drift.

\begin{figure}[htbp]
    \centering
    \begin{minipage}{0.85\textwidth}
        \centering
        \begin{subfigure}[b]{0.47\linewidth}
            \centering
            \includegraphics[width=\linewidth]
                {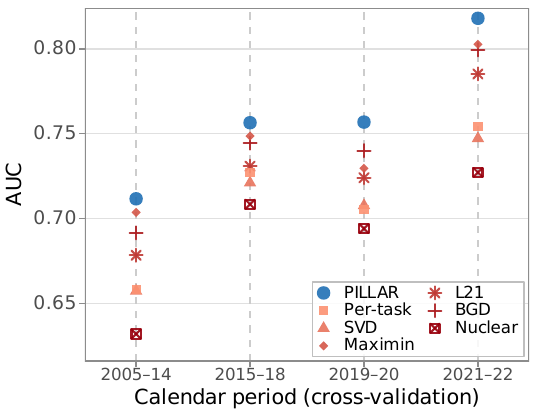}
            \caption{Cross-validation testing}
            \label{fig:cv_by_period}
        \end{subfigure}
        \hfill
        \begin{subfigure}[b]{0.5\linewidth}
            \centering
            \includegraphics[width=\linewidth]
                {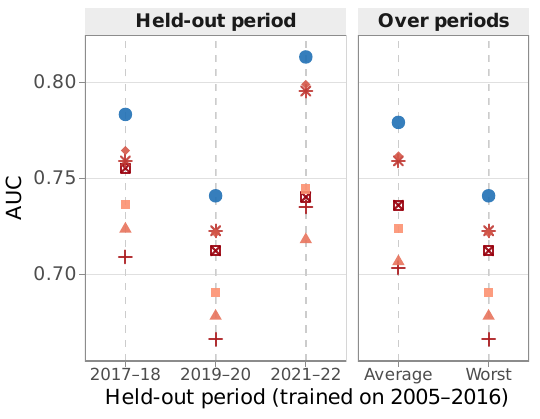}
            \caption{Temporal transportability}
            \label{fig:multi-period}
        \end{subfigure}
    \end{minipage}

    \caption{Discrimination in the cohort of Table~\ref{tab:real_cohort}; each column compares the seven methods within one evaluation stratum, PILLAR ($\lambda=0.1$) in blue. \textbf{(a)} Five-fold cross-validated AUC by calendar period, folds assigned within period $\times$ stratum cells and stratified on the outcome, averaged over folds then strata. \textbf{(b)} Transportability: fitted on 2005--2016 and evaluated on three held-out periods.}
    \label{fig:real_auc}
\end{figure}

\subsection{PILLAR representation interpretation}
\label{subsec:real_rep}

Because the model predicts suicide risk separately across age strata, the fitted coefficient patterns provide a way to examine how predictive associations vary with age. Figure~\ref{fig:real_heatmap} displays the fitted coefficients of a PILLAR fit to the full cohort, as the change in log-odds per within-stratum standard deviation of each predictor. We have several interpretations for the findings that PILLAR discovers.

\begin{figure}[htb]
    \centering
\includegraphics[width=0.9\linewidth]{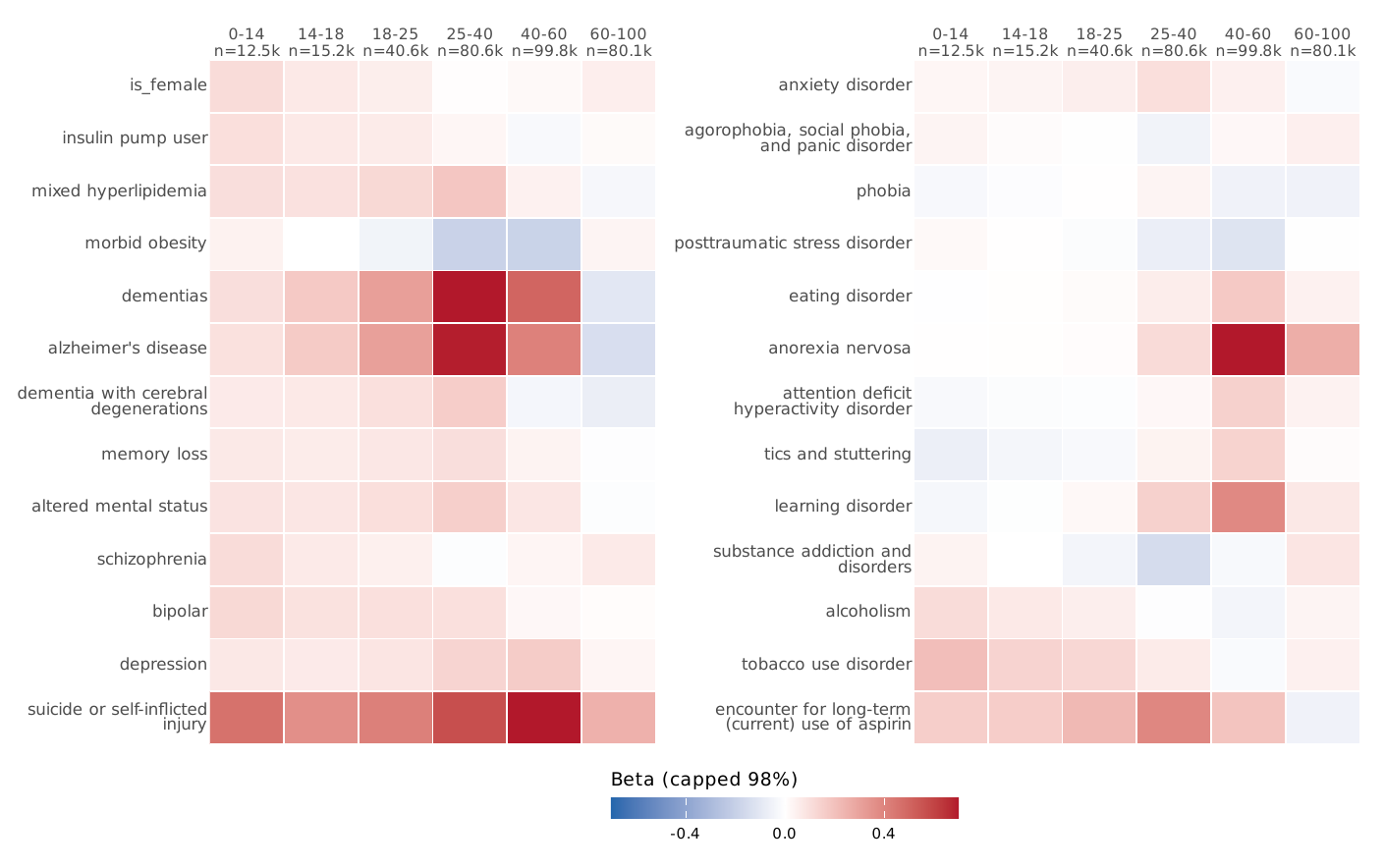}
    \caption{Fitted PILLAR coefficients by predictor (rows, in two blocks) and age stratum (columns, with cell sample size), as the change in log-odds.  The displayed set is the union over strata of the ten largest effects in each. The color caps at $98$th percentile of the absolute value. Blue is negative, red positive. }
    \label{fig:real_heatmap}
\end{figure}

First, a prior record for suicide or self-inflicted injury is the largest positive effect in all six age strata, ranging from $0.22$ in the oldest stratum to $0.85$ in adolescents. This is an example of a predictive pattern that is shared across age groups, although its magnitude differs by age. Second, many other predictors show clear age-specific patterns. Some change substantially in magnitude across strata, and some even change signs. 
For example, dementia and Alzheimer codes are most frequent in the oldest stratum, but their coefficients peak earlier in $40$--$60$ and {reverse} sign after $60$, which is consistent with the literature of post-diagnosis suicide risk among patients diagnosed before age $65$ \citep{alothman2022risk}. Several mechanisms are compatible with that reversal: a neurodegenerative code at $40$--$60$ marks young-onset disease with preserved insight, whereas after $60$ it marks the age-typical course, when planning capacity declines and supervision rises; and in the oldest stratum such coding is near-routine, so conditional on the remaining predictors it carries little residual signal. 
Other examples include that anorexia nervosa is most frequent in adolescence yet peaks in $40$--$60$, where a persisting diagnosis marks the long-standing illness whose suicide risk is markedly elevated
\citep{arcelus2011mortality}. 

Together, these patterns illustrate two components of PILLAR. The low-rank constraint allows the six age-specific coefficient vectors to share a small number of latent directions while still permitting their predictive effects to vary across age groups. At the same time, the adversarial weighting over calendar periods favors representations that remain predictive across temporal shifts rather than being driven primarily by the largest or easiest period. Thus, PILLAR uses low-rank structure to capture shared predictive patterns across age strata and robust reweighting to make those patterns more stable across calendar time.

\section{Discussion}
\label{sec:discussion}

PILLAR combines low-rank structure with prior-informed distributional robustness to address two complementary challenges in multi-source learning: sharing information across related tasks while protecting against uncertainty in the composition of the future deployment population. The low-rank structure reduces estimation complexity by exploiting common latent structure, while adversarial reweighting guards against source heterogeneity and distribution shift. At the same time, regularizing the adversary toward a reference mixture provides a continuum between prior-weighted pooling and vanilla group DRO. Our theory further shows that the benefit of additional sources is governed by the effective sample size $n_{\mathrm{eff}}$. When the adversarial weights remain diffuse, information accumulates across sources, whereas concentrated weights reduce the effective amount of information available for estimation. 

Several directions remain open. The first concerns the range of $\lambda$ covered by the theory. Theorem~\ref{thm:PILLAR_convergence} requires $\lambda\ge c_1\sqrt{\kappa^*}$, which is a sufficient condition for our theoretical guarantees rather than an inherent limitation of the estimator. Establishing rates for general $\lambda$, and in particular developing a limiting theory as $\lambda\to0$ or directly analyzing the unregularized case $\lambda=0$, would help close the gap between what is currently proved and what we observe empirically. Another important direction is to develop principled, data-adaptive rules for choosing both $\lambda$ and the prior encoded by $D(w)$. Such choices could account for the degree of heterogeneity across sources, available prior knowledge about the deployment population, and the downstream objective of interest. More broadly, PILLAR can be viewed as a ``soft" form of group DRO. Further 
characterizing how prior informed regularization controls the concentration of the adversarial weights and hence the effective sample size $n_{\mathrm{eff}}$, may provide useful insights for group DRO more generally. An important limitation of the current formulation is that robustness is defined through reweighting the observed source distributions. Interestingly, our out-of-distribution simulations and temporal transport experiment suggest that the learned representation can remain useful beyond this mixture family. Understanding when robustness to source reweighting translates into robustness to genuinely new environments warrants future research.

Although our simulations and real-data study focus on multi-task prediction, the principles underlying PILLAR extend naturally to other low-rank problems, including matrix recovery, principal component analysis (PCA), and representation learning across institutions or environments. Our current analysis focuses primarily on optimization and estimation; calibration, uncertainty quantification, and formal inference for robust low-rank predictors remain important open problems. Richer structural models, such as low-rank-plus-sparse effects or hierarchical sharing across tasks, may further improve the balance among robustness, predictive performance, and interpretability in applications. 


{
\spacingset{1.0}
\vspace{10pt}
\footnotesize
\setlength{\bibsep}{5pt}
\bibliographystyle{plainnat}
\bibliography{references}
}

\newpage

\appendix

\setcounter{assumption}{0}
\setcounter{lemma}{0}
\setcounter{theorem}{0}
\setcounter{proposition}{0}
\setcounter{corollary}{0}
\setcounter{definition}{0}
\setcounter{remark}{0}
\setcounter{claim}{0}

\renewcommand{\theassumption}{\thesection.\arabic{assumption}}
\renewcommand{\thelemma}{\thesection.\arabic{lemma}}
\renewcommand{\thetheorem}{\thesection.\arabic{theorem}}
\renewcommand{\theproposition}{\thesection.\arabic{proposition}}
\renewcommand{\thecorollary}{\thesection.\arabic{corollary}}
\renewcommand{\thedefinition}{\thesection.\arabic{definition}}
\renewcommand{\theremark}{\thesection.\arabic{remark}}
\renewcommand{\theclaim}{\thesection.\arabic{claim}}

\numberwithin{equation}{section}
\makeatletter
\@addtoreset{assumption}{section}
\@addtoreset{lemma}{section}
\@addtoreset{theorem}{section}
\@addtoreset{proposition}{section}
\@addtoreset{corollary}{section}
\@addtoreset{definition}{section}
\@addtoreset{remark}{section}
\@addtoreset{claim}{section}
\makeatother

\section{Proof of Theorem~\ref{thm:PILLAR_convergence}}
\label{appen:proof_thm1}

This appendix proves Theorem~\ref{thm:PILLAR_convergence}. The argument has a deterministic part and a stochastic part, which we present separately. Section~\ref{appen:subsec:notation} fixes notation and defines the two statistical errors that drive the final rate. Section~\ref{appen:subsec:lemmas} collects auxiliary lemmas on the Burer--Monteiro geometry and on the regularity of the losses; Section~\ref{appen:subsec:neighborhood} quantifies the local neighborhood in which the analysis takes place. Sections~\ref{appen:subsec:onestepZ} and~\ref{appen:subsec:onestepw} establish one-step contraction bounds for the factor update and for the mirror-ascent update, respectively. Section~\ref{appen:subsec:combine} combines the two recursions through a joint Lyapunov potential and proves a deterministic convergence theorem (Theorem~\ref{thm:appen:deterministic}). Section~\ref{appen:subsec:estat} bounds the statistical errors under Assumption~\ref{ass:loss_psi2}, and Section~\ref{appen:subsec:thm1_final} assembles the pieces into Theorem~\ref{thm:PILLAR_convergence}.

\subsection{Notation and statistical errors}
\label{appen:subsec:notation}

Throughout the appendices, $c,c',C,C',C_0,C_1,\dots$ denote positive universal constants whose values may change from line to line; constants that depend on the regularity parameters of Assumptions~\ref{ass:rsc_rss}--\ref{ass:w_regularizer_mirror} are indicated explicitly. For a matrix $A$, $\|A\|_2$, $\|A\|_F$ and $\|A\|_*$ denote the spectral, Frobenius and nuclear norms, and $\sigma_j(A)$ the $j$-th largest singular value. We abbreviate
\[
\sigma_{\min}:=\sigma_{\min}(\Theta^*)=\sigma_r(\Theta^*),
\qquad
\sigma_{\max}:=\sigma_{\max}(\Theta^*)=\sigma_1(\Theta^*),
\qquad
\kappa^*:=\sigma_{\max}/\sigma_{\min},
\]
where $(\Theta^*,w^*)$ is the population saddle point of \eqref{eq:PILLAR_def}: $\Theta^*\in\Xi(r)$ minimizes $\Phi_\lambda(\Theta):=\max_{w\in\Delta_m}\{\summ w_kL_k(\Theta)-\lambda D(w)\}$ over $\Xi(r)$, and $w^*=w_\lambda^*$ is the associated adversarial weight vector defined in \eqref{eq:PILLAR_def}. Without loss of generality we assume $\rank(\Theta^*)=r$.

\paragraph{Balanced factors and the lifted matrix.}
Write the singular value decomposition $\Theta^*=U_{\Theta^*}\Sigma_{\Theta^*}V_{\Theta^*}^\top$ and define the balanced factors
\[
U^*:=U_{\Theta^*}\Sigma_{\Theta^*}^{1/2}\in\RR^{d_1\times r},
\qquad
V^*:=V_{\Theta^*}\Sigma_{\Theta^*}^{1/2}\in\RR^{d_2\times r},
\]
so that $\Theta^*=U^*(V^*)^\top$ and $(U^*)^\top U^*=(V^*)^\top V^*=\Sigma_{\Theta^*}$. For a generic pair $(U,V)\in\RR^{d_1\times r}\times\RR^{d_2\times r}$ we write the lifted (stacked) matrix and its balanced-reference counterpart as
\[
Z:=\begin{pmatrix} U\\ V\end{pmatrix}\in\RR^{(d_1+d_2)\times r},
\qquad
Z^*:=\begin{pmatrix} U^*\\ V^*\end{pmatrix},
\]
and we also use the sign-flipped stack $\breve Z:=(U;-V)$. Note that $\sigma_i(Z^*)^2=2\sigma_i(\Theta^*)$ for $i\in[r]$; in particular $\|Z^*\|_2^2=2\sigma_{\max}$ and $\sigma_r(Z^*)^2=2\sigma_{\min}$. Because the factorization $\Theta=UV^\top$ is invariant under $(U,V)\mapsto(UO,VO)$ for $O\in\mathbb{O}(r)$, we measure the error of $Z$ by the rotation-invariant distance
\begin{equation}
\label{eq:appen:dist}
d^2(Z,Z^*):=\min_{O\in\mathbb{O}(r)}\bigl\{\|U-U^*O\|_F^2+\|V-V^*O\|_F^2\bigr\}
=\min_{O\in\mathbb{O}(r)}\|Z-Z^*O\|_F^2 .
\end{equation}
Given $Z$, we fix a minimizer $\widehat O\in\argmin{O\in\mathbb{O}(r)}\|Z-Z^*O\|_F$ and write
\[
\Delta_U:=U-U^*\widehat O,\qquad
\Delta_V:=V-V^*\widehat O,\qquad
\Delta_Z:=Z-Z^*\widehat O=(\Delta_U;\Delta_V),
\]
so that $d^2(Z,Z^*)=\|\Delta_Z\|_F^2=\|\Delta_U\|_F^2+\|\Delta_V\|_F^2$. Two matrices associated with the balancing penalty recur in the analysis:
\[
M:=U^\top U-V^\top V,
\qquad
N:=(U^*)^\top\Delta_U-(V^*)^\top\Delta_V .
\]
Since $p(U,V)=\frac14\|M\|_F^2$, a direct computation gives $\nabla_U p=UM$, $\nabla_V p=-VM$, that is, $\nabla_Z p(Z)=\breve ZM$.

\paragraph{Weighted gradients and statistical errors.}
For $w\in\Delta_m$ let
\[
\nabla_\Theta\widehat L(\Theta,w):=\summ w_k\nabla_\Theta\Lkh(\Theta)\in\RR^{d_1\times d_2},
\qquad
\widehat L(\Theta):=\bigl(\widehat L_1(\Theta),\dots,\widehat L_m(\Theta)\bigr)^\top .
\]
For an integer $s\ge1$ define the rank-restricted dual norm
\begin{equation}
\label{eq:appen:restricted_norm}
\|A\|_{F,s}:=\sup\bigl\{\langle A,\Delta\rangle:\ \Delta\in\Xi(s),\ \|\Delta\|_F\le1\bigr\},
\qquad A\in\RR^{d_1\times d_2};
\end{equation}
this is a norm, $\|A\|_{F,s}\le\|A\|_F$, and the constraint set is symmetric, so $|\langle A,\Delta\rangle|\le\|A\|_{F,s}\|\Delta\|_F$ for every $\Delta\in\Xi(s)$. The two statistical errors that appear in Theorem~\ref{thm:PILLAR_convergence} are defined as follows.

\begin{definition}[Statistical errors]
\label{def:stat_error}
With $(\Theta^*,w^*)$ the population saddle point,
\[
\estatt:=\Bigl\|\summ w_k^*\nabla_\Theta\Lkh(\Theta^*)\Bigr\|_{F,2r}
=\sup_{\Delta\in\Xi(2r),\,\|\Delta\|_F\le1}
\Bigl\langle\summ w_k^*\nabla_\Theta\Lkh(\Theta^*),\Delta\Bigr\rangle,
\]
and
\[
\estatw:=\bigl\|\widehat L(\Theta^*)-L(\Theta^*)\bigr\|_\infty .
\]
\end{definition}

The quantity $\estatt$ measures the size of the weighted empirical gradient at the truth in low-rank directions; $\estatw$ measures the fluctuation of the empirical group losses at the truth, which is the quantity to which the mirror step responds. Both vanish in the noiseless limit, and Section~\ref{appen:subsec:estat} shows that under Assumption~\ref{ass:loss_psi2} they are of the order stated in Theorem~\ref{thm:PILLAR_convergence}.

\subsection{Auxiliary lemmas}
\label{appen:subsec:lemmas}

The first lemma extracts from Assumption~\ref{ass:rsc_rss} the two consequences used in the sequel: strong monotonicity of the gradient map between low-rank matrices, and a Lipschitz property of the gradient in rank-restricted directions. The lemma is stated under the two-sided bound \eqref{eq:rsc_rss} for pairs in $\Xi(2r)$, which is implied by Assumption~\ref{ass:rsc_rss}: any two matrices in $\Xi(2r)$ differ by a matrix of rank at most $4r$, hence of nuclear-to-Frobenius ratio at most $2\sqrt r\le4\sqrt{2r}$, so their difference lies in $\mathcal A$. In particular, the analysis of the iterations uses \eqref{eq:rsc_rss} only along differences of rank at most $4r$.

\begin{lemma}[Restricted strong monotonicity and restricted gradient Lipschitz bound]
\label{lem:appen:monotone_lipschitz}
Let $g:\RR^{d_1\times d_2}\to\RR$ be differentiable and satisfy, for all $\Theta_1,\Theta_2\in\Xi(2r)$,
\[
\frac{\mu}{2}\|\Theta_2-\Theta_1\|_F^2
\le g(\Theta_2)-g(\Theta_1)-\langle\nabla g(\Theta_1),\Theta_2-\Theta_1\rangle
\le\frac{\beta_G}{2}\|\Theta_2-\Theta_1\|_F^2 .
\]
Then, for all $\Theta,\Theta'\in\Xi(r)$,
\begin{align}
\langle\nabla g(\Theta)-\nabla g(\Theta'),\,\Theta-\Theta'\rangle
&\ge\mu\|\Theta-\Theta'\|_F^2,
\label{eq:appen:monotone}\\
\|\nabla g(\Theta)-\nabla g(\Theta')\|_{F,r}
&\le\beta_G\|\Theta-\Theta'\|_F .
\label{eq:appen:res_lip}
\end{align}
\end{lemma}

\begin{proof}
Both $\Theta,\Theta'\in\Xi(r)\subset\Xi(2r)$, so the assumed two-sided bound applies to the pair. Adding the lower bound written at $(\Theta_1,\Theta_2)=(\Theta',\Theta)$ to the same bound written at $(\Theta_1,\Theta_2)=(\Theta,\Theta')$ cancels the function values and yields \eqref{eq:appen:monotone}.

For \eqref{eq:appen:res_lip}, fix $\Delta\in\Xi(r)$ with $\|\Delta\|_F\le1$ and $t>0$, and note that $\Theta'+t\Delta\in\Xi(2r)$. Applying the lower bound to the pair $(\Theta,\Theta'+t\Delta)$ and dropping the nonnegative quadratic term,
\[
g(\Theta'+t\Delta)\ \ge\ g(\Theta)+\langle\nabla g(\Theta),\Theta'-\Theta\rangle+t\,\langle\nabla g(\Theta),\Delta\rangle,
\]
while the upper bound applied to the pair $(\Theta',\Theta'+t\Delta)$ gives
\[
g(\Theta'+t\Delta)\ \le\ g(\Theta')+t\,\langle\nabla g(\Theta'),\Delta\rangle+\frac{\beta_G}{2}t^2 .
\]
Combining the two displays and then using the upper bound once more, now for the pair $(\Theta,\Theta')$,
\[
t\,\langle\nabla g(\Theta)-\nabla g(\Theta'),\Delta\rangle
\ \le\ g(\Theta')-g(\Theta)-\langle\nabla g(\Theta),\Theta'-\Theta\rangle+\frac{\beta_G}{2}t^2
\ \le\ \frac{\beta_G}{2}\|\Theta-\Theta'\|_F^2+\frac{\beta_G}{2}t^2 .
\]
If $\Theta\neq\Theta'$, dividing by $t$ and taking $t=\|\Theta-\Theta'\|_F$ gives $\langle\nabla g(\Theta)-\nabla g(\Theta'),\Delta\rangle\le\beta_G\|\Theta-\Theta'\|_F$; the case $\Theta=\Theta'$ is trivial. Taking the supremum over $\Delta$ yields \eqref{eq:appen:res_lip}.
\end{proof}

The restricted norm \eqref{eq:appen:restricted_norm} interacts with the factors through the following elementary bound.

\begin{lemma}[Restricted norm and factored products]
\label{lem:appen:restricted_product}
For any $A\in\RR^{d_1\times d_2}$, $U\in\RR^{d_1\times r}$ and $V\in\RR^{d_2\times r}$,
\[
\|AV\|_F\le\|A\|_{F,r}\,\|V\|_2,
\qquad
\|A^\top U\|_F\le\|A\|_{F,r}\,\|U\|_2 .
\]
\end{lemma}

\begin{proof}
Let $\Pi\in\RR^{d_2\times d_2}$ be the orthogonal projection onto the column space of $V$, so that $V=\Pi V$ and $\rank(A\Pi)\le\rank(\Pi)\le r$. Then $AV=(A\Pi)V$, whence $\|AV\|_F\le\|A\Pi\|_F\|V\|_2$. If $A\Pi\neq0$, then $\|A\Pi\|_F=\langle A\Pi,A\Pi\rangle/\|A\Pi\|_F=\langle A,\,A\Pi/\|A\Pi\|_F\rangle\le\|A\|_{F,r}$, because $\Pi$ is symmetric and idempotent and $A\Pi/\|A\Pi\|_F$ belongs to the unit Frobenius ball of $\Xi(r)$. The second inequality follows by applying the first to $A^\top$ and noting that $\|A^\top\|_{F,r}=\|A\|_{F,r}$.
\end{proof}

The next three lemmas concern the geometry of the lifted matrix.

\begin{lemma}[From factor error to matrix error]
\label{lemma:thetatoZ}
For $\Theta=UV^\top$ and $\Theta^*=U^*(V^*)^\top$ with stacked matrices $Z=(U;V)$ and $Z^*=(U^*;V^*)$,
\begin{equation}
\label{eq:appen:thetatoZ}
\|\Theta-\Theta^*\|_F^2\le2\max\{\|Z\|_2^2,\|Z^*\|_2^2\}\;d^2(Z,Z^*).
\end{equation}
\end{lemma}

\begin{proof}
With $\widehat O$ the optimal rotation, $\Theta^*=U^*\widehat O(V^*\widehat O)^\top$ and
\begin{align*}
\|\Theta-\Theta^*\|_F
&=\bigl\|U(V-V^*\widehat O)^\top+(U-U^*\widehat O)(V^*\widehat O)^\top\bigr\|_F\\
&\le\|U\|_2\|\Delta_V\|_F+\|V^*\|_2\|\Delta_U\|_F
\le\max\{\|Z\|_2,\|Z^*\|_2\}\bigl(\|\Delta_U\|_F+\|\Delta_V\|_F\bigr),
\end{align*}
using $\|AB\|_F\le\|A\|_2\|B\|_F$, $\|U\|_2\le\|Z\|_2$ and $\|V^*\|_2\le\|Z^*\|_2$. Squaring and using $(a+b)^2\le2(a^2+b^2)$ gives \eqref{eq:appen:thetatoZ}.
\end{proof}

The following two results are Lemma~5.4 and Lemma~5.14 of \citet{tu2016low}, restated in our notation.

\begin{lemma}[{\citealp[Lemma~5.4]{tu2016low}}]
\label{lem:lift_geometry}
Let $Z^*\in\RR^{(d_1+d_2)\times r}$ have full column rank. Then, for any $Z\in\RR^{(d_1+d_2)\times r}$,
\begin{equation}
\label{eq:lift_lower}
d^2(Z,Z^*)\le\frac{1}{2(\sqrt2-1)\,\sigma_r^2(Z^*)}\,\bigl\|ZZ^\top-Z^*(Z^*)^\top\bigr\|_F^2
\le\frac{5}{4\,\sigma_r^2(Z^*)}\,\bigl\|ZZ^\top-Z^*(Z^*)^\top\bigr\|_F^2 .
\end{equation}
\end{lemma}

\begin{lemma}[{\citealp[Lemma~5.14]{tu2016low}}]
\label{lemma:Ztotheta}
Let $\Theta,\Theta^*\in\RR^{d_1\times d_2}$ be rank-$r$ matrices with rank-$r$ singular value decompositions $\Theta=\tilde U\Sigma\tilde V^\top$ and $\Theta^*=\tilde U^*\Sigma^*(\tilde V^*)^\top$, and let $Z=(\tilde U\Sigma^{1/2};\tilde V\Sigma^{1/2})$ and $Z^*=(\tilde U^*(\Sigma^*)^{1/2};\tilde V^*(\Sigma^*)^{1/2})$ be the corresponding balanced stacks. If $\|\Theta-\Theta^*\|_F\le\tfrac12\sigma_{\min}(\Theta^*)$, then
\begin{equation}
\label{eq:Ztotheta}
d^2(Z,Z^*)\le\frac{2}{\sqrt2-1}\cdot\frac{\|\Theta-\Theta^*\|_F^2}{\sigma_{\min}(\Theta^*)}
\le\frac{5\,\|\Theta-\Theta^*\|_F^2}{\sigma_{\min}(\Theta^*)}.
\end{equation}
\end{lemma}

The next lemma is an exact algebraic identity for the balancing-penalty term and makes the subsequent combination step transparent. Recall $M=U^\top U-V^\top V$ and $\nabla_Zp(Z)=\breve ZM$.

\begin{lemma}[Balancing penalty identity]
\label{lem:balance_geometry}
Assume $(U^*)^\top U^*=(V^*)^\top V^*$. Then
\begin{equation}
\label{eq:balance_inner}
\bigl\langle\nabla_Zp(Z),\,\Delta_Z\bigr\rangle
=\frac12\|M\|_F^2+\frac12\bigl\langle M,\;\Delta_U^\top\Delta_U-\Delta_V^\top\Delta_V\bigr\rangle
\ \ge\ \frac12\|M\|_F^2-\frac12\|M\|_F\,d^2(Z,Z^*).
\end{equation}
\end{lemma}

\begin{proof}
Write $D_\Delta:=\Delta_U^\top\Delta_U-\Delta_V^\top\Delta_V$. Substituting $U=U^*\widehat O+\Delta_U$ and $V=V^*\widehat O+\Delta_V$ into $M$ and using $(U^*)^\top U^*-(V^*)^\top V^*=0$,
\begin{equation}
\label{eq:appen:M_expand}
M=\widehat O^\top N+N^\top\widehat O+D_\Delta,
\qquad
N=(U^*)^\top\Delta_U-(V^*)^\top\Delta_V .
\end{equation}
Next, $\langle\nabla_Zp(Z),\Delta_Z\rangle=\langle\breve ZM,\Delta_Z\rangle=\operatorname{tr}\bigl(M\,\breve Z^\top\Delta_Z\bigr)$, where we used that $M$ is symmetric, and
\[
\breve Z^\top\Delta_Z=U^\top\Delta_U-V^\top\Delta_V
=\widehat O^\top N+D_\Delta .
\]
Since $\operatorname{tr}(MA)=\operatorname{tr}(MA^\top)$ for symmetric $M$, we may replace $\widehat O^\top N$ by its symmetric part $\frac12(\widehat O^\top N+N^\top\widehat O)$, which equals $\frac12(M-D_\Delta)$ by \eqref{eq:appen:M_expand}. Hence
\[
\bigl\langle\nabla_Zp(Z),\Delta_Z\bigr\rangle
=\Bigl\langle M,\;\frac12(M-D_\Delta)+D_\Delta\Bigr\rangle
=\frac12\|M\|_F^2+\frac12\langle M,D_\Delta\rangle,
\]
which is the claimed identity. The inequality follows from
$|\langle M,D_\Delta\rangle|\le\|M\|_F\bigl(\|\Delta_U^\top\Delta_U\|_F+\|\Delta_V^\top\Delta_V\|_F\bigr)\le\|M\|_F\bigl(\|\Delta_U\|_F^2+\|\Delta_V\|_F^2\bigr)=\|M\|_F\,d^2(Z,Z^*)$.
\end{proof}

\begin{lemma}[Gram coercivity]
\label{lem:appen:gram}
Assume $(U^*)^\top U^*=(V^*)^\top V^*$. Then
\begin{equation}
\label{eq:appen:gram_identity}
\bigl\|ZZ^\top-Z^*(Z^*)^\top\bigr\|_F^2
=4\|\Theta-\Theta^*\|_F^2+\|M\|_F^2-2\|N\|_F^2,
\end{equation}
where $N$ is as in \eqref{eq:appen:M_expand}. Consequently, with $\sigma_r^2(Z^*)=2\sigma_{\min}$,
\begin{equation}
\label{eq:appen:gram_coercive}
\|\Theta-\Theta^*\|_F^2
\ \ge\ \frac{2}{5}\,\sigma_{\min}\,d^2(Z,Z^*)-\frac14\|M\|_F^2 .
\end{equation}
\end{lemma}

\begin{proof}
Write $X:=\|UU^\top-U^*(U^*)^\top\|_F^2$, $Y:=\|VV^\top-V^*(V^*)^\top\|_F^2$ and $W:=\|UV^\top-U^*(V^*)^\top\|_F^2$, so that, reading off the four blocks of $ZZ^\top-Z^*(Z^*)^\top$,
\[
\bigl\|ZZ^\top-Z^*(Z^*)^\top\bigr\|_F^2=X+Y+2W .
\]
Expanding the traces, with $P:=(U^*)^\top U$ and $Q:=(V^*)^\top V$,
\begin{align*}
X&=\|U^\top U\|_F^2-2\|P\|_F^2+\|(U^*)^\top U^*\|_F^2,\qquad
Y=\|V^\top V\|_F^2-2\|Q\|_F^2+\|(V^*)^\top V^*\|_F^2,\\
W&=\langle U^\top U,V^\top V\rangle-2\langle P,Q\rangle+\langle(U^*)^\top U^*,(V^*)^\top V^*\rangle,
\end{align*}
where we used $\operatorname{tr}(UV^\top V^*(U^*)^\top)=\operatorname{tr}\{((U^*)^\top U)(V^\top V^*)\}=\langle P,Q\rangle$ for both cross terms. Hence
\[
X+Y-2W=\|U^\top U-V^\top V\|_F^2+\bigl\|(U^*)^\top U^*-(V^*)^\top V^*\bigr\|_F^2-2\|P-Q\|_F^2
=\|M\|_F^2-2\|P-Q\|_F^2,
\]
using the balancedness of $Z^*$. Finally, $P-Q=(U^*)^\top U-(V^*)^\top V=(U^*)^\top\Delta_U-(V^*)^\top\Delta_V=N$, again by balancedness, and \eqref{eq:appen:gram_identity} follows from $X+Y+2W=(X+Y-2W)+4W$.

For \eqref{eq:appen:gram_coercive}, Lemma~\ref{lem:lift_geometry} together with $\sigma_r^2(Z^*)=2\sigma_{\min}$ gives $\tfrac{8}{5}\sigma_{\min}\,d^2(Z,Z^*)\le\|ZZ^\top-Z^*(Z^*)^\top\|_F^2$, and hence by \eqref{eq:appen:gram_identity},
\[
\frac{8}{5}\,\sigma_{\min}\,d^2(Z,Z^*)
\ \le\ 4\|\Theta-\Theta^*\|_F^2+\|M\|_F^2-2\|N\|_F^2
\ \le\ 4\|\Theta-\Theta^*\|_F^2+\|M\|_F^2 .
\]
Dividing by $4$ and rearranging yields \eqref{eq:appen:gram_coercive}.
\end{proof}

Finally, we record the first-order optimality property of the dual component of the saddle point; it is the only property of $w^*$ used in the mirror-ascent analysis, and it requires no interiority of $w^*$.

\begin{lemma}[Dual first-order condition]
\label{lem:appen:dual_foc}
Let $w^*=w^*_\lambda$ be as in \eqref{eq:PILLAR_def}. Then
\begin{equation}
\label{eq:appen:dual_foc}
\bigl\langle L(\Theta^*)-\lambda\nabla D(w^*),\,u-w^*\bigr\rangle\le0
\qquad\text{for all }u\in\Delta_m .
\end{equation}
\end{lemma}

\begin{proof}
The map $w\mapsto\langle w,L(\Theta^*)\rangle-\lambda D(w)$ is concave and differentiable on $\Delta_m$ and is maximized at $w^*$; \eqref{eq:appen:dual_foc} is the standard variational inequality for a maximizer over a convex set.
\end{proof}

\subsection{Local neighborhood and initialization}
\label{appen:subsec:neighborhood}

Because the factorized problem is nonconvex, the analysis is local. Define
\begin{equation}
\label{eq:appen:mu_tilde}
\tilde\mu:=\min\{\mu,\,4\},
\qquad
C_{\mathrm{loc}}:=\frac{\beta_G^2}{2\mu}+\frac{\mu}{16}+\frac12,
\end{equation}
and the neighborhood radius
\begin{equation}
\label{eq:appen:I0}
I_0^2:=\min\left\{\frac{\tilde\mu\,\sigma_{\min}}{20\,C_{\mathrm{loc}}},\ \frac{\sigma_{\min}}{32},\ \frac{\varrho^2}{6\,\sigma_{\max}}\right\},
\end{equation}
where $\varrho$ is the radius in Assumption~\ref{ass:rsc_rss}. (If $\varrho$ can be taken of the order $\sigma_{\min}^{1/2}$ or larger, the third branch is never active; it is retained only so that the gradient bound $\|\nabla\Lkh\|_F\le\beta_L$ of Assumption~\ref{ass:rsc_rss} is available along the iterates.) The analysis is carried out on the event that the current iterate satisfies the neighborhood condition
\begin{equation}
\label{eq:local_neighborhood}
d(Z,Z^*)\le I_0 ,
\end{equation}
which will be propagated by induction. Three consequences of \eqref{eq:local_neighborhood} are used repeatedly.

\begin{lemma}[Consequences of the neighborhood condition]
\label{lemma:ZandZ*}
Suppose $d(Z,Z^*)\le I_0$. Then
\begin{enumerate}
\item[(i)] $\|Z\|_2\le\tfrac98\|Z^*\|_2$, and hence $2\max\{\|Z\|_2^2,\|Z^*\|_2^2\}\le6\sigma_{\max}$;
\item[(ii)] $\|\Theta-\Theta^*\|_F^2\le6\,\sigma_{\max}\,d^2(Z,Z^*)\le\varrho^2$, so in particular $\Theta=UV^\top$ lies in the ball of Assumption~\ref{ass:rsc_rss};
\item[(iii)] if in addition $Z^0$ satisfies $d(Z^0,Z^*)\le I_0$, then $\|Z\|_2\le\sqrt2\,\|Z^0\|_2$.
\end{enumerate}
\end{lemma}

\begin{proof}
(i) Since $I_0^2\le\sigma_{\min}/32\le\sigma_{\max}/32=\|Z^*\|_2^2/64$, the triangle inequality gives
$\|Z\|_2\le\|Z^*\widehat O\|_2+\|\Delta_Z\|_2\le\|Z^*\|_2+I_0\le\tfrac98\|Z^*\|_2$. Consequently
$2\max\{\|Z\|_2^2,\|Z^*\|_2^2\}\le2\cdot\tfrac{81}{64}\|Z^*\|_2^2=\tfrac{81}{16}\,\sigma_{\max}\le6\sigma_{\max}$.
(ii) Combine (i) with Lemma~\ref{lemma:thetatoZ} and $I_0^2\le\varrho^2/(6\sigma_{\max})$.
(iii) As in (i), $\|Z^0\|_2\ge\|Z^*\|_2-I_0\ge\tfrac78\|Z^*\|_2$, so $\|Z\|_2\le\tfrac98\cdot\tfrac87\|Z^0\|_2\le\sqrt2\,\|Z^0\|_2$.
\end{proof}

The initialization condition of Assumption~\ref{ass:initialization} places the algorithm inside this neighborhood. We take the constants of Assumption~\ref{ass:initialization} to be
\begin{equation}
\label{eq:appen:ITheta_Iw}
I_\Theta:=\min\Bigl\{\frac{1}{16},\ \frac{c_{I_0}}{40}\Bigr\},
\qquad
I_w:=\frac{\sqrt{\mu\tilde\mu}}{14}\,c_{I_0},
\qquad\text{where}\quad
c_{I_0}:=\frac{I_0^2}{\sigma_{\min}} .
\end{equation}

\begin{lemma}[Accuracy of $Z^0$]
\label{lem:appen:Z0}
Suppose $\Theta^0$ satisfies the first condition in \eqref{eq:init_assum} with $I_\Theta$ as in \eqref{eq:appen:ITheta_Iw}. Let $\Theta^0_r=\tilde U\tilde\Sigma\tilde V^\top$ be the rank-$r$ truncated singular value decomposition of $\Theta^0$ computed in Algorithm~\ref{alg:BGDMA}, and set $U^0=\tilde U\tilde\Sigma^{1/2}$, $V^0=\tilde V\tilde\Sigma^{1/2}$ and $Z^0=(U^0;V^0)$. Then
\begin{equation}
\label{eq:appen:Z0_accuracy}
d^2(Z^0,Z^*)\le20\,I_\Theta\,\sigma_{\min}\min\{1,\sigma_{\min}\}\le\frac{I_0^2}{2}.
\end{equation}
\end{lemma}

\begin{proof}
Since $\Theta^0_r$ is the best rank-$r$ approximation of $\Theta^0$ in Frobenius norm and $\rank(\Theta^*)\le r$,
\[
\|\Theta^0_r-\Theta^*\|_F\le\|\Theta^0_r-\Theta^0\|_F+\|\Theta^0-\Theta^*\|_F\le2\|\Theta^0-\Theta^*\|_F .
\]
By \eqref{eq:init_assum} and $I_\Theta\le1/16$, $\|\Theta^0_r-\Theta^*\|_F\le2\sqrt{I_\Theta}\,\sigma_{\min}\le\tfrac12\sigma_{\min}$, so Lemma~\ref{lemma:Ztotheta} applies to the pair $(\Theta^0_r,\Theta^*)$ and yields
\[
d^2(Z^0,Z^*)\le\frac{5\|\Theta^0_r-\Theta^*\|_F^2}{\sigma_{\min}}
\le\frac{20\,\|\Theta^0-\Theta^*\|_F^2}{\sigma_{\min}}
\le20\,I_\Theta\,\sigma_{\min}\min\{1,\sigma_{\min}\}.
\]
Finally $20I_\Theta\sigma_{\min}\min\{1,\sigma_{\min}\}\le20I_\Theta\sigma_{\min}\le\tfrac12 c_{I_0}\sigma_{\min}=I_0^2/2$ by \eqref{eq:appen:ITheta_Iw}.
\end{proof}

\subsection{One-step contraction for the factor update}
\label{appen:subsec:onestepZ}

Fix an iteration and write $(Z,w)$ for the current iterate and $Z^+$ for the next factor iterate; thus $\Theta=UV^\top$ and, by \eqref{eq:UV_update},
\[
Z^+=Z-\etat\,\nabla_Z\widehat f(U,V,w),
\qquad
\nabla_Z\widehat f(U,V,w)=\widehat G(Z,w)+\nabla_Zp(Z),
\]
where, by the chain rule,
\begin{equation}
\label{eq:appen:chain}
\widehat G(Z,w)
=\begin{pmatrix}\nabla_\Theta\widehat L(\Theta,w)\,V\\[2pt]\nabla_\Theta\widehat L(\Theta,w)^\top U\end{pmatrix},
\qquad
\nabla_Zp(Z)=\breve ZM .
\end{equation}

\begin{lemma}[One-step contraction for $Z$]
\label{lem:appen:onestep_Z}
Suppose Assumptions~\ref{ass:rsc_rss}--\ref{ass:w_regularizer_mirror} hold, the current iterate satisfies $d(Z,Z^*)\le I_0$ with $I_0$ as in \eqref{eq:appen:I0}, and the step size satisfies
\begin{equation}
\label{eq:stepsizes_revised}
\etat\ \le\ \frac{1}{16\,\|Z\|_2^2}\,\min\Bigl\{\frac{2\mu}{\beta_G^2},\,1\Bigr\} .
\end{equation}
Then
\begin{equation}
\label{eq:appen:Z_update_final}
d^2(Z^+,Z^*)
\ \le\ \Bigl(1-\etat\,\frac{\tilde\mu\,\sigma_{\min}}{10}\Bigr)\,d^2(Z,Z^*)
\ +\ \frac{14\,\etat}{\mu}\Bigl(\estatt^2+\beta_L^2\,\|w-w^*\|_1^2\Bigr).
\end{equation}
\end{lemma}

\begin{proof}
Throughout the proof write $d^2:=d^2(Z,Z^*)$, $W:=\|\Theta-\Theta^*\|_F^2$, $\delta_w:=\|w-w^*\|_1$, and
\[
E:=\estatt+\beta_L\,\delta_w .
\]
Using the definition of $d(\cdot,\cdot)$ and the suboptimality of the rotation $\widehat O$,
\begin{align}
d^2(Z^+,Z^*)
&\le\|Z^+-Z^*\widehat O\|_F^2
=\bigl\|\Delta_Z-\etat\nabla_Z\widehat f(U,V,w)\bigr\|_F^2 \notag\\
&=d^2-2\etat\bigl\langle\nabla_Z\widehat f(U,V,w),\Delta_Z\bigr\rangle+\etat^2\bigl\|\nabla_Z\widehat f(U,V,w)\bigr\|_F^2 \notag\\
&\le d^2-2\etat(T_1+R_1)+2\etat^2(T_2+R_2),
\label{eq:split_T_R}
\end{align}
where, recalling \eqref{eq:appen:chain} and using $\|a+b\|_F^2\le2\|a\|_F^2+2\|b\|_F^2$,
\[
T_1:=\bigl\langle\widehat G(Z,w),\Delta_Z\bigr\rangle,\quad
R_1:=\bigl\langle\nabla_Zp(Z),\Delta_Z\bigr\rangle,\quad
T_2:=\bigl\|\widehat G(Z,w)\bigr\|_F^2,\quad
R_2:=\bigl\|\nabla_Zp(Z)\bigr\|_F^2 .
\]

\medskip\noindent\textbf{Step 1: lower bound for $T_1$.}
By \eqref{eq:appen:chain},
\begin{equation}
T_1=\bigl\langle\nabla_\Theta\widehat L(\Theta,w),\,\Delta_UV^\top+U\Delta_V^\top\bigr\rangle .
\label{eq:T1_chain}
\end{equation}
Using $V=\Delta_V+V^*\widehat O$ and $U=\Delta_U+U^*\widehat O$,
\begin{align}
\Delta_UV^\top+U\Delta_V^\top
&=\Delta_U\Delta_V^\top+\Delta_U(V^*\widehat O)^\top+U\Delta_V^\top \notag\\
&=\Delta_U\Delta_V^\top+(U-U^*\widehat O)(V^*\widehat O)^\top+U(V-V^*\widehat O)^\top \notag\\
&=\Delta_U\Delta_V^\top+UV^\top-U^*(V^*)^\top
=\Delta_U\Delta_V^\top+\Theta-\Theta^* .
\label{eq:decomp_matrix}
\end{align}
Plugging \eqref{eq:decomp_matrix} into \eqref{eq:T1_chain} and decomposing the weighted gradient around $(\Theta^*,w^*)$ yields $T_1=T_{11}+T_{12}+T_{13}$ with
\begin{align*}
T_{11}&:=\bigl\langle\nabla_\Theta\widehat L(\Theta,w)-\nabla_\Theta\widehat L(\Theta^*,w),\,\Theta-\Theta^*\bigr\rangle,\\
T_{12}&:=\bigl\langle\nabla_\Theta\widehat L(\Theta^*,w),\,\Theta-\Theta^*\bigr\rangle,\qquad
T_{13}:=\bigl\langle\nabla_\Theta\widehat L(\Theta,w),\,\Delta_U\Delta_V^\top\bigr\rangle .
\end{align*}

\emph{Term $T_{11}$.} Both $\Theta$ and $\Theta^*$ lie in $\Xi(r)$, so \eqref{eq:appen:monotone} of Lemma~\ref{lem:appen:monotone_lipschitz}, applied to each $\Lkh$ and averaged with the weights $w_k\ge0$, $\sum_kw_k=1$, gives
\begin{equation}
\label{eq:T11_bound}
T_{11}\ \ge\ \mu\,W .
\end{equation}

\emph{Term $T_{12}$.} Write $\nabla_\Theta\widehat L(\Theta^*,w)=\nabla_\Theta\widehat L(\Theta^*,w^*)+\sum_k(w_k-w_k^*)\nabla_\Theta\Lkh(\Theta^*)$. Since $\Theta-\Theta^*\in\Xi(2r)$, Definition~\ref{def:stat_error} gives $|\langle\nabla_\Theta\widehat L(\Theta^*,w^*),\Theta-\Theta^*\rangle|\le\estatt\sqrt W$, while the gradient bound of Assumption~\ref{ass:rsc_rss} and H\"older's inequality give
\[
\Bigl|\sum_k(w_k-w_k^*)\bigl\langle\nabla_\Theta\Lkh(\Theta^*),\Theta-\Theta^*\bigr\rangle\Bigr|
\le\beta_L\,\delta_w\,\sqrt W .
\]
Hence
\begin{equation}
\label{eq:T12_bound}
T_{12}\ \ge\ -E\,\sqrt W .
\end{equation}

\emph{Term $T_{13}$.} The matrix $\Delta_U\Delta_V^\top$ has rank at most $r$ and $\|\Delta_U\Delta_V^\top\|_F\le\|\Delta_U\|_F\|\Delta_V\|_F\le d^2/2$. Decomposing $\nabla_\Theta\widehat L(\Theta,w)$ as in the previous two steps and using \eqref{eq:appen:res_lip} for the first piece, Definition~\ref{def:stat_error} for the second, and the gradient bound for the third,
\begin{equation}
\label{eq:T13_bound}
|T_{13}|
\le\Bigl(\bigl\|\nabla_\Theta\widehat L(\Theta,w)-\nabla_\Theta\widehat L(\Theta^*,w)\bigr\|_{F,r}+\estatt+\beta_L\delta_w\Bigr)\frac{d^2}{2}
\le\bigl(\beta_G\sqrt W+E\bigr)\frac{d^2}{2} ,
\end{equation}
where the restricted Lipschitz bound \eqref{eq:appen:res_lip} was applied to each $\Lkh$ and averaged.

Combining \eqref{eq:T11_bound}--\eqref{eq:T13_bound} and applying Young's inequality $ab\le a^2/(2\epsilon)+\epsilon b^2/2$ three times --- to $E\sqrt W$ with $\epsilon=\mu/4$, to $(\beta_Gd^2/2)\sqrt W$ with $\epsilon=\mu/4$, and to $E\cdot(d^2/2)$ with $\epsilon=\mu/2$ --- we obtain, using also $d^2\le I_0^2$,
\begin{align}
T_1
&\ge\mu W-E\sqrt W-\bigl(\beta_G\sqrt W+E\bigr)\frac{d^2}{2} \notag\\
&\ge\mu W-\Bigl(\frac{2E^2}{\mu}+\frac{\mu}{8}W\Bigr)
-\Bigl(\frac{\beta_G^2d^4}{2\mu}+\frac{\mu}{8}W\Bigr)
-\Bigl(\frac{E^2}{\mu}+\frac{\mu\,d^4}{16}\Bigr) \notag\\
&\ge\frac{3\mu}{4}\,W-\frac{3E^2}{\mu}
-\Bigl(\frac{\beta_G^2}{2\mu}+\frac{\mu}{16}\Bigr)I_0^2\,d^2 .
\label{eq:T1_final_lower}
\end{align}

\medskip\noindent\textbf{Step 2: lower bound for $R_1$ and combination with the curvature term.}
Lemma~\ref{lem:balance_geometry} and Young's inequality $\tfrac12ab\le\tfrac18a^2+\tfrac12b^2$ (with $a=\|M\|_F$, $b=d^2$) give
\begin{equation}
\label{eq:appen:R1}
R_1\ \ge\ \frac12\|M\|_F^2-\frac12\|M\|_F\,d^2
\ \ge\ \frac38\|M\|_F^2-\frac12\,I_0^2\,d^2 .
\end{equation}
Split the curvature term from \eqref{eq:T1_final_lower} as $\tfrac{3\mu}{4}W=\tfrac{\mu}{2}W+\tfrac{\mu}{4}W$ and recall $\tilde\mu=\min\{\mu,4\}\le\mu$. Applying the Gram coercivity bound \eqref{eq:appen:gram_coercive} to the fraction $\tfrac{\tilde\mu}{4}W\le\tfrac{\mu}{4}W$,
\begin{equation}
\label{eq:appen:R11}
\frac{\mu}{4}\,W\ \ge\ \frac{\tilde\mu}{4}\,W
\ \ge\ \frac{\tilde\mu}{10}\,\sigma_{\min}\,d^2-\frac{\tilde\mu}{16}\|M\|_F^2
\ \ge\ \frac{\tilde\mu}{10}\,\sigma_{\min}\,d^2-\frac14\|M\|_F^2 ,
\end{equation}
where the last step uses $\tilde\mu\le4$. Adding \eqref{eq:T1_final_lower}, \eqref{eq:appen:R1} and \eqref{eq:appen:R11}, and using the definition \eqref{eq:appen:mu_tilde} of $C_{\mathrm{loc}}$ together with $C_{\mathrm{loc}}I_0^2\le\tilde\mu\sigma_{\min}/20$ from \eqref{eq:appen:I0},
\begin{equation}
\label{eq:T1R1_bound}
T_1+R_1
\ \ge\ \frac{\tilde\mu\,\sigma_{\min}}{20}\,d^2
+\frac{\mu}{2}\,W+\frac18\|M\|_F^2-\frac{3E^2}{\mu} .
\end{equation}

\medskip\noindent\textbf{Step 3: upper bounds for $T_2$ and $R_2$.}
By \eqref{eq:appen:chain} and Lemma~\ref{lem:appen:restricted_product},
\begin{align*}
T_2&=\bigl\|\nabla_\Theta\widehat L(\Theta,w)V\bigr\|_F^2+\bigl\|\nabla_\Theta\widehat L(\Theta,w)^\top U\bigr\|_F^2\\
&\le\bigl(\|U\|_2^2+\|V\|_2^2\bigr)\,\bigl\|\nabla_\Theta\widehat L(\Theta,w)\bigr\|_{F,r}^2
\le2\|Z\|_2^2\,\bigl\|\nabla_\Theta\widehat L(\Theta,w)\bigr\|_{F,r}^2 .
\end{align*}
Decomposing $\nabla_\Theta\widehat L(\Theta,w)$ as in Step~1 and using the triangle inequality for $\|\cdot\|_{F,r}$ together with \eqref{eq:appen:res_lip},
\begin{equation}
\label{eq:T2_bound}
T_2\le2\|Z\|_2^2\bigl(\beta_G\sqrt W+E\bigr)^2\le4\|Z\|_2^2\bigl(\beta_G^2W+E^2\bigr).
\end{equation}
For the penalty part, $\nabla_Zp(Z)=\breve ZM$ and $\|\breve Z\|_2=\|Z\|_2$, so
\begin{equation}
\label{eq:R2_bound}
R_2\le\|Z\|_2^2\|M\|_F^2\le2\|Z\|_2^2\|M\|_F^2 .
\end{equation}

\medskip\noindent\textbf{Step 4: conclusion.}
Substituting \eqref{eq:T1R1_bound}, \eqref{eq:T2_bound} and \eqref{eq:R2_bound} into \eqref{eq:split_T_R},
\begin{align}
d^2(Z^+,Z^*)
&\le\Bigl(1-\etat\,\frac{\tilde\mu\sigma_{\min}}{10}\Bigr)d^2
+\etat\bigl(8\etat\|Z\|_2^2\beta_G^2-\mu\bigr)W
+\etat\Bigl(4\etat\|Z\|_2^2-\frac14\Bigr)\|M\|_F^2 \notag\\
&\qquad
+\etat\Bigl(\frac{6}{\mu}+8\etat\|Z\|_2^2\Bigr)E^2 .
\label{eq:Z_step_prestepsize}
\end{align}
The step-size condition \eqref{eq:stepsizes_revised} makes the coefficients of $W$ and of $\|M\|_F^2$ nonpositive, so these two terms may be dropped; it also implies $8\etat\|Z\|_2^2\le\mu/\beta_G^2\le1/\mu$, since $\mu\le\beta_G$. Hence the coefficient of $E^2$ is at most $7\etat/\mu$, and $E^2\le2\estatt^2+2\beta_L^2\delta_w^2$ completes the proof.
\end{proof}

\begin{remark}[A step size computable from $Z^0$]
\label{rem:appen:stepsize}
By Lemma~\ref{lemma:ZandZ*}(iii), $\|Z\|_2^2\le2\|Z^0\|_2^2$ along any trajectory that remains in the neighborhood \eqref{eq:local_neighborhood}. Consequently, the fixed and computable choice
\begin{equation}
\label{eq:appen:stepsize_practical}
\etat\le\frac{1}{32\,\|Z^0\|_2^2}\,\min\Bigl\{\frac{2\mu}{\beta_G^2},\,1\Bigr\}
\end{equation}
guarantees \eqref{eq:stepsizes_revised} at every iteration. Since $\|Z^0\|_2^2\le\tfrac{81}{32}\,\sigma_{\max}$ by Lemma~\ref{lemma:ZandZ*}(i), condition \eqref{eq:appen:stepsize_practical} is in turn implied by
\[
\etat\ \le\ \frac{c_2}{\beta\,\sigma_{\max}(\Theta^*)},
\qquad
c_2
:=
\min\left\{
\frac{\beta}{81}
\min\left(\frac{2\mu}{\beta_G^2},1\right),
\frac{10}{\mu}
\right\},
\]
and $c_2$ depends only on $(\mu,\beta)$. This is the form of the step-size condition quoted in Theorem~\ref{thm:PILLAR_convergence}.
\end{remark}

\subsection{One-step contraction for the mirror ascent update}
\label{appen:subsec:onestepw}

We now analyze the update \eqref{eq:mirror_ascent_general}. It is convenient to work with the (strongly convex) negative dual objective
\[
h_\Theta(w):=\lambda D(w)-w^\top\widehat L(\Theta),
\qquad
g(\Theta,w):=\nabla_wh_\Theta(w)=\lambda\nabla D(w)-\widehat L(\Theta),
\]
so that one mirror-ascent step on $\widehat f$ is one mirror-descent step on $h_\Theta$:
\begin{equation}
\label{eq:appen:mirror}
w^+=\argmin{u\in\Delta_m}\Bigl\{\langle g(\Theta,w),u-w\rangle+\frac{1}{\etaw}B_\phi(u,w)\Bigr\}.
\end{equation}
Since $h_\Theta$ differs from $\lambda D$ by a linear term, Assumption~\ref{ass:w_regularizer_mirror} implies that $h_\Theta$ is $\lambda$-convex and $\lambda\beta_D$-smooth relative to $\phi$:
\begin{equation}
\label{eq:appen:h_rel}
\lambda B_\phi(u,w)\ \le\ h_\Theta(u)-h_\Theta(w)-\langle g(\Theta,w),u-w\rangle\ \le\ \lambda\beta_D B_\phi(u,w),
\qquad u,w\in\Delta_m .
\end{equation}

\begin{lemma}[One-step contraction for $w$]
\label{lem:appen:onestep_w}
Suppose Assumptions~\ref{ass:rsc_rss}--\ref{ass:w_regularizer_mirror} hold, $\etaw\le1/(\lambda\beta_D)$, and the current iterate satisfies $\|\Theta-\Theta^*\|_F\le\varrho$ with $\Theta=UV^\top\in\Xi(r)$. Then
\begin{equation}
\label{eq:appen:wupdate}
B_\phi(w^*,w^+)\ \le\ (1-\lambda\etaw)\,B_\phi(w^*,w)
+\frac{\etaw}{\lambda}\Bigl(\beta_L^2\,\|\Theta-\Theta^*\|_F^2+\estatw^2\Bigr).
\end{equation}
\end{lemma}

\begin{proof}
By the relative smoothness upper bound in \eqref{eq:appen:h_rel} applied at the pair $(w^+,w)$,
\begin{equation}
\label{eq:appen:hsmooth}
h_\Theta(w^+)\le h_\Theta(w)+\bigl\langle g(\Theta,w),w^+-w\bigr\rangle+\lambda\beta_DB_\phi(w^+,w),
\end{equation}
and by the relative convexity lower bound applied at $(u,w)$, for every $u\in\Delta_m$,
\begin{equation}
\label{eq:appen:hconvex}
h_\Theta(w)-h_\Theta(u)\le\bigl\langle g(\Theta,w),w-u\bigr\rangle-\lambda B_\phi(u,w).
\end{equation}
Adding \eqref{eq:appen:hsmooth} and \eqref{eq:appen:hconvex} gives
\begin{equation}
\label{eq:appen:hdiff}
h_\Theta(w^+)-h_\Theta(u)\le\bigl\langle g(\Theta,w),w^+-u\bigr\rangle-\lambda B_\phi(u,w)+\lambda\beta_DB_\phi(w^+,w).
\end{equation}
The first-order optimality condition for \eqref{eq:appen:mirror} states that, for every $u\in\Delta_m$,
\[
\Bigl\langle g(\Theta,w)+\frac{1}{\etaw}\bigl(\nabla\phi(w^+)-\nabla\phi(w)\bigr),\,u-w^+\Bigr\rangle\ge0 .
\]
Combining this with the three-point identity for Bregman divergences \citep{chen1993convergence,beck2003mirror},
\[
\langle\nabla\phi(a)-\nabla\phi(b),c-a\rangle=B_\phi(c,b)-B_\phi(c,a)-B_\phi(a,b),
\]
applied with $(a,b,c)=(w^+,w,u)$, we obtain
\begin{equation}
\label{eq:3point}
\bigl\langle g(\Theta,w),w^+-u\bigr\rangle
\le\frac{1}{\etaw}\Bigl[B_\phi(u,w)-B_\phi(u,w^+)-B_\phi(w^+,w)\Bigr].
\end{equation}
Substituting \eqref{eq:3point} into \eqref{eq:appen:hdiff},
\[
h_\Theta(w^+)-h_\Theta(u)
\le-\frac{1}{\etaw}B_\phi(u,w^+)+\Bigl(\frac{1}{\etaw}-\lambda\Bigr)B_\phi(u,w)
+\Bigl(\lambda\beta_D-\frac{1}{\etaw}\Bigr)B_\phi(w^+,w).
\]
Since $\etaw\le1/(\lambda\beta_D)$, the last term is nonpositive and may be dropped. Taking $u=w^*$ and rearranging,
\begin{equation}
\label{eq:appen:hstep1}
B_\phi(w^*,w^+)\le(1-\lambda\etaw)\,B_\phi(w^*,w)-\etaw\bigl(h_\Theta(w^+)-h_\Theta(w^*)\bigr).
\end{equation}
It remains to bound $-(h_\Theta(w^+)-h_\Theta(w^*))$. By the relative convexity lower bound in \eqref{eq:appen:h_rel}, applied at the pair $(w^+,w^*)$,
\[
-\bigl(h_\Theta(w^+)-h_\Theta(w^*)\bigr)
\le-\bigl\langle g(\Theta,w^*),w^+-w^*\bigr\rangle-\lambda B_\phi(w^+,w^*).
\]
We decompose $-g(\Theta,w^*)=\widehat L(\Theta)-\lambda\nabla D(w^*)$ into three pieces:
\[
-g(\Theta,w^*)
=\bigl\{\widehat L(\Theta)-\widehat L(\Theta^*)\bigr\}
+\bigl\{\widehat L(\Theta^*)-L(\Theta^*)\bigr\}
+\bigl\{L(\Theta^*)-\lambda\nabla D(w^*)\bigr\}.
\]
For the first piece, the segment joining $\Theta^*$ and $\Theta$ lies in the ball $\{\|\Theta'-\Theta^*\|_F\le\varrho\}$, by the hypothesis $\|\Theta-\Theta^*\|_F\le\varrho$ and convexity of the ball; the mean value theorem and the gradient bound of Assumption~\ref{ass:rsc_rss} therefore give $\|\widehat L(\Theta)-\widehat L(\Theta^*)\|_\infty=\max_k|\Lkh(\Theta)-\Lkh(\Theta^*)|\le\beta_L\|\Theta-\Theta^*\|_F$. The second piece has $\ell_\infty$-norm $\estatw$ by Definition~\ref{def:stat_error}. For the third piece, Lemma~\ref{lem:appen:dual_foc} with $u=w^+\in\Delta_m$ gives $\langle L(\Theta^*)-\lambda\nabla D(w^*),w^+-w^*\rangle\le0$. Hence, by H\"older's inequality,
\begin{align}
-\bigl(h_\Theta(w^+)-h_\Theta(w^*)\bigr)
&\le\bigl(\beta_L\|\Theta-\Theta^*\|_F+\estatw\bigr)\|w^+-w^*\|_1-\lambda B_\phi(w^+,w^*) \notag\\
&\le\bigl(\beta_L\|\Theta-\Theta^*\|_F+\estatw\bigr)\|w^+-w^*\|_1-\frac{\lambda}{2}\|w^+-w^*\|_1^2
\label{eq:appen:hyoung}\\
&\le\frac{1}{2\lambda}\bigl(\beta_L\|\Theta-\Theta^*\|_F+\estatw\bigr)^2, \notag
\end{align}
where the second inequality uses the $1$-strong convexity of $\phi$ with respect to $\|\cdot\|_1$ from Assumption~\ref{ass:w_regularizer_mirror}, namely $B_\phi(w^+,w^*)\ge\tfrac12\|w^+-w^*\|_1^2$, and the third maximizes $at-\lambda t^2/2$ over $t\ge0$. Substituting into \eqref{eq:appen:hstep1} and using $(a+b)^2\le2a^2+2b^2$ yields \eqref{eq:appen:wupdate}.
\end{proof}

\subsection{Combining the two recursions: deterministic convergence}
\label{appen:subsec:combine}

We now combine Lemmas~\ref{lem:appen:onestep_Z} and~\ref{lem:appen:onestep_w}. Define the Lyapunov potential
\begin{equation}
\label{eq:appen:potential}
\Psi^t:=d^2(Z^t,Z^*)+K\,B_\phi(w^*,w^t),
\qquad
K:=\frac{56\,\etat\,\beta_L^2}{\mu\,\lambda\,\etaw},
\end{equation}
whose weight $K$ balances the speeds of the two recursions.

\begin{theorem}[Deterministic linear convergence of BGDMA]
\label{thm:appen:deterministic}
Suppose Assumptions~\ref{ass:rsc_rss}, \ref{ass:w_regularizer_mirror} and~\ref{ass:initialization} hold, with the constants $I_\Theta,I_w$ of \eqref{eq:appen:ITheta_Iw} in Assumption~\ref{ass:initialization}, and let $I_0,\tilde\mu,C_{\mathrm{loc}}$ be as in \eqref{eq:appen:mu_tilde}--\eqref{eq:appen:I0}. Assume:
\begin{enumerate}
\item[(C1)] \textup{(regularization)}\quad $\displaystyle\lambda\ \ge\ c_1\sqrt{\kappa^*},\qquad c_1:=\frac{82\,\beta_L^2}{\sqrt{\mu\tilde\mu}}\,;$
\item[(C2)] \textup{(step sizes)}\quad $\displaystyle\etat\le\frac{1}{32\|Z^0\|_2^2}\min\Bigl\{\frac{2\mu}{\beta_G^2},1\Bigr\}
\quad\text{and}\quad
\frac{\etat\,\tilde\mu\,\sigma_{\min}}{10\,\lambda}\ \le\ \etaw\ \le\ \frac{1}{\lambda\beta_D}$
\emph{(the window for $\etaw$ is nonempty whenever $\etat\le10/(\beta_D\tilde\mu\sigma_{\min})$, which can always be arranged by shrinking $\etat$)};
\item[(C3)] \textup{(statistical error)}\quad $\displaystyle
\estatt^2+\frac{4\beta_L^2}{\lambda^2}\,\estatw^2\ \le\ \frac{\mu\,\tilde\mu\,\sigma_{\min}}{280}\,I_0^2 .$
\end{enumerate}
Then, with $\gamma:=\etat\tilde\mu\sigma_{\min}/20\in(0,1)$, the iterates of Algorithm~\ref{alg:BGDMA} satisfy, for every $T\ge0$,
\begin{equation}
\label{eq:appen:psiupdate_T}
d^2(Z^T,Z^*)\ \le\ \Psi^T\ \le\ (1-\gamma)^T\,I_0^2
+\frac{280}{\mu\,\tilde\mu\,\sigma_{\min}}\Bigl(\estatt^2+\frac{4\beta_L^2}{\lambda^2}\estatw^2\Bigr),
\end{equation}
and consequently
\begin{equation}
\label{eq:appen:theta_T}
\bigl\|\widehat\Theta^T-\Theta^*\bigr\|_F^2
\ \le\ 6\,\sigma_{\max}\Bigl[(1-\gamma)^T\,I_0^2
+\frac{280}{\mu\tilde\mu\,\sigma_{\min}}\Bigl(\estatt^2+\frac{4\beta_L^2}{\lambda^2}\estatw^2\Bigr)\Bigr].
\end{equation}
\end{theorem}

\begin{proof}
We first record the one-step inequality for the potential, valid whenever $d(Z,Z^*)\le I_0$.

\medskip\noindent\textbf{Step 1: one-step inequality for $\Psi$.}
Multiply \eqref{eq:appen:wupdate} by $K$ and add \eqref{eq:appen:Z_update_final}:
\begin{align}
\Psi^+
&\le\Bigl(1-\etat\frac{\tilde\mu\sigma_{\min}}{10}\Bigr)d^2
+K(1-\lambda\etaw)B_\phi(w^*,w) \notag\\
&\qquad+\frac{14\etat\beta_L^2}{\mu}\,\|w-w^*\|_1^2
+\frac{K\etaw\beta_L^2}{\lambda}\,\|\Theta-\Theta^*\|_F^2
+\frac{14\etat}{\mu}\estatt^2+\frac{K\etaw}{\lambda}\estatw^2 .
\label{eq:Psi_step1}
\end{align}
By the strong convexity of $\phi$ in Assumption~\ref{ass:w_regularizer_mirror}, $\|w-w^*\|_1^2\le2B_\phi(w^*,w)$, and by Lemma~\ref{lemma:ZandZ*}(ii), $\|\Theta-\Theta^*\|_F^2\le6\sigma_{\max}d^2$. With the choice of $K$ in \eqref{eq:appen:potential},
\[
\frac{14\etat\beta_L^2}{\mu}\cdot2\,B_\phi(w^*,w)=\frac{K\lambda\etaw}{2}\,B_\phi(w^*,w),
\qquad
\frac{K\etaw\beta_L^2}{\lambda}\cdot6\sigma_{\max}
=\frac{336\,\etat\,\beta_L^4\,\sigma_{\max}}{\mu\,\lambda^2}\,,
\]
so that
\begin{equation}
\label{eq:Psi_step2}
\Psi^+
\le\Bigl(1-\etat\frac{\tilde\mu\sigma_{\min}}{10}+\frac{336\etat\beta_L^4\sigma_{\max}}{\mu\lambda^2}\Bigr)d^2
+K\Bigl(1-\frac{\lambda\etaw}{2}\Bigr)B_\phi(w^*,w)
+\frac{14\etat}{\mu}\estatt^2+\frac{56\etat\beta_L^2}{\mu\lambda^2}\estatw^2 .
\end{equation}
Condition (C1) gives $\lambda^2\ge6724\,\beta_L^4\kappa^*/(\mu\tilde\mu)$, hence
\[
\frac{336\,\beta_L^4\,\sigma_{\max}}{\mu\,\lambda^2}
\le\frac{336\,\tilde\mu\,\sigma_{\min}}{6724}
\le\frac{\tilde\mu\,\sigma_{\min}}{20},
\]
so the coefficient of $d^2$ in \eqref{eq:Psi_step2} is at most $1-\etat\tilde\mu\sigma_{\min}/20=1-\gamma$. Moreover, the lower bound on $\etaw$ in (C2) gives $\lambda\etaw/2\ge\etat\tilde\mu\sigma_{\min}/20=\gamma$, so the coefficient of $KB_\phi(w^*,w)$ is also at most $1-\gamma$. Therefore, whenever $d(Z,Z^*)\le I_0$,
\begin{equation}
\label{eq:appen:psiupdate}
\Psi^+\ \le\ (1-\gamma)\,\Psi+\frac{14\etat}{\mu}\Bigl(\estatt^2+\frac{4\beta_L^2}{\lambda^2}\estatw^2\Bigr).
\end{equation}
Note also that $\gamma\le\lambda\etaw/2\le1/(2\beta_D)\le1/2$, since $\beta_D\ge1$ by Assumption~\ref{ass:w_regularizer_mirror}.

\medskip\noindent\textbf{Step 2: the initialization satisfies $\Psi^0\le I_0^2$.}
Lemma~\ref{lem:appen:Z0} gives $d^2(Z^0,Z^*)\le I_0^2/2$. For the dual part, the lower bound on $\etaw$ in (C2) yields
\[
K=\frac{56\,\etat\,\beta_L^2}{\mu\,\lambda\,\etaw}
\ \le\ \frac{56\,\etat\,\beta_L^2}{\mu\,\lambda}\cdot\frac{10\,\lambda}{\etat\,\tilde\mu\,\sigma_{\min}}
=\frac{560\,\beta_L^2}{\mu\,\tilde\mu\,\sigma_{\min}} .
\]
Hence, by the second condition in \eqref{eq:init_assum}, followed by (C1) with $\kappa^*\ge1$ (so that $1/\lambda\le\sqrt{\mu\tilde\mu}/(82\beta_L^2)$) and the definition of $I_w$ in \eqref{eq:appen:ITheta_Iw},
\[
K\,B_\phi(w^*,w^0)
\ \le\ \frac{560\,\beta_L^2}{\mu\tilde\mu\,\sigma_{\min}}\cdot\frac{I_w\,\sigma_{\min}^2}{\lambda}
\ \le\ \frac{560\,\beta_L^2\,I_w\,\sigma_{\min}}{\mu\tilde\mu}\cdot\frac{\sqrt{\mu\tilde\mu}}{82\,\beta_L^2}
\ \le\ \frac{7\,I_w\,\sigma_{\min}}{\sqrt{\mu\tilde\mu}}
\ =\ \frac{c_{I_0}\,\sigma_{\min}}{2}
\ =\ \frac{I_0^2}{2}.
\]
Adding the two parts gives $\Psi^0\le I_0^2$.

\medskip\noindent\textbf{Step 3: induction.}
Suppose $\Psi^t\le I_0^2$ for some $t\ge0$. Then $d(Z^t,Z^*)\le I_0$, so \eqref{eq:appen:psiupdate} applies at step $t$, and by (C3),
\[
\Psi^{t+1}\le(1-\gamma)I_0^2+\frac{14\etat}{\mu}\cdot\frac{\mu\tilde\mu\sigma_{\min}}{280}I_0^2
=(1-\gamma)I_0^2+\frac{\etat\tilde\mu\sigma_{\min}}{20}I_0^2=I_0^2 .
\]
Hence $\Psi^t\le I_0^2$ for all $t$, and \eqref{eq:appen:psiupdate} is valid at every step. Unrolling \eqref{eq:appen:psiupdate} and bounding the geometric series $\sum_{j<T}(1-\gamma)^j\le1/\gamma$,
\[
\Psi^T\le(1-\gamma)^T\Psi^0+\frac{1}{\gamma}\cdot\frac{14\etat}{\mu}\Bigl(\estatt^2+\frac{4\beta_L^2}{\lambda^2}\estatw^2\Bigr)
=(1-\gamma)^T\Psi^0+\frac{280}{\mu\tilde\mu\sigma_{\min}}\Bigl(\estatt^2+\frac{4\beta_L^2}{\lambda^2}\estatw^2\Bigr),
\]
which is \eqref{eq:appen:psiupdate_T}. Finally, $d(Z^T,Z^*)\le I_0$ and Lemma~\ref{lemma:ZandZ*}(ii) convert \eqref{eq:appen:psiupdate_T} into \eqref{eq:appen:theta_T}.
\end{proof}

\subsection{Bounds on the statistical errors}
\label{appen:subsec:estat}

We now bound $\estatt$ and $\estatw$ under Assumption~\ref{ass:loss_psi2}. The gradient error $\estatt$ concentrates at the rate dictated by the effective sample size of the adversarial mixture,
\[
\Delta_n:=\summ\frac{(w_k^*)^2}{n_k}=\frac{1}{n_{\mathrm{eff}}} .
\]
As is standard in $M$-estimation, the centering of the weighted empirical gradient at $\Theta^*$ is supplied by the population first-order condition, Assumption~\ref{ass:pop_foc}. We also use the following two standard facts about sub-Gaussian variables; see, for example, Sections~2.5--2.6 of \citet{vershynin2018high}: there is a universal constant $C_\psi\ge1$ such that
\begin{equation}
\label{eq:appen:psi2_facts}
\PP\bigl(|X|\ge t\bigr)\le2\exp\bigl(-t^2/\|X\|_{\psi_2}^2\bigr),
\qquad
\Bigl\|\sum_i X_i\Bigr\|_{\psi_2}^2\le C_\psi\sum_i\|X_i\|_{\psi_2}^2
\end{equation}
for any random variable $X$ and any independent, mean-zero $X_1,X_2,\dots$ The covering bound below for low-rank matrices is Lemma~3.1 of \citet{candes2011tight}.

\begin{lemma}[{Covering the low-rank sphere; \citealp[Lemma~3.1]{candes2011tight}}]
\label{lem:appen:lowrank_net}
Let $\mathcal S_s:=\{\Delta\in\RR^{d_1\times d_2}:\rank(\Delta)\le s,\ \|\Delta\|_F=1\}$. For any $\varepsilon\in(0,1)$ there exists an $\varepsilon$-net $\bar{\mathcal S}_s\subset\mathcal S_s$ of $\mathcal S_s$ in Frobenius norm with
\[
|\bar{\mathcal S}_s|\ \le\ \Bigl(\frac9\varepsilon\Bigr)^{(d_1+d_2+1)s}.
\]
\end{lemma}

\begin{proposition}[Statistical error rates]
\label{prop:appen:estat}
Suppose Assumptions~\ref{ass:loss_psi2} and~\ref{ass:pop_foc} hold. There exists a universal constant $C>0$ such that, for any $\delta_0\in(0,1)$, each of the following holds with probability at least $1-\delta_0$:
\begin{align}
\estatt&\ \le\ C\,K_G\,\sqrt{\Delta_n}\;\sqrt{r(d_1+d_2)+\log(1/\delta_0)},
\label{eq:appen:estat_theta}\\[2pt]
\estatw&\ \le\ C\,K_L\,\sqrt{\frac{\log(m/\delta_0)}{n_{\min}}} .
\label{eq:appen:estat_w}
\end{align}
\end{proposition}

\begin{proof}
\textbf{Part I: bound for $\estatt$.} By Assumption~\ref{ass:pop_foc},
\[
S:=\summ w_k^*\,\nabla_\Theta\Lkh(\Theta^*)
=\summ\frac{w_k^*}{n_k}\sum_{i=1}^{n_k}G_i\supk,
\]
with $G_i\supk$ the centered gradients of Assumption~\ref{ass:loss_psi2}, which are independent and mean zero. Fix $\Delta\in\mathcal S_{2r}$. By the definition of $\|\cdot\|_{\psi_2,\mathrm F}$, each $\langle G_i\supk,\Delta\rangle$ is mean-zero with $\|\langle G_i\supk,\Delta\rangle\|_{\psi_2}\le K_G$, so by \eqref{eq:appen:psi2_facts},
\[
\bigl\|\langle S,\Delta\rangle\bigr\|_{\psi_2}^2
\le C_\psi\summ\sum_{i=1}^{n_k}\Bigl(\frac{w_k^*}{n_k}\Bigr)^2K_G^2
=C_\psi K_G^2\,\Delta_n,
\]
and hence
\[
\PP\bigl(|\langle S,\Delta\rangle|\ge t\bigr)\le2\exp\Bigl(-\frac{t^2}{C_\psi K_G^2\Delta_n}\Bigr).
\]
Let $\bar{\mathcal S}_{2r}$ be a $\tfrac14$-net of $\mathcal S_{2r}$ from Lemma~\ref{lem:appen:lowrank_net}, so that $\log|\bar{\mathcal S}_{2r}|\le2r(d_1+d_2+1)\log36$. We claim that
\begin{equation}
\label{eq:appen:net_reduction}
\estatt=\sup_{\Delta\in\mathcal S_{2r}}\langle S,\Delta\rangle
\ \le\ 2\max_{\Delta\in\bar{\mathcal S}_{2r}}\langle S,\Delta\rangle .
\end{equation}
Indeed, given $\Delta\in\mathcal S_{2r}$ pick $\bar\Delta\in\bar{\mathcal S}_{2r}$ with $\|\Delta-\bar\Delta\|_F\le\tfrac14$. The difference $\Delta-\bar\Delta$ has rank at most $4r$; splitting its singular value decomposition into two orthogonal parts of rank at most $2r$, say $\Delta-\bar\Delta=\Delta_1+\Delta_2$ with $\|\Delta_1\|_F^2+\|\Delta_2\|_F^2=\|\Delta-\bar\Delta\|_F^2$, we get
\[
\langle S,\Delta-\bar\Delta\rangle\le\estatt\bigl(\|\Delta_1\|_F+\|\Delta_2\|_F\bigr)\le\sqrt2\,\estatt\,\|\Delta-\bar\Delta\|_F\le\frac{\sqrt2}{4}\,\estatt .
\]
Taking suprema, $\estatt\le\max_{\bar{\mathcal S}_{2r}}\langle S,\cdot\rangle+\tfrac{\sqrt2}{4}\estatt$, and $(1-\tfrac{\sqrt2}{4})^{-1}\le2$ proves \eqref{eq:appen:net_reduction}. (The supremum over the unit ball equals the supremum over the sphere because the objective is linear.) A union bound over $\bar{\mathcal S}_{2r}$ now gives, for every $t>0$,
\[
\PP\Bigl(\estatt\ge2t\Bigr)
\le2\exp\Bigl(2r(d_1+d_2+1)\log36-\frac{t^2}{C_\psi K_G^2\Delta_n}\Bigr),
\]
and choosing $t^2=C_\psi K_G^2\Delta_n\{2r(d_1+d_2+1)\log36+\log(2/\delta_0)\}$ yields \eqref{eq:appen:estat_theta}.

\medskip\noindent\textbf{Part II: bound for $\estatw$.} For each fixed $k$, $\widehat L_k(\Theta^*)-L_k(\Theta^*)=n_k^{-1}\sum_{i=1}^{n_k}e_i\supk$ with the $e_i\supk$ of Assumption~\ref{ass:loss_psi2} independent, mean zero, and $\|e_i\supk\|_{\psi_2}\le K_L$. By \eqref{eq:appen:psi2_facts}, $\|\widehat L_k(\Theta^*)-L_k(\Theta^*)\|_{\psi_2}^2\le C_\psi K_L^2/n_k$, whence
\[
\PP\biggl(\bigl|\widehat L_k(\Theta^*)-L_k(\Theta^*)\bigr|
\ge K_L\sqrt{\frac{C_\psi\log(2m/\delta_0)}{n_k}}\biggr)\le\frac{\delta_0}{m}.
\]
A union bound over $k\in[m]$ and $n_k\ge n_{\min}$ give \eqref{eq:appen:estat_w}.
\end{proof}

\subsection{Proof of Theorem~\ref{thm:PILLAR_convergence}}
\label{appen:subsec:thm1_final}

\begin{proof}[Proof of Theorem~\ref{thm:PILLAR_convergence}]
We prove the theorem with the explicit constants
\[
c_1=\frac{82\beta_L^2}{\sqrt{\mu\tilde\mu}},
\qquad
c_2\ \text{as in Remark~\ref{rem:appen:stepsize}},
\qquad
c_3=\frac{\tilde\mu}{10},
\qquad
\gamma=\frac{\etat\tilde\mu\sigma_{\min}}{40},
\]
so that the conditions on $\lambda$ and on the step sizes in Theorem~\ref{thm:PILLAR_convergence} coincide with (C1)--(C2) of Theorem~\ref{thm:appen:deterministic}. The phrase ``sufficiently large sample sizes'' in the statement refers to the condition
\begin{equation}
\label{eq:appen:samplesize}
n_{\mathrm{eff}}\ \ge\ \frac{C\,K_G^2\bigl\{r(d_1+d_2)+\log(1/\delta)\bigr\}}{\mu\tilde\mu\,\sigma_{\min}I_0^2},
\qquad
\lambda^2n_{\min}\ \ge\ \frac{C\,\beta_L^2K_L^2\,\log(m/\delta)}{\mu\tilde\mu\,\sigma_{\min}I_0^2},
\end{equation}
for a sufficiently large universal constant $C$. Recall from \eqref{eq:appen:I0} that
\[
\sigma_{\min}I_0^2\ \ge\ c(\mu,\beta)\,\sigma_{\min}^2\,\min\bigl\{1,\ \varrho^2/(\sigma_{\min}\kappa^*)\bigr\},
\]
so \eqref{eq:appen:samplesize} is a genuine sample-size requirement and involves no unknown quantities beyond those in the theorem.

Fix $\delta\in(0,1)$ and let $\mathcal E$ be the intersection of the two events of Proposition~\ref{prop:appen:estat}, each taken with $\delta_0=\delta/2$; then $\PP(\mathcal E)\ge1-\delta$. On $\mathcal E$, the bounds \eqref{eq:appen:estat_theta}--\eqref{eq:appen:estat_w} and \eqref{eq:appen:samplesize} imply the statistical-error condition (C3) of Theorem~\ref{thm:appen:deterministic}. Conditions (C1)--(C2) hold by hypothesis. Theorem~\ref{thm:appen:deterministic} therefore applies, and \eqref{eq:appen:theta_T} together with \eqref{eq:appen:estat_theta}--\eqref{eq:appen:estat_w} gives, on $\mathcal E$,
\begin{align*}
\bigl\|\widehat\Theta^T-\Theta^*\bigr\|_F
&\le\sqrt{6\sigma_{\max}}\;I_0\,(1-2\gamma)^{T/2}
+\sqrt{\frac{1680\,\sigma_{\max}}{\mu\tilde\mu\,\sigma_{\min}}}\,
\Bigl(\estatt+\frac{2\beta_L}{\lambda}\estatw\Bigr)\\
&\le C_0\,(1-\gamma)^{T}
+C_1\sqrt{\frac{r(d_1+d_2)+\log(1/\delta)}{n_{\mathrm{eff}}}
+\frac{\log(m/\delta)}{\lambda^2\,n_{\min}}},
\end{align*}
where $C_0:=\sqrt{6\sigma_{\max}}\,I_0$ and $C_1:=C\sqrt{\kappa^*/(\mu\tilde\mu)}\,\max\{K_G,\beta_LK_L\}$ for a universal constant $C$; here we used $(1-2\gamma)^{1/2}\le1-\gamma$ and $a+b\le\sqrt2\,\sqrt{a^2+b^2}$. This is \eqref{eq:theta_convergence}. The final claim follows by choosing $T\ge C\log(1/\varepsilon_n)$ with $C\ge\{\log(1/(1-\gamma))\}^{-1}\max\{1,\log C_0\}$, so that $C_0(1-\gamma)^T\le\varepsilon_n$, whence $\|\widehat\Theta^T-\Theta^*\|_F\le(1+C_1)\varepsilon_n=:C_2\varepsilon_n$.
\end{proof}

\begin{remark}[On the constants]
\label{rem:appen:constants}
The proof shows that $c_1$ and $c_2$ depend only on $(\mu,\beta)$, that $\gamma$ is proportional to $\etat\sigma_{\min}\min\{\mu,1\}$, and that $C_1$ depends only on $(\mu,\beta,K_G,K_L)$ and on $\sqrt{\kappa^*}$. The condition (C3), made explicit in \eqref{eq:appen:samplesize}, is the familiar requirement that the statistical error be smaller than the radius of the basin of attraction; an analogous condition appears in the pooled low-rank setting \citep[cf.][]{yu2020simultaneous}.
\end{remark}

\subsection{Proof of Corollary~\ref{cor:fdiv}}
\label{appen:subsec:cor_fdiv}

Throughout this subsection we work in the configuration of Corollary~\ref{cor:fdiv}: $m=n$, $n_k=1$, $X_1,\dots,X_n$ independent with common law $\PP$, common loss $\ell$, and $D(w)=D_f(w\|w^{\mathrm{ref}})$ with $w^{\mathrm{ref}}=n^{-1}\mathbf 1$. Since $f$ is convex with $f(1)=0$, Jensen's inequality gives $D\ge0=D(w^{\mathrm{ref}})$, and $w^{\mathrm{ref}}$ is the unique minimizer of $D$ because Assumption~\ref{ass:w_regularizer_mirror} makes $D$ strictly convex on $\Delta_m$. We use the maximizer map $M_\eta(\cdot)$ of Lemma~\ref{lem:appen:inner_stability} in Appendix~\ref{appen:thm3proof}; in particular, $M_\eta(c\mathbf 1)=w^{\mathrm{ref}}$ for every $c\in\RR$ and $\eta>0$.

\begin{proof}[Proof of Corollary~\ref{cor:fdiv}]
\textbf{Part (i).}
With $n_k=1$ we have $\Lkh(\Theta)=\ell(\Theta,X_k)$, and each $w\in\Delta_n$ corresponds to the distribution $\QQ_w:=\sum_{k=1}^n w_k\,\delta_{X_k}\ll\widehat\PP_n$, for which
\[
\sum_{k=1}^n w_k\Lkh(\Theta)=\EE_{\QQ_w}\,\ell(\Theta,X),
\qquad
D_f\bigl(\QQ_w\,\|\,\widehat\PP_n\bigr)=\sum_{k=1}^n\frac1n\,f(nw_k)=D_f\bigl(w\,\|\,n^{-1}\mathbf 1\bigr),
\]
the divergences being computed over the sample indices. Ranging over $w\in\Delta_n$ exhausts all $\QQ\ll\widehat\PP_n$, which proves \eqref{eq:kl_dro_lagrangian}; the balancing penalty in \eqref{eq:empirical_minmax} vanishes at any balanced factorization of a given $\Theta$ and therefore does not affect the minimal value. At the population level, $L_k(\Theta)=\EE_\PP\,\ell(\Theta,X)=:L(\Theta)$ for every $k$, so for each fixed $\Theta$,
\[
\max_{w\in\Delta_n}\Bigl\{\sum_{k=1}^n w_kL(\Theta)-\lambda D(w)\Bigr\}
=L(\Theta)-\lambda\min_{w\in\Delta_n}D(w)
=L(\Theta),
\]
with the maximum attained at $w=w^{\mathrm{ref}}$. Hence $\Phi_\lambda=L$ for every $\lambda>0$, so that $\Theta^*\in\argmin{\Theta\in\Xi(r)}L(\Theta)$ and, the population loss vector at $\Theta^*$ being constant across sources, $w^*=M_\lambda\bigl(L(\Theta^*)\mathbf 1\bigr)=w^{\mathrm{ref}}=n^{-1}\mathbf 1$; consequently $n_{\mathrm{eff}}=\{\sum_{k=1}^n(1/n)^2\}^{-1}=n$.

\textbf{Part (ii).}
By part (i), the saddle point, and hence $\kappa^*$, does not depend on $\lambda$, so the condition $\lambda_n\ge c_1\sqrt{\kappa^*}$ holds as soon as $n\ge(c_1\sqrt{\kappa^*}/\lambda_0)^2$. Fix $\delta\in(0,1)$. In the sample-size condition \eqref{eq:appen:samplesize}, the left-hand sides are $n_{\mathrm{eff}}=n$ and $\lambda_n^2n_{\min}=\lambda_0^2\,n$, while the right-hand sides grow at most logarithmically in $n$; hence \eqref{eq:appen:samplesize} holds for all $n\ge n_0(\delta)$. Theorem~\ref{thm:PILLAR_convergence} applied with $m=n$, $n_{\min}=1$ and $n_{\mathrm{eff}}=n$ then gives, with probability at least $1-\delta$, for all $T\ge C\log(1/\varepsilon_n)$,
\[
\bigl\|\widehat\Theta^T-\Theta^*\bigr\|_F\le C_2\,\varepsilon_n,
\qquad
\varepsilon_n^2
=\frac{r(d_1+d_2)+\log(1/\delta)}{n}+\frac{\log(n/\delta)}{\lambda_0^2\,n}
\;\le\;\bigl(1+\lambda_0^{-2}\bigr)\,\frac{r(d_1+d_2)+\log(n/\delta)}{n}.
\]
Since $\delta\in(0,1)$ was arbitrary and $\log(1/\varepsilon_n)\le C'\log n$, this is the claim of part (ii).

\textbf{Part (iii).}
Write $\widehat L:=\widehat L(\widehat\Theta^T)$, let $c$ be the mid-range of $(\Lkh(\widehat\Theta^T))_{k\le n}$ and $\widehat H$ its range, so that $\|\widehat L-c\mathbf 1\|_\infty=\widehat H/2$. In the notation of Lemma~\ref{lem:appen:inner_stability}, $\widehat w_n=M_{\lambda_n}(\widehat L)$. Comparing the value of the inner maximization at $\widehat w_n$ with its value at $w^{\mathrm{ref}}$, where $D$ vanishes, and using $\langle\mathbf 1,\widehat w_n-w^{\mathrm{ref}}\rangle=0$,
\[
\lambda_n\,\widehat\rho_n
=\lambda_n\,D(\widehat w_n)
\le\bigl\langle\widehat L,\,\widehat w_n-w^{\mathrm{ref}}\bigr\rangle
=\bigl\langle\widehat L-c\mathbf 1,\,\widehat w_n-w^{\mathrm{ref}}\bigr\rangle
\le\frac{\widehat H}{2}\,\bigl\|\widehat w_n-w^{\mathrm{ref}}\bigr\|_1 .
\]
Moreover, $\|\widehat w_n-w^{\mathrm{ref}}\|_1=\|M_{\lambda_n}(\widehat L)-M_{\lambda_n}(c\mathbf 1)\|_1\le\widehat H/(2\lambda_n)$ by \eqref{eq:appen:inner_stab_theta}, whence
\begin{equation}
\label{eq:appen:radius_quadratic}
\widehat\rho_n\;\le\;\frac{\widehat H^2}{4\lambda_n^2}.
\end{equation}
To bound $\widehat H$, write, for each $k$,
\[
\Lkh(\widehat\Theta^T)=L(\Theta^*)+e_k+\bigl\{\Lkh(\widehat\Theta^T)-\Lkh(\Theta^*)\bigr\},
\qquad e_k=\ell(\Theta^*,X_k)-L(\Theta^*).
\]
By Assumption~\ref{ass:loss_psi2}, \eqref{eq:appen:psi2_facts} and a union bound, $\max_{k\le n}|e_k|\le K_L\sqrt{\log(4n/\delta)}$ with probability at least $1-\delta/2$. Moreover, on the event of part (ii) the iterates remain in the neighborhood \eqref{eq:local_neighborhood} (Step~3 of the proof of Theorem~\ref{thm:appen:deterministic}), so Lemma~\ref{lemma:ZandZ*}(ii) places $\widehat\Theta^T$ in the ball $\{\|\Theta-\Theta^*\|_F\le\varrho\}$, and the mean value theorem with the gradient bound of Assumption~\ref{ass:rsc_rss} yields $|\Lkh(\widehat\Theta^T)-\Lkh(\Theta^*)|\le\beta_L\varrho$ for every $k$. Hence, with probability at least $1-2\delta$, $\widehat H\le2K_L\sqrt{\log(4n/\delta)}+2\beta_L\varrho$, and \eqref{eq:appen:radius_quadratic} gives
\[
\widehat\rho_n
\le\frac{\bigl\{2K_L\sqrt{\log(4n/\delta)}+2\beta_L\varrho\bigr\}^2}{4\lambda_0^2\,n}
\le\frac{2K_L^2\log(4n/\delta)+2\beta_L^2\varrho^2}{\lambda_0^2\,n},
\]
which is $O_{\PP}(\log n/n)$, as claimed.
\end{proof}

\section{Proof of Theorem~\ref{thm:softmax_init}}
\label{appen:proof_thm2}

This appendix proves Theorem~\ref{thm:softmax_init}, which shows that the convex soft-max program \eqref{eq:softmax_init}, followed by the soft-max weight assignment \eqref{eq:softmax_weight_init}, produces an initialization satisfying Assumption~\ref{ass:initialization} with the constants \eqref{eq:appen:ITheta_Iw}. Throughout this appendix the regularizer is the Kullback--Leibler divergence $D(w)=D_{\mathrm{KL}}(w\|w^{\mathrm{ref}})$ with $w^{\mathrm{ref}}\in\operatorname{int}(\Delta_m)$, and the mirror map is the negative entropy $\phi(w)=\sum_kw_k\log w_k$, so that $B_\phi(u,w)=D_{\mathrm{KL}}(u\|w)$ and Assumption~\ref{ass:w_regularizer_mirror} holds with $\beta_D=1$; the $1$-strong convexity of $\phi$ with respect to $\|\cdot\|_1$ is Pinsker's inequality.

Section~\ref{appen:subsec:softmax_prelim} develops the calculus of the soft-max profile. Section~\ref{appen:subsec:nuclear_geometry} collects the (standard) decomposability geometry of the nuclear norm. Section~\ref{appen:subsec:grad_bound} bounds the gradient of the profiled loss at $\Theta^*$, which dictates the choice of $\tau_n$. Section~\ref{appen:subsec:thm2_proof} contains the proof of the theorem.

\subsection{The soft-max profile}
\label{appen:subsec:softmax_prelim}

For $\ell\in\RR^m$ define the log-partition and soft-max maps
\begin{equation}
\label{eq:appen:softmax_def}
A_\lambda(\ell):=\lambda\log\Bigl(\summ w_k^{\mathrm{ref}}e^{\ell_k/\lambda}\Bigr),
\qquad
\mathcal S_k(\ell):=\frac{w_k^{\mathrm{ref}}e^{\ell_k/\lambda}}{\sum_{j=1}^mw_j^{\mathrm{ref}}e^{\ell_j/\lambda}},
\quad k\in[m],
\end{equation}
so that $\nabla A_\lambda(\ell)=\mathcal S(\ell)\in\operatorname{int}(\Delta_m)$. By the conjugate representation of the Kullback--Leibler divergence (the identity \eqref{eq:kl_profile_identity}),
\begin{equation}
\label{eq:appen:profile}
\widehat F(\Theta):=\max_{w\in\Delta_m}\Bigl\{\summ w_k\Lkh(\Theta)-\lambda D(w)\Bigr\}
=A_\lambda\bigl(\widehat L(\Theta)\bigr),
\end{equation}
with maximizer $\mathcal S(\widehat L(\Theta))$; in particular \eqref{eq:softmax_init} minimizes $\widehat F(\Theta)+\tau_n\|\Theta\|_*$, the weight initializer \eqref{eq:softmax_weight_init} is $w^0=\mathcal S(\widehat L(\Theta^0))$, and, at the population level, the inner maximizer at $\Theta^*$ is
\begin{equation}
\label{eq:appen:wstar_softmax}
w^*=\mathcal S\bigl(L(\Theta^*)\bigr).
\end{equation}

The analysis is carried out under the assumptions of Theorem~\ref{thm:softmax_init}: Assumptions~\ref{ass:rsc_rss}, \ref{ass:loss_psi2} and~\ref{ass:pop_foc}, the radius condition $\varrho\ge\sigma_{\min}/4$ in Assumption~\ref{ass:rsc_rss}, and the convexity of each $\Lkh$ on $\RR^{d_1\times d_2}$. Convexity is needed because the nuclear-norm solution and the associated error directions need not be of low rank; the restricted strong convexity of Assumption~\ref{ass:rsc_rss} enters along the nuclear-norm cone $\mathcal C$ defined in \eqref{eq:appen:cone} below, which is contained in the set $\mathcal A$ of \eqref{eq:ratio_cone} by Lemma~\ref{lem:tangent_rank}. Note also that, since $\|\nabla\Lkh(\Theta^*)\|_F\le\beta_L$ almost surely by Assumption~\ref{ass:rsc_rss}, Jensen's inequality yields the population bound $\|\nabla L_k(\Theta^*)\|_F\le\beta_L$ as well.

The next two lemmas develop the calculus of the profile \eqref{eq:appen:profile}.

\begin{lemma}[Properties of the soft-max profile]
\label{lem:appen:profile}
Suppose each $\Lkh$ is convex on $\RR^{d_1\times d_2}$. Then:
\begin{enumerate}
\item[(i)] $\widehat F$ is convex on $\RR^{d_1\times d_2}$ and differentiable, with
\[
\nabla\widehat F(\Theta)=\summ\mathcal S_k\bigl(\widehat L(\Theta)\bigr)\,\nabla\Lkh(\Theta);
\]
\item[(ii)] if, in addition, Assumption~\ref{ass:rsc_rss} holds, then for all $\Delta\in\mathcal A$,
\[
\widehat F(\Theta^*+\Delta)-\widehat F(\Theta^*)-\bigl\langle\nabla\widehat F(\Theta^*),\Delta\bigr\rangle\ \ge\ \frac{\mu}{2}\|\Delta\|_F^2 .
\]
\end{enumerate}
\end{lemma}

\begin{proof}
(i) $A_\lambda$ is convex and nondecreasing in each coordinate (its gradient $\mathcal S(\ell)$ lies in $\Delta_m$), and each $\Lkh$ is convex; hence the composition $\widehat F=A_\lambda\circ\widehat L$ is convex. Differentiability and the gradient formula follow from the chain rule.

(ii) Write $\ell:=\widehat L(\Theta^*)$, $\widehat w:=\mathcal S(\ell)$, and let $v\in\RR^m$ have coordinates $v_k:=\langle\nabla\Lkh(\Theta^*),\Delta\rangle$. By the lower bound in \eqref{eq:rsc_rss}, applied to the pair $(\Theta^*,\Theta^*+\Delta)$ whose difference $\Delta$ lies in $\mathcal A$,
\[
\widehat L_k(\Theta^*+\Delta)\ \ge\ \ell_k+v_k+\frac{\mu}{2}\|\Delta\|_F^2,
\qquad k\in[m].
\]
Since $A_\lambda$ is nondecreasing in each coordinate and satisfies the translation property $A_\lambda(\ell+c\mathbf 1)=A_\lambda(\ell)+c$ for $c\in\RR$,
\[
\widehat F(\Theta^*+\Delta)=A_\lambda\bigl(\widehat L(\Theta^*+\Delta)\bigr)
\ \ge\ A_\lambda\Bigl(\ell+v+\frac{\mu}{2}\|\Delta\|_F^2\,\mathbf 1\Bigr)
=A_\lambda(\ell+v)+\frac{\mu}{2}\|\Delta\|_F^2 .
\]
Finally, by convexity of $A_\lambda$, $A_\lambda(\ell+v)\ge A_\lambda(\ell)+\langle\mathcal S(\ell),v\rangle=\widehat F(\Theta^*)+\langle\nabla\widehat F(\Theta^*),\Delta\rangle$, where the last equality is the gradient formula from (i). Combining the two displays proves (ii).
\end{proof}

\begin{lemma}[Stability of the soft-max map]
\label{lem:appen:softmax_stability}
For all $a,b\in\RR^m$,
\begin{equation}
\label{eq:appen:softmax_KL}
D_{\mathrm{KL}}\bigl(\mathcal S(a)\,\big\|\,\mathcal S(b)\bigr)\ \le\ \frac{\|a-b\|_\infty^2}{2\lambda^2},
\qquad
\bigl\|\mathcal S(a)-\mathcal S(b)\bigr\|_1\ \le\ \frac{\|a-b\|_\infty}{\lambda}.
\end{equation}
\end{lemma}

\begin{proof}
Let $A(x):=\log\bigl(\sum_kw_k^{\mathrm{ref}}e^{x_k}\bigr)$, so that $\mathcal S(a)=\nabla A(a/\lambda)$ and $A_\lambda(a)=\lambda A(a/\lambda)$. The family $\{p_x:=\nabla A(x)\}$ is an exponential family with log-partition $A$, for which the Kullback--Leibler divergence is the Bregman divergence of $A$:
\[
D_{\mathrm{KL}}(p_s\|p_t)=A(t)-A(s)-\langle\nabla A(s),t-s\rangle .
\]
By Taylor's theorem with integral remainder, the right-hand side equals $\tfrac12(t-s)^\top\nabla^2A(\xi)(t-s)$ for some $\xi$ on the segment $[s,t]$, and $\nabla^2A(\xi)=\operatorname{diag}(p_\xi)-p_\xi p_\xi^\top$, whence
\[
(t-s)^\top\nabla^2A(\xi)(t-s)=\operatorname{Var}_{p_\xi}\bigl(t-s\bigr)\le\max_k(t_k-s_k)^2=\|t-s\|_\infty^2 .
\]
Taking $s=a/\lambda$, $t=b/\lambda$ gives the first bound in \eqref{eq:appen:softmax_KL}. The second follows from the first by Pinsker's inequality, $\|p-q\|_1\le\sqrt{2D_{\mathrm{KL}}(p\|q)}$.
\end{proof}

\subsection{Nuclear-norm geometry}
\label{appen:subsec:nuclear_geometry}

Let $\Theta^*$ have compact singular value decomposition $\Theta^*=U_*\Sigma_*V_*^\top$ with $\rank(\Theta^*)=r$. Define the tangent space and cone
\begin{align}
T&:=\bigl\{U_*A^\top+BV_*^\top:\ A\in\RR^{d_2\times r},\ B\in\RR^{d_1\times r}\bigr\},
\label{eq:appen:tangent}\\
\mathcal C&:=\bigl\{\Delta\in\RR^{d_1\times d_2}:\ \|\mathcal P_{T^\perp}\Delta\|_*\le3\|\mathcal P_T\Delta\|_*\bigr\},
\label{eq:appen:cone}
\end{align}
where $\mathcal P_T$ and $\mathcal P_{T^\perp}$ are the orthogonal projections onto $T$ and $T^\perp$ with respect to the Frobenius inner product. The following three lemmas are standard \citep{negahban2012restricted}; we include their short proofs for completeness.

\begin{lemma}[Decomposability of the nuclear norm]
\label{lem:decomp}
For any $\Delta\in\RR^{d_1\times d_2}$,
\[
\|\Theta^*+\Delta\|_*-\|\Theta^*\|_*\ \ge\ \|\mathcal P_{T^\perp}\Delta\|_*-\|\mathcal P_T\Delta\|_* .
\]
\end{lemma}

\begin{proof}
The subdifferential of the nuclear norm at $\Theta^*$ is
$\partial\|\Theta^*\|_*=\{U_*V_*^\top+W:\ \mathcal P_T(W)=0,\ \|W\|_2\le1\}$.
Choose $W$ with $\mathcal P_T(W)=0$, $\|W\|_2\le1$ and $\langle W,\mathcal P_{T^\perp}\Delta\rangle=\|\mathcal P_{T^\perp}\Delta\|_*$ (take $W$ to be the product of the singular-vector factors of $\mathcal P_{T^\perp}\Delta$, which lies in $T^\perp$). By convexity,
\[
\|\Theta^*+\Delta\|_*\ \ge\ \|\Theta^*\|_*+\bigl\langle U_*V_*^\top+W,\Delta\bigr\rangle
=\|\Theta^*\|_*+\bigl\langle U_*V_*^\top,\mathcal P_T\Delta\bigr\rangle+\|\mathcal P_{T^\perp}\Delta\|_* ,
\]
and $\langle U_*V_*^\top,\mathcal P_T\Delta\rangle\ge-\|U_*V_*^\top\|_2\|\mathcal P_T\Delta\|_*\ge-\|\mathcal P_T\Delta\|_*$.
\end{proof}

\begin{lemma}[Cone condition]
\label{lem:cone_condition}
Suppose each $\Lkh$ is convex on $\RR^{d_1\times d_2}$ and $\tau_n\ge2\|\nabla\widehat F(\Theta^*)\|_2$. Then any solution $\Theta^0$ of \eqref{eq:softmax_init} satisfies $\widehat\Delta:=\Theta^0-\Theta^*\in\mathcal C$.
\end{lemma}

\begin{proof}
Optimality of $\Theta^0$ and feasibility of $\Theta^*$ give
\[
\widehat F(\Theta^*+\widehat\Delta)+\tau_n\|\Theta^*+\widehat\Delta\|_*\ \le\ \widehat F(\Theta^*)+\tau_n\|\Theta^*\|_* .
\]
By convexity of $\widehat F$ (Lemma~\ref{lem:appen:profile}(i)), $\widehat F(\Theta^*+\widehat\Delta)-\widehat F(\Theta^*)\ge\langle\nabla\widehat F(\Theta^*),\widehat\Delta\rangle\ge-\|\nabla\widehat F(\Theta^*)\|_2\|\widehat\Delta\|_*\ge-\tfrac{\tau_n}{2}\|\widehat\Delta\|_*$, where we used the duality of the spectral and nuclear norms. Combining with Lemma~\ref{lem:decomp},
\[
0\ \ge\ -\frac{\tau_n}{2}\bigl(\|\mathcal P_T\widehat\Delta\|_*+\|\mathcal P_{T^\perp}\widehat\Delta\|_*\bigr)
+\tau_n\bigl(\|\mathcal P_{T^\perp}\widehat\Delta\|_*-\|\mathcal P_T\widehat\Delta\|_*\bigr),
\]
which rearranges to $\|\mathcal P_{T^\perp}\widehat\Delta\|_*\le3\|\mathcal P_T\widehat\Delta\|_*$.
\end{proof}

\begin{lemma}[Rank control on the tangent component]
\label{lem:tangent_rank}
Every $\Delta\in T$ satisfies $\rank(\Delta)\le2r$, and hence $\|\mathcal P_T\Delta\|_*\le\sqrt{2r}\,\|\Delta\|_F$ for every $\Delta$. If moreover $\Delta\in\mathcal C$, then $\|\Delta\|_*\le4\sqrt{2r}\,\|\Delta\|_F$.
\end{lemma}

\begin{proof}
Elements of $T$ are sums of two matrices of rank at most $r$. Therefore $\|\mathcal P_T\Delta\|_*\le\sqrt{2r}\,\|\mathcal P_T\Delta\|_F\le\sqrt{2r}\,\|\Delta\|_F$, since $\mathcal P_T$ is an orthogonal projection. For $\Delta\in\mathcal C$, $\|\Delta\|_*\le\|\mathcal P_T\Delta\|_*+\|\mathcal P_{T^\perp}\Delta\|_*\le4\|\mathcal P_T\Delta\|_*$.
\end{proof}

\subsection{A spectral-norm bound for the profiled gradient}
\label{appen:subsec:grad_bound}

The tuning parameter $\tau_n$ must dominate $\|\nabla\widehat F(\Theta^*)\|_2$. By Lemma~\ref{lem:appen:profile}(i) and Assumption~\ref{ass:pop_foc},
\begin{equation}
\label{eq:appen:gradF_decomp}
\nabla\widehat F(\Theta^*)
=\underbrace{\summ\widehat w_k\bigl\{\nabla\Lkh(\Theta^*)-\nabla L_k(\Theta^*)\bigr\}}_{(\mathrm I)}
+\underbrace{\summ(\widehat w_k-w^*_k)\,\nabla L_k(\Theta^*)}_{(\mathrm{II})},
\qquad
\widehat w:=\mathcal S\bigl(\widehat L(\Theta^*)\bigr),
\end{equation}
where we added and subtracted $\summ\widehat w_k\nabla L_k(\Theta^*)$ and used $\summ w^*_k\nabla L_k(\Theta^*)=0$. Term $(\mathrm I)$ is a $\widehat w$-weighted average of the centered gradient matrices; because $\widehat w$ is random, we bound it by the maximum over groups. Term $(\mathrm{II})$ reflects the sensitivity of the soft-max weights to the fluctuations of the group losses, and is controlled by Lemma~\ref{lem:appen:softmax_stability}. We use the following standard net bound; see \citet[Section~4.4]{vershynin2018high}.

\begin{lemma}[$\varepsilon$-net bound for the spectral norm]
\label{lem:enet_spectral}
Fix $\varepsilon\in(0,1/2)$ and let $\mathcal N_1\subset\mathbb S^{d_1-1}$, $\mathcal N_2\subset\mathbb S^{d_2-1}$ be $\varepsilon$-nets of the unit spheres, which may be chosen with $|\mathcal N_1|\le(1+2/\varepsilon)^{d_1}$ and $|\mathcal N_2|\le(1+2/\varepsilon)^{d_2}$. Then, for any $A\in\RR^{d_1\times d_2}$,
\[
\|A\|_2\ \le\ \frac{1}{1-2\varepsilon}\max_{u\in\mathcal N_1,\ v\in\mathcal N_2}\ \langle u,Av\rangle .
\]
\end{lemma}

\begin{lemma}[Gradient bound for the profiled loss]
\label{lem:appen:gradF_bound}
Suppose Assumptions~\ref{ass:rsc_rss}, \ref{ass:loss_psi2} and~\ref{ass:pop_foc} hold. There is a universal constant $C_A>0$ such that, for any $\delta\in(0,1)$, with probability at least $1-\delta/2$,
\begin{equation}
\label{eq:appen:gradF_bound}
\bigl\|\nabla\widehat F(\Theta^*)\bigr\|_2\ \le\ \frac{C_A\,A_\delta}{\sqrt{n_{\min}}},
\qquad\text{where}\quad
A_\delta=K_G\sqrt{d_1+d_2+\log(m/\delta)}+\frac{\beta K_L}{\lambda}\sqrt{\log(m/\delta)}
\end{equation}
as in Theorem~\ref{thm:softmax_init}, and moreover, on the same event,
\begin{equation}
\label{eq:appen:ew_event}
\bigl\|\widehat L(\Theta^*)-L(\Theta^*)\bigr\|_\infty\ \le\ C_A\,K_L\sqrt{\frac{\log(m/\delta)}{n_{\min}}} .
\end{equation}
\end{lemma}

\begin{proof}
Write $\bar G_k:=\nabla\Lkh(\Theta^*)-\nabla L_k(\Theta^*)=n_k^{-1}\sum_{i=1}^{n_k}G_i\supk$ with the centered gradients $G_i\supk$ of Assumption~\ref{ass:loss_psi2}.

\emph{Term $(\mathrm I)$.} Since $\widehat w\in\Delta_m$, $\|(\mathrm I)\|_2\le\max_{k\in[m]}\|\bar G_k\|_2$. Fix $k$ and unit vectors $u,v$. Then $\langle u,\bar G_kv\rangle=n_k^{-1}\sum_i\langle G_i\supk,uv^\top\rangle$ with $\|uv^\top\|_F=1$, so by Assumption~\ref{ass:loss_psi2} and \eqref{eq:appen:psi2_facts}, $\|\langle u,\bar G_kv\rangle\|_{\psi_2}^2\le C_\psi K_G^2/n_k$ and
\[
\PP\bigl(|\langle u,\bar G_kv\rangle|\ge t\bigr)\le2\exp\Bigl(-\frac{n_kt^2}{C_\psi K_G^2}\Bigr).
\]
Take $\varepsilon=1/4$ in Lemma~\ref{lem:enet_spectral}, so $|\mathcal N_1||\mathcal N_2|\le9^{d_1+d_2}$, and apply a union bound over $\mathcal N_1\times\mathcal N_2$ and over $k\in[m]$: with probability at least $1-\delta/4$,
\[
\max_{k\in[m]}\|\bar G_k\|_2
\le2\max_{k}\max_{u\in\mathcal N_1,v\in\mathcal N_2}|\langle u,\bar G_kv\rangle|
\le C\,K_G\sqrt{\frac{d_1+d_2+\log(m/\delta)}{n_{\min}}} .
\]

\emph{Term $(\mathrm{II})$.} By $\|\nabla L_k(\Theta^*)\|_F\le\beta_L$ (Assumption~\ref{ass:rsc_rss} and Jensen's inequality), H\"older's inequality, \eqref{eq:appen:wstar_softmax} and Lemma~\ref{lem:appen:softmax_stability},
\[
\|(\mathrm{II})\|_2\le\|(\mathrm{II})\|_F
\le\beta_L\,\bigl\|\widehat w-w^*\bigr\|_1
=\beta_L\,\bigl\|\mathcal S(\widehat L(\Theta^*))-\mathcal S(L(\Theta^*))\bigr\|_1
\le\frac{\beta_L}{\lambda}\,\bigl\|\widehat L(\Theta^*)-L(\Theta^*)\bigr\|_\infty .
\]
Exactly as in Part~II of the proof of Proposition~\ref{prop:appen:estat}, with probability at least $1-\delta/4$ the bound \eqref{eq:appen:ew_event} holds. Intersecting the two events and using $\beta_L\le\beta$ completes the proof.
\end{proof}

\subsection{Proof of Theorem~\ref{thm:softmax_init}}
\label{appen:subsec:thm2_proof}

\begin{proof}[Proof of Theorem~\ref{thm:softmax_init}]
Fix $\delta\in(0,1)$ and work throughout on the event $\mathcal E_0$ of Lemma~\ref{lem:appen:gradF_bound}, which has probability at least $1-\delta/2$; all statements below hold on $\mathcal E_0$. Choose $C_\tau\ge2C_A$, so that
\[
\tau_n=\frac{C_\tau A_\delta}{\sqrt{n_{\min}}}\ \ge\ 2\,\bigl\|\nabla\widehat F(\Theta^*)\bigr\|_2 .
\]

\medskip\noindent\textbf{Step 1: Frobenius error of the convex solution.}
Let $\widehat\Delta:=\Theta^0-\Theta^*$. By Lemma~\ref{lem:cone_condition}, $\widehat\Delta\in\mathcal C$, and hence $\widehat\Delta\in\mathcal A$ by Lemma~\ref{lem:tangent_rank}. Optimality of $\Theta^0$ gives
\[
\widehat F(\Theta^*+\widehat\Delta)-\widehat F(\Theta^*)\ \le\ \tau_n\bigl(\|\Theta^*\|_*-\|\Theta^*+\widehat\Delta\|_*\bigr)
\ \le\ \tau_n\bigl(\|\mathcal P_T\widehat\Delta\|_*-\|\mathcal P_{T^\perp}\widehat\Delta\|_*\bigr),
\]
using Lemma~\ref{lem:decomp}. On the other hand, restricted strong convexity of the profile (Lemma~\ref{lem:appen:profile}(ii)) and the choice of $\tau_n$ give
\[
\widehat F(\Theta^*+\widehat\Delta)-\widehat F(\Theta^*)
\ \ge\ \bigl\langle\nabla\widehat F(\Theta^*),\widehat\Delta\bigr\rangle+\frac{\mu}{2}\|\widehat\Delta\|_F^2
\ \ge\ -\frac{\tau_n}{2}\bigl(\|\mathcal P_T\widehat\Delta\|_*+\|\mathcal P_{T^\perp}\widehat\Delta\|_*\bigr)+\frac{\mu}{2}\|\widehat\Delta\|_F^2 .
\]
Combining the two displays and dropping the resulting nonpositive multiple of $\|\mathcal P_{T^\perp}\widehat\Delta\|_*$,
\[
\frac{\mu}{2}\|\widehat\Delta\|_F^2\ \le\ \frac{3\tau_n}{2}\,\|\mathcal P_T\widehat\Delta\|_*
\ \le\ \frac{3\tau_n}{2}\,\sqrt{2r}\,\|\widehat\Delta\|_F ,
\]
by Lemma~\ref{lem:tangent_rank}. Hence
\begin{equation}
\label{eq:appen:theta0_bound}
\|\Theta^0-\Theta^*\|_F\ \le\ \frac{3\sqrt{2r}\,\tau_n}{\mu}
\ =\ \frac{3\sqrt2\,C_\tau}{\mu}\,\sqrt{\frac{r}{n_{\min}}}\,A_\delta .
\end{equation}
If $n_{\min}\ge C\,rA_\delta^2/\{\mu^2I_\Theta\sigma_{\min}^2\min(1,\sigma_{\min})\}$ with $C\ge18C_\tau^2$, then \eqref{eq:appen:theta0_bound} yields
\begin{equation}
\label{eq:appen:theta0_final}
\|\Theta^0-\Theta^*\|_F^2\ \le\ I_\Theta\,\sigma_{\min}^2\,\min\{1,\sigma_{\min}\},
\end{equation}
which is the first condition in \eqref{eq:init_assum}. Note for later use that, since $I_\Theta\le1/16$ by \eqref{eq:appen:ITheta_Iw}, \eqref{eq:appen:theta0_final} implies $\|\Theta^0-\Theta^*\|_F\le\sigma_{\min}/4$.

\medskip\noindent\textbf{Step 2: closeness of the initial weights.}
Recall $w^0=\mathcal S(\widehat L(\Theta^0))$ and $w^*=\mathcal S(L(\Theta^*))$, and that $B_\phi(w^*,w^0)=D_{\mathrm{KL}}(w^*\|w^0)$. By Lemma~\ref{lem:appen:softmax_stability},
\begin{equation}
\label{eq:appen:KL_quad}
\lambda\,B_\phi(w^*,w^0)\ \le\ \frac{1}{2\lambda}\,\bigl\|L(\Theta^*)-\widehat L(\Theta^0)\bigr\|_\infty^2 .
\end{equation}
Decompose, for each $k$,
\[
\bigl|L_k(\Theta^*)-\widehat L_k(\Theta^0)\bigr|
\ \le\ \bigl|L_k(\Theta^*)-\widehat L_k(\Theta^*)\bigr|
+\bigl|\widehat L_k(\Theta^*)-\widehat L_k(\Theta^0)\bigr| .
\]
The first term is at most $C_AK_L\sqrt{\log(m/\delta)/n_{\min}}$ on $\mathcal E_0$, by \eqref{eq:appen:ew_event}. For the second term, the segment joining $\Theta^*$ and $\Theta^0$ lies in the ball $\{\|\Theta-\Theta^*\|_F\le\sigma_{\min}/4\}$ by Step~1, and $\sigma_{\min}/4\le\varrho$ by hypothesis, so the mean value theorem and the gradient bound of Assumption~\ref{ass:rsc_rss} give
\[
\bigl|\widehat L_k(\Theta^*)-\widehat L_k(\Theta^0)\bigr|\ \le\ \beta_L\,\|\Theta^0-\Theta^*\|_F
\ \le\ \frac{3\sqrt2\,C_\tau\,\beta_L}{\mu}\,\sqrt{\frac{r}{n_{\min}}}\,A_\delta ,
\]
using \eqref{eq:appen:theta0_bound}. Substituting the two bounds into \eqref{eq:appen:KL_quad} and using $\beta_L\le\beta$,
\begin{align*}
\lambda\,B_\phi(w^*,w^0)
&\ \le\ \frac{1}{2\lambda\,n_{\min}}
\Bigl(3\sqrt2\,C_\tau\,\frac{\beta\sqrt r\,A_\delta}{\mu}+C_AK_L\sqrt{\log(m/\delta)}\Bigr)^{\!2}\\
&\ \le\ \frac{C'}{\lambda\,n_{\min}}
\Bigl(\frac{\beta\sqrt r\,A_\delta}{\mu}+K_L\sqrt{\log(m/\delta)}\Bigr)^{\!2}
\end{align*}
with $C':=\max\{9C_\tau^2,C_A^2\}$. Hence, if
$n_{\min}\ge C\bigl(\beta\sqrt rA_\delta/\mu+K_L\sqrt{\log(m/\delta)}\bigr)^2/\{\lambda I_w\sigma_{\min}^2\}$ with $C\ge C'$, then
\begin{equation}
\label{eq:appen:w0_final}
\lambda\,B_\phi(w^*,w^0)\ \le\ I_w\,\sigma_{\min}^2 ,
\end{equation}
which is the second condition in \eqref{eq:init_assum}.

\medskip
Taking $C\ge\max\{18C_\tau^2,C'\}$ in the sample-size condition of the theorem, both \eqref{eq:appen:theta0_final} and \eqref{eq:appen:w0_final} hold on $\mathcal E_0$, an event of probability at least $1-\delta/2\ge1-\delta$. This completes the proof.
\end{proof}

\begin{remark}
\label{rem:appen:thm2}
The proof clarifies the two roles of the sample-size condition in Theorem~\ref{thm:softmax_init}: the first branch makes the nuclear-norm error \eqref{eq:appen:theta0_bound} small relative to the basin radius of Theorem~\ref{thm:PILLAR_convergence}, and the second branch makes the soft-max weights, which respond to the group losses at scale $1/\lambda$, correspondingly accurate. The two terms of $A_\delta$ mirror the decomposition \eqref{eq:appen:gradF_decomp}: the first is the fluctuation of the gradient noise in spectral norm, the second the transmission of loss-level noise through the soft-max weighting.
\end{remark}

\section{Proof of Theorem \ref{thm:prior_benefit_gdro}}
\label{appen:thm3proof}

In this section we prove Theorem~\ref{thm:prior_benefit_gdro}, which quantifies the benefit of the PILLAR model over the vanilla group DRO model in terms of conservativeness. The argument compares different regularization levels, so, unlike the proofs of Theorems~\ref{thm:PILLAR_convergence} and~\ref{thm:softmax_init}, we retain the subscript throughout and write $(\Tsl,\wsl)$ for the PILLAR saddle point of Definition~\ref{def:pillar_target}, with $\Phi_\lambda$ the regularized worst-mixture risk \eqref{eq:Phi_lambda}. For a radius $\rho>0$, define the prior-centered worst-case risk and its low-rank optimum
\[
    \mathcal R_\rho(\Theta)
    :=
    \sup_{q\in\Delta_m:\,D(q)\le \rho}
    \sum_{k=1}^m q_k L_k(\Theta),
    \qquad
    V_\rho
    :=
    \inf_{\Theta\in\Xi(r)}
    \mathcal R_\rho(\Theta).
\]
Thus $V_\rho$ is the best achievable low-rank risk when the evaluation mixture is only required to lie in the $D$-ball of radius $\rho$ around the prior. Define the dual certificate and the certified gap
\[
    J_\rho(\lambda)
    :=
    \Phi_\lambda(\Tsl)+\lambda\rho,
    \qquad
    C_\rho(\lambda)
    :=
    J_\rho(\lambda)-\mathcal R_\rho(\Theta^*_{\mathrm{G}}),
\]
and recall from Theorem~\ref{thm:prior_benefit_gdro} the prior gap and the source-risk range function
\[
    \Delta_{\mathrm{ref}}
    :=
    \sum_{k=1}^m w_k^{\mathrm{ref}}
    \{L_k(\Theta^*_{\mathrm{G}})-L_k(\Theta^*_{\mathrm{ref}})\},
    \qquad
    H(\Theta)
    :=
    \max_{k\in[m]}L_k(\Theta)-\min_{k\in[m]}L_k(\Theta),
\]
so that the constant $H$ of \eqref{eq:source_hetero} equals $H(\Theta^*_{\mathrm{ref}})$. We also write, for each fixed $\Theta$ and $\eta>0$,
\begin{equation}
\label{eq:appen:inner_max}
    w_\eta(\Theta)
    :=
    \argmax{w\in\Delta_m}
    \Bigl\{\sum_{k=1}^m w_kL_k(\Theta)-\eta D(w)\Bigr\},
\end{equation}
which is well defined and unique because Assumption~\ref{ass:w_regularizer_mirror} makes $D$ strictly convex; in this notation $\wsl=w_\lambda(\Tsl)$. The following result is a complete version of Theorem~\ref{thm:prior_benefit_gdro}; its first display \eqref{eq:prior_gain_explicit} contains \eqref{eq:worst_prior_valid_gain} of the main text.

\begin{theorem}
\label{thm:prior_benefit_gdro_full}
Suppose that the $L_k$ are continuous, that $D$ is nonnegative, convex, differentiable and continuous on $\Delta_m$ with $D(w^{\mathrm{ref}})=0$ for some $w^{\mathrm{ref}}\in\operatorname{relint}(\Delta_m)$, and that Assumption~\ref{ass:w_regularizer_mirror} holds. Then, for any $0<\rho<\sup_{w\in\Delta_m}D(w)$ and any $\lambda>0$:
\begin{enumerate}
\item[(i)] \textup{(Certified gain.)}
\begin{equation}
    \mathcal R_\rho(\Tsl)-\mathcal R_\rho(\Theta^*_{\mathrm{G}})
    \le
    C_\rho(\lambda)
    \le
    -\Delta_{\mathrm{ref}}
    +
    \frac{H(\Theta^*_{\mathrm{ref}})^2}{8\lambda}
    +
    \lambda\rho .
    \label{eq:prior_gain_explicit}
\end{equation}
In particular, PILLAR has strictly smaller prior-centered worst-case risk than vanilla group DRO whenever
$\Delta_{\mathrm{ref}}>H(\Theta^*_{\mathrm{ref}})^2/(8\lambda)+\lambda\rho$.
\item[(ii)] \textup{(Radius-matched oracle.)}
\begin{equation}
\inf_{\eta\ge0}J_\rho(\eta)=V_\rho,
    \qquad
    \inf_{\eta\ge0}C_\rho(\eta)
    =
    V_\rho-\mathcal R_\rho(\Theta^*_{\mathrm{G}})
    \le 0 ;
    \label{eq:dual_radius_oracle}
\end{equation}
hence an optimally penalized PILLAR target is never worse than vanilla group DRO under the criterion $\mathcal R_\rho$.
\item[(iii)] \textup{(Radius matching.)} Suppose that, at a given $\lambda$, the minimizer $\Tsl$ and the inner maximizer $\wsl=w_\lambda(\Tsl)$ are unique, and that $\Theta^*_{\lambda'}\to\Tsl$ as $\lambda'\to\lambda$ (which holds, for instance, whenever the sublevel sets $\{\Theta\in\Xi(r):\Phi_\lambda(\Theta)\le c\}$ are compact). Then $J_\rho$ and $C_\rho$ are differentiable at $\lambda$, with
\begin{equation}
\label{eq:appen:radius_matching}
J_\rho'(\lambda)=C_\rho'(\lambda)=\rho-D(\wsl).
\end{equation}
Thus an interior minimizer of the certificate $C_\rho(\cdot)$ is radius-matched, in the sense that $D(\wsl)=\rho$. Moreover, for every $\lambda>0$,
\begin{equation}
\label{eq:appen:D_H_bound}
    D(\wsl)
    \le
    \frac{H(\Tsl)}{\lambda},
\end{equation}
so that whenever $\lambda\rho>H(\Tsl)$, the penalty is too large for the radius $\rho$: the adversary satisfies $D(\wsl)<\rho$, and increasing $\lambda$ locally worsens the robustness certificate.
\end{enumerate}
\end{theorem}

Theorem~\ref{thm:prior_benefit_gdro} is the special case of part (i): under $D(w^{\mathrm{ref}})=0$, the display \eqref{eq:prior_gain_explicit} is exactly \eqref{eq:worst_prior_valid_gain}. Parts (ii) and (iii) refine the picture by characterizing the $\lambda$ that optimizes the certificate.

Before the proof we isolate a stability property of the inner maximizer \eqref{eq:appen:inner_max}, stated for a general linear term; it is used to justify the envelope formula \eqref{eq:appen:radius_matching}, in Remark~\ref{rem:appen:lambda_order} below, and in the proof of Corollary~\ref{cor:fdiv}.

\begin{lemma}[Stability of the inner maximizer]
\label{lem:appen:inner_stability}
For $a\in\RR^m$ and $\eta>0$, let $M_\eta(a):=\argmax{w\in\Delta_m}\{\langle a,w\rangle-\eta D(w)\}$, so that $w_\eta(\Theta)=M_\eta(L(\Theta))$ in \eqref{eq:appen:inner_max}. Suppose Assumption~\ref{ass:w_regularizer_mirror} holds and $D$ is bounded on $\Delta_m$. Then, for all $a,b\in\RR^m$ and all $0<\eta\le\eta'$,
\begin{align}
\bigl\|M_\eta(a)-M_\eta(b)\bigr\|_1
&\ \le\ \frac{1}{\eta}\,\|a-b\|_\infty,
\label{eq:appen:inner_stab_theta}\\
\frac{\eta}{2}\,\bigl\|M_{\eta'}(a)-M_\eta(a)\bigr\|_1^2
&\ \le\ (\eta'-\eta)\,\bigl\{D\bigl(M_\eta(a)\bigr)-D\bigl(M_{\eta'}(a)\bigr)\bigr\}
\ \le\ (\eta'-\eta)\,\sup_{\Delta_m}D .
\label{eq:appen:inner_stab_eta}
\end{align}
In particular, for fixed $\Theta$ the map $\eta\mapsto w_\eta(\Theta)$ is continuous on $(0,\infty)$, and for fixed $\eta$ the map $\Theta\mapsto w_\eta(\Theta)$ is continuous wherever $L(\cdot)$ is. If, moreover, $D\ge0$ and $D(w^{\mathrm{ref}})=0$, then $M_\eta(c\mathbf 1)=w^{\mathrm{ref}}$ for every $c\in\RR$ and $\eta>0$.
\end{lemma}

\begin{proof}
By Assumption~\ref{ass:w_regularizer_mirror}, adding the relative-convexity inequality at the pairs $(u,w)$ and $(w,u)$ and using $B_\phi(u,w)+B_\phi(w,u)\ge\|u-w\|_1^2$ gives the strong monotonicity property
\begin{equation}
\label{eq:appen:D_monotone}
\bigl\langle\nabla D(u)-\nabla D(w),\,u-w\bigr\rangle\ \ge\ \|u-w\|_1^2,
\qquad u,w\in\Delta_m .
\end{equation}
\emph{Proof of \eqref{eq:appen:inner_stab_theta}.} Write $w:=M_\eta(a)$ and $w':=M_\eta(b)$. The first-order optimality conditions for the two maximizations give
\[
\bigl\langle a-\eta\nabla D(w),\,w'-w\bigr\rangle\le0,
\qquad
\bigl\langle b-\eta\nabla D(w'),\,w-w'\bigr\rangle\le0 .
\]
Adding the two inequalities and rearranging,
\[
\eta\,\bigl\langle\nabla D(w)-\nabla D(w'),\,w-w'\bigr\rangle
\ \le\ \bigl\langle a-b,\,w-w'\bigr\rangle
\ \le\ \|a-b\|_\infty\,\|w-w'\|_1 ,
\]
and \eqref{eq:appen:inner_stab_theta} follows from \eqref{eq:appen:D_monotone}.

\emph{Proof of \eqref{eq:appen:inner_stab_eta}.} Write $\psi_\eta(w):=\langle w,a\rangle-\eta D(w)$, $w:=M_\eta(a)$ and $w'':=M_{\eta'}(a)$. Since $\psi_\eta$ is, by Assumption~\ref{ass:w_regularizer_mirror}, $\eta$-strongly concave relative to $\phi$ and maximized at $w$, the relative-concavity inequality at $(w'',w)$ together with the optimality of $w$ gives
\[
\psi_\eta(w'')\ \le\ \psi_\eta(w)-\eta B_\phi(w'',w)\ \le\ \psi_\eta(w)-\frac{\eta}{2}\|w''-w\|_1^2 .
\]
On the other hand, $\psi_\eta=\psi_{\eta'}+(\eta'-\eta)D$ and the optimality of $w''$ for $\psi_{\eta'}$ give
\[
\psi_\eta(w'')=\psi_{\eta'}(w'')+(\eta'-\eta)D(w'')\ \ge\ \psi_{\eta'}(w)+(\eta'-\eta)D(w'')
=\psi_\eta(w)+(\eta'-\eta)\bigl\{D(w'')-D(w)\bigr\}.
\]
Combining the two displays yields \eqref{eq:appen:inner_stab_eta}. The continuity claims are immediate consequences.
Finally, $M_\eta(c\mathbf 1)=\argmax{w\in\Delta_m}\{c-\eta D(w)\}=\argmin{w\in\Delta_m}D(w)=w^{\mathrm{ref}}$, the minimizer being unique because $B_D\ge B_\phi>0$ at distinct arguments.
\end{proof}

\begin{remark}[Order of the certificate-optimal $\lambda$]
\label{rem:appen:lambda_order}
Minimizing the right-hand side of \eqref{eq:prior_gain_explicit} over $\lambda>0$ gives, by the AM--GM inequality,
\[
\inf_{\lambda>0}\Bigl\{\frac{H(\Theta^*_{\mathrm{ref}})^2}{8\lambda}+\lambda\rho\Bigr\}
=H(\Theta^*_{\mathrm{ref}})\sqrt{\rho/2},
\qquad\text{attained at}\qquad
\lambda^*=\frac{H(\Theta^*_{\mathrm{ref}})}{\sqrt{8\rho}}\;\asymp\;\rho^{-1/2}.
\]
This inverse square-root scaling agrees with the radius-matching characterization of part (iii), and is not tied to any particular divergence. Indeed, let $c$ be the mid-range of $L(\Tsl)$, so that $\|L(\Tsl)-c\mathbf 1\|_\infty=H(\Tsl)/2$. Comparing the value of the inner problem at $\wsl=M_\lambda(L(\Tsl))$ with its value at $w^{\mathrm{ref}}$, where $D$ vanishes, gives $\lambda D(\wsl)\le\langle L(\Tsl)-c\mathbf 1,\,\wsl-w^{\mathrm{ref}}\rangle$, and combining H\"older's inequality with the stability bound $\|\wsl-w^{\mathrm{ref}}\|_1=\|M_\lambda(L(\Tsl))-M_\lambda(c\mathbf 1)\|_1\le H(\Tsl)/(2\lambda)$ of Lemma~\ref{lem:appen:inner_stability} sharpens the linear bound \eqref{eq:appen:D_H_bound} to the quadratic bound
\[
D(\wsl)\ \le\ \frac{H(\Tsl)^2}{4\lambda^2};
\]
for the KL regularizer, the soft-max identity $\wsl=\mathcal S(L(\Tsl))$ and Lemma~\ref{lem:appen:softmax_stability} improve the constant to $H(\Tsl)^2/(8\lambda^2)$. An interior certificate-optimal $\lambda$, which satisfies $D(\wsl)=\rho$ by \eqref{eq:appen:radius_matching}, must therefore obey $\lambda\le H(\Tsl)/(2\sqrt\rho)$, matching the order of $\lambda^*$ above. At the statistically calibrated radius $\rho\asymp1/n$, both give $\lambda\asymp\sqrt n$, in agreement with Corollary~\ref{cor:fdiv}.
\end{remark}

\begin{proof}[Proof of Theorem~\ref{thm:prior_benefit_gdro_full}]
\textbf{Part (i).}
Fix $\Theta\in\Xi(r)$. For any $q\in\Delta_m$ with $D(q)\le\rho$,
\[
    \sum_{k=1}^m q_k L_k(\Theta)
    =
    \left\{
        \sum_{k=1}^m q_k L_k(\Theta)-\lambda D(q)
    \right\}
    +
    \lambda D(q)
    \le
    \Phi_\lambda(\Theta)+\lambda\rho .
\]
Taking the supremum over all such $q$ gives the Lagrangian upper bound
\begin{equation}
        \mathcal R_\rho(\Theta)
    \le
    \Phi_\lambda(\Theta)+\lambda\rho
    \qquad\text{for every }\Theta .
    \label{eq:proof_lagrange_upper}
\end{equation}
Applying \eqref{eq:proof_lagrange_upper} at $\Theta=\Tsl$ and subtracting $\mathcal R_\rho(\Theta^*_{\mathrm{G}})$ yields the first inequality in \eqref{eq:prior_gain_explicit}:
\[
    \mathcal R_\rho(\Tsl)-\mathcal R_\rho(\Theta^*_{\mathrm{G}})
    \le
    J_\rho(\lambda)-\mathcal R_\rho(\Theta^*_{\mathrm{G}})
    =
    C_\rho(\lambda).
\]

We next bound $C_\rho(\lambda)$ from above by the more interpretable expression on the right-hand side of \eqref{eq:prior_gain_explicit}. Since $D\ge0$ and $D(w^{\mathrm{ref}})=0$, the point $w^{\mathrm{ref}}$ minimizes $D$ over $\Delta_m$; and since $w^{\mathrm{ref}}\in\operatorname{relint}(\Delta_m)$, for every $w\in\Delta_m$ the point $w^{\mathrm{ref}}+t(w-w^{\mathrm{ref}})$ lies in $\Delta_m$ for all $t$ in a neighborhood of $0$, so the directional derivative of $D$ at $w^{\mathrm{ref}}$ vanishes in every direction $w-w^{\mathrm{ref}}$:
\[
    \bigl\langle \nabla D(w^{\mathrm{ref}}),\,w-w^{\mathrm{ref}}\bigr\rangle=0,
    \qquad
    \forall w\in\Delta_m .
\]
Consequently, by Assumption~\ref{ass:w_regularizer_mirror},
\begin{equation}
    D(w)
    =
    D(w)-D(w^{\mathrm{ref}})
    -
    \bigl\langle \nabla D(w^{\mathrm{ref}}),\,w-w^{\mathrm{ref}}\bigr\rangle
    \ge
    B_\phi(w,w^{\mathrm{ref}})
    \ge
    \frac12\|w-w^{\mathrm{ref}}\|_1^2 .
    \label{eq:D_controls_l1}
\end{equation}
Next, because $w,w^{\mathrm{ref}}\in\Delta_m$, we may center the risk vector at the mid-range value $c_{\mathrm{ref}}:=\tfrac12\{\max_kL_k(\Theta^*_{\mathrm{ref}})+\min_kL_k(\Theta^*_{\mathrm{ref}})\}$ without changing the inner product:
\begin{equation}
\label{eq:appen:midrange}
\begin{aligned}
    (w-w^{\mathrm{ref}})^\top L(\Theta^*_{\mathrm{ref}})
    &=(w-w^{\mathrm{ref}})^\top\bigl\{L(\Theta^*_{\mathrm{ref}})-c_{\mathrm{ref}}\mathbf 1\bigr\}\\
    &\le
    \bigl\|L(\Theta^*_{\mathrm{ref}})-c_{\mathrm{ref}}\mathbf 1\bigr\|_\infty\|w-w^{\mathrm{ref}}\|_1
    =\frac{H(\Theta^*_{\mathrm{ref}})}{2}\,\|w-w^{\mathrm{ref}}\|_1 .
\end{aligned}
\end{equation}
Combining \eqref{eq:D_controls_l1} and \eqref{eq:appen:midrange},
\[
\begin{aligned}
    \Phi_\lambda(\Theta^*_{\mathrm{ref}})
    &=
    \sup_{w\in\Delta_m}
    \left\{
        (w^{\mathrm{ref}})^\top L(\Theta^*_{\mathrm{ref}})
        +
        (w-w^{\mathrm{ref}})^\top L(\Theta^*_{\mathrm{ref}})
        -
        \lambda D(w)
    \right\} \\
    &\le
    (w^{\mathrm{ref}})^\top L(\Theta^*_{\mathrm{ref}})
    +
    \sup_{t\ge0}
    \left\{
        \frac{H(\Theta^*_{\mathrm{ref}})}{2}\,t-\frac{\lambda}{2}\,t^2
    \right\}
    =
    (w^{\mathrm{ref}})^\top L(\Theta^*_{\mathrm{ref}})
    +
    \frac{H(\Theta^*_{\mathrm{ref}})^2}{8\lambda}.
\end{aligned}
\]
Since $\Tsl$ minimizes $\Phi_\lambda$ over $\Xi(r)$ and $\Theta^*_{\mathrm{ref}}\in\Xi(r)$,
\[
    J_\rho(\lambda)
    =
    \Phi_\lambda(\Tsl)+\lambda\rho
    \le
    \Phi_\lambda(\Theta^*_{\mathrm{ref}})+\lambda\rho
    \le
    (w^{\mathrm{ref}})^\top L(\Theta^*_{\mathrm{ref}})
    +
    \frac{H(\Theta^*_{\mathrm{ref}})^2}{8\lambda}
    +
    \lambda\rho .
\]
Moreover, $w^{\mathrm{ref}}$ is feasible for the supremum defining $\mathcal R_\rho$ because $D(w^{\mathrm{ref}})=0\le\rho$, so
$\mathcal R_\rho(\Theta^*_{\mathrm{G}})\ge(w^{\mathrm{ref}})^\top L(\Theta^*_{\mathrm{G}})$. Therefore
\[
    C_\rho(\lambda)
    =
    J_\rho(\lambda)-\mathcal R_\rho(\Theta^*_{\mathrm{G}})
    \le
    (w^{\mathrm{ref}})^\top\bigl\{L(\Theta^*_{\mathrm{ref}})-L(\Theta^*_{\mathrm{G}})\bigr\}
    +
    \frac{H(\Theta^*_{\mathrm{ref}})^2}{8\lambda}
    +
    \lambda\rho
    =
    -\Delta_{\mathrm{ref}}
    +
    \frac{H(\Theta^*_{\mathrm{ref}})^2}{8\lambda}
    +
    \lambda\rho ,
\]
which proves \eqref{eq:prior_gain_explicit}. The strict-improvement statement follows immediately.

\medskip\noindent\textbf{Part (ii).}
For fixed $\Theta$, the problem defining $\mathcal R_\rho(\Theta)$ maximizes a linear objective over the compact convex set $\{q\in\Delta_m:D(q)\le\rho\}$, whose Slater point is $w^{\mathrm{ref}}$ since $D(w^{\mathrm{ref}})=0<\rho$. Strong Lagrange duality for convex programs \citep[e.g.,][Section~7.1]{shapiro2021lectures} therefore gives
\[
    \mathcal R_\rho(\Theta)
    =
    \inf_{\eta\ge0}
    \left\{
        \sup_{w\in\Delta_m}
        \Bigl[
            \sum_{k=1}^m w_kL_k(\Theta)-\eta D(w)
        \Bigr]
        +
        \eta\rho
    \right\}
    =
    \inf_{\eta\ge0}
    \{\Phi_\eta(\Theta)+\eta\rho\}.
\]
Taking the infimum over $\Theta\in\Xi(r)$ and exchanging the order of the two infima,
\[
    V_\rho
    =
    \inf_{\Theta\in\Xi(r)}\inf_{\eta\ge0}
    \{\Phi_\eta(\Theta)+\eta\rho\}
    =
    \inf_{\eta\ge0}
    \inf_{\Theta\in\Xi(r)}
    \{\Phi_\eta(\Theta)+\eta\rho\}
    =
    \inf_{\eta\ge0}J_\rho(\eta).
\]
Subtracting the constant $\mathcal R_\rho(\Theta^*_{\mathrm{G}})$ yields the second identity in \eqref{eq:dual_radius_oracle}, and $V_\rho\le\mathcal R_\rho(\Theta^*_{\mathrm{G}})$ because $\Theta^*_{\mathrm{G}}\in\Xi(r)$ is feasible for the infimum defining $V_\rho$.

\medskip\noindent\textbf{Part (iii).}
Fix $0<\lambda<\lambda'$. We establish the two-sided envelope inequality
\begin{equation}
\label{eq:appen:envelope_sandwich}
(\lambda'-\lambda)\,\bigl\{\rho-D\bigl(w_\lambda(\Theta^*_{\lambda'})\bigr)\bigr\}
\ \le\ J_\rho(\lambda')-J_\rho(\lambda)
\ \le\ (\lambda'-\lambda)\,\bigl\{\rho-D\bigl(w_{\lambda'}(\Tsl)\bigr)\bigr\}.
\end{equation}
For the upper bound, evaluate $\Phi_{\lambda'}$ at its own maximizer: for any $\Theta$,
\[
\Phi_{\lambda'}(\Theta)
=\bigl\langle w_{\lambda'}(\Theta),L(\Theta)\bigr\rangle-\lambda D\bigl(w_{\lambda'}(\Theta)\bigr)-(\lambda'-\lambda)D\bigl(w_{\lambda'}(\Theta)\bigr)
\le\Phi_\lambda(\Theta)-(\lambda'-\lambda)D\bigl(w_{\lambda'}(\Theta)\bigr).
\]
Taking $\Theta=\Tsl$ and using $J_\rho(\lambda')\le\Phi_{\lambda'}(\Tsl)+\lambda'\rho$ together with $\Phi_\lambda(\Tsl)=J_\rho(\lambda)-\lambda\rho$ gives the upper bound in \eqref{eq:appen:envelope_sandwich}. Symmetrically, for any $\Theta$,
\[
\Phi_\lambda(\Theta)
=\bigl\langle w_{\lambda}(\Theta),L(\Theta)\bigr\rangle-\lambda' D\bigl(w_{\lambda}(\Theta)\bigr)+(\lambda'-\lambda)D\bigl(w_{\lambda}(\Theta)\bigr)
\le\Phi_{\lambda'}(\Theta)+(\lambda'-\lambda)D\bigl(w_{\lambda}(\Theta)\bigr);
\]
taking $\Theta=\Theta^*_{\lambda'}$ and using $J_\rho(\lambda)\le\Phi_\lambda(\Theta^*_{\lambda'})+\lambda\rho$ and $\Phi_{\lambda'}(\Theta^*_{\lambda'})=J_\rho(\lambda')-\lambda'\rho$ gives the lower bound.

Now let $\lambda'\downarrow\lambda$ (the case $\lambda'\uparrow\lambda$ is symmetric). By Lemma~\ref{lem:appen:inner_stability}, $w_{\lambda'}(\Tsl)\to w_\lambda(\Tsl)=\wsl$ by \eqref{eq:appen:inner_stab_eta}; and by the assumed convergence $\Theta^*_{\lambda'}\to\Tsl$, the continuity of $L(\cdot)$, and \eqref{eq:appen:inner_stab_theta}, also $w_{\lambda}(\Theta^*_{\lambda'})\to \wsl$. Since $D$ is continuous, dividing \eqref{eq:appen:envelope_sandwich} by $\lambda'-\lambda$ and letting $\lambda'\to\lambda$ shows that $J_\rho$ is differentiable at $\lambda$ with $J_\rho'(\lambda)=\rho-D(\wsl)$; the same holds for $C_\rho$, which differs from $J_\rho$ by a constant. This proves \eqref{eq:appen:radius_matching}, and at an interior minimizer of $C_\rho$ the derivative vanishes, forcing $D(\wsl)=\rho$.

It remains to prove \eqref{eq:appen:D_H_bound}. By optimality of $\wsl$ in the inner problem and $D(w^{\mathrm{ref}})=0$,
\[
    \sum_{k=1}^m w^*_{\lambda,k}L_k(\Tsl)
    -
    \lambda D(\wsl)
    \ge
    \sum_{k=1}^m w_k^{\mathrm{ref}}L_k(\Tsl)-\lambda D(w^{\mathrm{ref}})
    =
    \sum_{k=1}^m w_k^{\mathrm{ref}}L_k(\Tsl) .
\]
Rearranging, and using that the difference between the expectations of the vector $L(\Tsl)$ under two probability weights is at most its range,
\[
    \lambda D(\wsl)
    \le
    \sum_{k=1}^m
    (w^*_{\lambda,k}-w_k^{\mathrm{ref}})L_k(\Tsl)
    \le
    H(\Tsl).
\]
Dividing by $\lambda$ proves \eqref{eq:appen:D_H_bound}. Finally, if $\lambda\rho>H(\Tsl)$, then $D(\wsl)\le H(\Tsl)/\lambda<\rho$, so \eqref{eq:appen:radius_matching} gives $C_\rho'(\lambda)=\rho-D(\wsl)>0$: increasing $\lambda$ locally increases the certified worst-case risk and moves the penalized problem away from the radius-$\rho$ target. This completes the proof.
\end{proof}

\section{Details of the simulation studies}
\label{appen:sim_details}

This appendix records the full specification of the experiments of
Section~\ref{sec:sim}. Table~\ref{tab:sim_design} collects every design parameter;
the text below gives the generating constants, the population target, the
initialization and tuning, and the three evaluation protocols.

\subsection{Data-generating mechanism}
\label{appen:subsec:sim_dgp}

All experiments use \eqref{eq:sim_dgp}. The shared latent center is drawn once, as
$U_0\in\RR^{d_1\times r^*}$ and $V_0\in\RR^{d_2\times r^*}$ with independent
$\mathcal N(0,0.3^2)$ entries; the source-specific perturbations
$\Delta^U_k,\Delta^V_k$ have independent $\mathcal N(0,0.15^2)$ entries. Every source
is therefore driven by the same $r^*$-dimensional latent representation while
differing from the others. The resulting heterogeneity is substantial: in the design
of Section~\ref{subsec:sim_pred} the relative separation
$\|\Theta_j-\Theta_k\|_F/\|\Theta_k\|_F$ lies between $0.78$ and $0.86$, so that
pooling across sources is informative yet no single $\Theta$ is simultaneously optimal
for all of them. The covariate means $\mu_k$, covariances $\Sigma_k$ and noise levels
$\sigma_k$ are listed in Table~\ref{tab:sim_design}; in the sweep over $m$, the first
three sources are fixed templates and each further source draws its own
$(\mu_k,\Sigma_k,\sigma_k)$ from the stated ranges, together with its own $\Theta_k$
from the same perturbation model.

Two conventions for the source loss appear in Section~\ref{sec:sim}. Writing
\begin{equation}
\label{eq:sim_loss}
    \Lkh(\Theta)
    =\frac{1}{2n_kd_2}
    \bigl\{\|Y\uk-X\uk\Theta\|_F^2-b_k\bigr\},
    \qquad
    b_k\in\{0,\ \|Y\uk\|_F^2\},
\end{equation}
the offset $b_k$ is a source-specific constant that leaves $\nabla_\Theta\Lkh$
unchanged but shifts the levels of the loss vector, and hence moves both $\wsl$
and $\Tsl$. The population illustration and the predictive comparisons take
$b_k=\|Y\uk\|_F^2$, giving \eqref{eq:sim_pop_risk_main} in the main text, in which the
noise variance cancels; the two estimation-error sweeps take $b_k=0$, so that $\Lkh$
is the mean squared error and the population risk carries the additive constant
$d_2\sigma_k^2$ of \eqref{eq:sim_pop_risk} below. Each experiment computes its own
population target under its own convention, so comparisons within an experiment are
internally consistent.

\subsection{Population target}
\label{appen:subsec:sim_pop}

Because \eqref{eq:sim_dgp} is Gaussian, the population risks are available in closed
form. With $A_k=\Sigma_k+\mu_k\mu_k^\top$,
\begin{equation}
\label{eq:sim_pop_risk}
    L_k(\Theta)
    =\frac{1}{2d_2}
    \Bigl\{
    \operatorname{tr}\bigl[(\Theta_k-\Theta)^\top A_k(\Theta_k-\Theta)\bigr]+c_k
    \Bigr\},
    \qquad
    c_k=
    \begin{cases}
    d_2\sigma_k^2, & b_k=0,\\[2pt]
    -\operatorname{tr}\bigl[\Theta_k^\top A_k\Theta_k\bigr], & b_k=\|Y\uk\|_F^2 .
    \end{cases}
\end{equation}
We obtain the population saddle point $(\Tsl,\wsl)$ by running the BGDMA
updates on the exact gradients of \eqref{eq:sim_pop_risk}, with a small constant step
size and a tight stopping tolerance, and use the resulting $\Tsl$ as the estimation
target in Section~\ref{subsec:sim_rate}. For Figure~\ref{fig:sim_population}(a) the
reference saddle is computed once at $\lambda=1$, to tolerance $10^{-11}$, and shared
across all four step sizes, so that the four trajectories are directed at a common
limit and their slopes are comparable. For Figure~\ref{fig:sim_population}(b) the
population problem is re-solved at each of the sixteen values of $\lambda$.

\subsection{Initialization, step sizes and benchmark tuning}
\label{appen:subsec:sim_tuning}

Following Section~\ref{subsec:init}, BGDMA is warm-started at the pooled ridge
estimate $\Theta^0=(X^\top X+0.01I_{d_1})^{-1}X^\top Y$, followed by the balanced
rank-$r$ singular value decomposition $U^0=\tilde U\tilde\Sigma^{1/2}$,
$V^0=\tilde V\tilde\Sigma^{1/2}$, together with $w^0=w^{\mathrm{ref}}$. This is a
cheap surrogate for the nuclear-norm warm start analysed in
Theorem~\ref{thm:softmax_init}: it costs one $d_1\times d_1$ solve rather than a
sequence of singular-value thresholdings, and it sufficed to enter the local basin in
every configuration we ran. Step sizes are held constant within an experiment at the
values in Table~\ref{tab:sim_design}, with no decay schedule, line search or restart,
so that the reported behaviour is that of the plain algorithm analysed in
Theorem~\ref{thm:PILLAR_convergence}.

The benchmarks are tuned as follows. \textsc{Per-task} and \textsc{SVD} use ridge
parameter $0.01$; \textsc{L21} and \textsc{Nuclear} are solved by accelerated
proximal gradient descent with penalties $0.01$ and $0.08$ respectively and step size
$1/L$ for the pooled smooth part; \textsc{BGD} and \textsc{Maximin} use gradient
descent, the latter with the same exponentiated-gradient dual update as BGDMA at
$\lambda=0$, each with its own step size and iteration budget chosen so that its
iterates stabilise. All estimators receive the same training data and, where
applicable, the same rank and the same initialization.

\subsection{Evaluation protocols}
\label{appen:subsec:sim_eval}

\emph{In-distribution.} Test data of size $n_{\mathrm{test}}=2000$ are generated from
each of the $m$ training sources through \eqref{eq:sim_dgp}, and the reward
\eqref{eq:sim_reward} is evaluated per source.

\emph{Out-of-distribution.} We generate $100$ independent environments. Each has its
own signal matrix, drawn from the perturbation model of
Appendix~\ref{appen:subsec:sim_dgp} about a center itself displaced from $(U_0,V_0)$,
an isotropic covariance $\sigma_xI_{d_1}$ with $\sigma_x\sim\mathcal U(0.8,1.5)$, a
mean vector with independent $\mathcal U(0,0.5)$ coordinates, and a noise level
$\sigma\sim\mathcal U(1.5,2.5)$; each estimator is evaluated on
$n_{\mathrm{test}}=2000$ observations per environment. Because the covariate mean and
the noise range lie outside the training configuration, none of these environments is
a mixture of the observed sources.

\emph{Continuum of mixtures.} We use 20 logarithmically spaced Dirichlet
concentrations $\alpha\in[0.1,100]$. At each level we draw $500$ weight vectors
$v\sim\mathrm{Dirichlet}(\alpha\mathbf 1_m)$, generate $n_{\mathrm{test}}=2000$
observations from $\QQ_v=\summ v_k\PP\uk$ for each, and record $\min_v R_v(\Theta)$
over the $500$ draws. The horizontal axis of Figure~\ref{fig:sim_continuum} is the
grid index rescaled to $[0,1]$, equivalently the rescaled $\log\alpha$. The entire
pipeline, including data generation and refitting of every estimator, is repeated over
$20$ independent training draws and averaged.

\begin{table}[htbp]
\centering
\footnotesize
\setlength{\tabcolsep}{3pt}
\caption{Design of the simulation experiments. All use \eqref{eq:sim_dgp} and the
loss \eqref{eq:sim_loss}; $\mathbf 1\in\RR^{d_1}$ and $\mathcal U$ denotes the uniform
distribution. In the sweep over $m$ the first three sources are the fixed templates
listed and each further source draws its own parameters from the stated ranges. Paired
entries in the last column refer to Figure~\ref{fig:sim_pred} and
Figure~\ref{fig:sim_continuum} respectively.}
\label{tab:sim_design}
\begin{tabular}{@{}llll l@{}}
\toprule
 & \textbf{Population} & \textbf{Sweep in $n$} & \textbf{Sweep in $m$} & \textbf{Prediction} \\
\midrule
Reported in
  & Fig.~\ref{fig:sim_population}
  & Fig.~\ref{fig:sim_rate}(a)
  & Fig.~\ref{fig:sim_rate}(b)
  & Figs.~\ref{fig:sim_pred},~\ref{fig:sim_continuum} \\
$(d_1,d_2)$ & $(100,20)$ & $(100,20)$ & $(100,20)$ & $(200,40)$ \\
$r^*$ & $5$ & $5$ & $5$ & $5$ \\
$m$ & $3$ & $3$ & $4,6,11,20,34,59$ & $3$ \\
$n_k$ & population risks & $600$--$3600$ & $400$ & $150$; $200$ \\
$(\sigma_k)$
  & $(1,0.5,0.8)$ & $(1,0.5,0.75)$
  & $(3,2,1)$, then $\mathcal U(0.8,1.2)$ & $(2,1,1.5)$ \\
$\mu_k$
  & $(0,0.3,-0.5)\mathbf 1$ & $(0,0.3,-0.5)\mathbf 1$
  & $(0,0.5,-0.5)\mathbf 1$, then $\mathcal U(-0.3,0.3)\mathbf 1$
  & $(0,0.5,-0.5)\mathbf 1$ \\
$\Sigma_k/I_{d_1}$
  & $(1,0.8,1.2)$ & $(1,0.75,1.2)$
  & $(0.03,0.75,1)$, then $\mathcal U(0.8,1.2)$ & $(1,0.75,1.2)$ \\
Working rank $r$ & $5$ & $3,5,10$ & $3,5,10$ & $5$ \\
$\lambda$ & $1$; grid on $[0.1,100]$ & $0.1,0.5,1$ & $0.1,0.5,1$ & $0,0.1,0.5,1,5,10$ \\
$\eta_\Theta=\eta_w$ & $0.005$--$0.04$; $0.01$ & $0.008$ & $0.008$ & $0.02$ \\
Iterations & $\le 3\!\times\!10^{4}$ & $3000$ & $3000$ & $3000$ \\
Offset $b_k$ & $\|Y\uk\|_F^2$ & $0$ & $0$ & $\|Y\uk\|_F^2$ \\
\bottomrule
\end{tabular}
\end{table}

\end{document}